\documentclass[
a4paper,reqno]{amsart}
\usepackage[T1]{fontenc}
\usepackage[english]{babel}
\usepackage{amssymb,bbm,enumerate}
\usepackage{color}
\usepackage[linktocpage=true,colorlinks=true, linkcolor=blue, citecolor=red, urlcolor=green]{hyperref}

\usepackage{float}
\restylefloat{table}

\usepackage{tikz}

\renewcommand{\phi}{\varphi}

\newcommand{\mc}[1]{\mathcal{#1}}
\newcommand{\mf}[1]{\mathfrak{#1}}
\newcommand{\mb}[1]{\mathbb{#1}}

\newcommand{\id}{\mathbbm{1}}

\newcommand{\tint}{{\textstyle\int}}

\DeclareMathOperator{\Mat}{Mat}

\DeclareMathOperator{\diag}{diag}
\DeclareMathOperator{\tr}{tr}

\DeclareMathOperator{\ad}{ad}
\DeclareMathOperator{\im}{Im}

\DeclareMathOperator{\Span}{Span}

\DeclareMathOperator{\Res}{Res}

\theoremstyle{plain}
\newtheorem{theorem}{Theorem}[section]
\newtheorem{lemma}[theorem]{Lemma}
\newtheorem{proposition}[theorem]{Proposition}
\newtheorem{corollary}[theorem]{Corollary}

\theoremstyle{definition}
\newtheorem{definition}[theorem]{Definition}
\newtheorem{example}[theorem]{Example}

\theoremstyle{remark}
\newtheorem{remark}[theorem]{Remark}

\numberwithin{equation}{section}

\definecolor{light}{gray}{.9}

\begin{document}

\title[Classical $\mc W$-algebras for $\mf{gl}_N$]{Classical affine $\mc W$-algebras for $\mf{gl}_N$ 
and associated integrable Hamiltonian hierarchies}

\author{Alberto De Sole}
\address{Dipartimento di Matematica, Sapienza Universit\`a di Roma,
P.le Aldo Moro 2, 00185 Rome, Italy}
\email{desole@mat.uniroma1.it}
\urladdr{www1.mat.uniroma1.it/\$$\sim$\$desole}

\author{Victor G. Kac}
\address{Dept of Mathematics, MIT,
77 Massachusetts Avenue, Cambridge, MA 02139, USA}
\email{kac@math.mit.edu}

\author{Daniele Valeri}
\address{Yau Mathematical Sciences Center, Tsinghua University, 100084 Beijing, China}
\email{daniele@math.tsinghua.edu.cn}



\begin{abstract}
We apply the new method for constructing integrable Hamiltonian hierarchies
of Lax type equations developed in our previous paper,
to show that all $\mc W$-algebras $\mc W(\mf{gl}_N,f)$
carry such a hierarchy.
As an application, we show that all vector constrained KP hierarchies
and their matrix generalizations are obtained from these hierarchies
by Dirac reduction,
which provides the former with a bi-Poisson structure.
\end{abstract}


\maketitle

\tableofcontents

\section{Introduction}\label{sec:1}

In their seminal paper \cite{DS85}
Drinfeld and Sokolov constructed the $\mc W$-algebra $\mc W(\mf g,f)$
for each reductive Lie algebra $\mf g$ and its principal nilpotent element $f$
and discovered the associated integrable hierarchy of bi-Hamiltonian PDE.
Furthermore, they showed that for $\mf g=\mf{gl}_N$
this $\mc W$-algebra is isomorphic to the Adler-Gelfand-Dickey algebra \cite{GD78,Adl79},
and that the associated integrable hierarchy is the $N$-th KdV hierarchy
constructed by Gelfand and Dickey \cite{GD76}
using fractional powers of differential operators.
In the proof of this isomorphism Drinfeld and Sokolov used quasideterminants,
a few years before Gelfand and Retakh began their systematic study in the early 90's,
see \cite{GGRW05} for references.

The classical affine $\mc W$-algebras $\mc W(\mf g,f)$,
for arbitrary reductive Lie algebras $\mf g$ and their non-zero nilpotent elements $f$,
have been studied both in physics and mathematics literature,
see \cite{BFOFW90,DSK06,DSKV13} and references there.
In particular, it has been understood that the adequate setup for the theory
of $\mc W$-algebras is the language of $\lambda$-brackets
in the framework of Poisson vertex algebras (PVA).
This approach has led to an explicit formula for the $\lambda$-brackets
in all classical affine $\mc W$-algebras \cite{DSKV16}.

However, the problem whether any $\mc W$-algebra $\mc W(\mf g,f)$ carries an integrable hierarchy
of Hamiltonian PDE has been solved so far only under a very restrictive assumption on 
the nilpotent element $f$ by adopting the Drinfeld-Sokolov method
\cite{dGHM92,FHM93,BdGHM93,DF95,FGMS95,FGMS96,DSKV13,DSKV14a}.

In fact, one has a pencil of compatible PVA structures on the differential algebra $\mc W(\mf g,f)$,
depending on an element $S$ of $\mf g$,
and we shall denote by $\mc W_\epsilon(\mf g,f,S)$, $\epsilon\in\mb F$,
the corresponding family of PVAs.
Then, the related problem is whether the family $\mc W_\epsilon(\mf g,f,S)$, $\epsilon\in\mb F$,
carries an integrable bi-Hamiltonian hierarchy satisfying the so called Lenard-Magri scheme of integrability \cite{Mag78}.

In the present paper we solve these problems for $\mf g=\mf{gl}_N$
and its arbitrary nilpotent element $f$,
by making use of the scheme of integrability developed in our recent paper \cite{DSKVnew}.
The main ingredients of this scheme are the notion of an Adler type matrix pseudodifferential operator
with respect to a $\lambda$-bracket on a differential algebra,
introduced in \cite{DSKV15},
and the notion of a generalized quasideterminant.

\subsection{PVA's and Adler type pseudodifferential operators}\label{sec:1.1}

Recall that a $\lambda$-\emph{bracket} $\{\cdot\,_\lambda\,\cdot\}$ on a differential algebra $(\mc V,\partial)$
is a bilinear map $\mc V\times\mc V\to\mc V[\lambda]$,
satisfying the following axioms \cite{DSK06,BDSK09} ($a,b,c\in\mc V$):
\begin{enumerate}[(i)]
\item
sesquilinearity:
$\{\partial a_\lambda b\}=-\lambda\{a_\lambda b\}$,
$\{a_\lambda\partial b\}=(\lambda+\partial)\{a_\lambda b\}$;
\item
Leibniz rules:
$\{a_\lambda bc\}=\{a_\lambda b\}c+\{a_\lambda c\}b$,
$\{ab_\lambda c\}=\{a_{\lambda+\partial} c\}_\to b
+\{b_{\lambda+\partial} c\}_\to a$.
\end{enumerate}
If the $\lambda$-bracket $\{\cdot\,_\lambda\,\cdot\}$ satisfies, in addition,
skewsymmetry and Jacobi identity (see Section \ref{sec:2.1})
then $\mc V$ is called a Poisson vertex algebra (PVA) and the $\lambda$-bracket is called
a PVA $\lambda$-bracket.

In the present paper we shall consider pencils $\lambda$-brackets
$\{a_\lambda b\}_\epsilon=\{a_\lambda b\}_0+\epsilon\{a_\lambda b\}_1$, $\epsilon\in\mb F$,
on a differential algebra $\mc V$.
In such case, the $\lambda$-bracket $\{\cdot\,_\lambda\,\cdot\}_0$ (resp. $\{\cdot\,_\lambda\,\cdot\}_1$)
will be called the $0$-th (resp. $1$-st) $\lambda$-bracket.
If they are PVA $\lambda$-brackets,
they are called the $0$-th and the $1$-st PVA (or Poisson) structures on $\mc V$.
(Unfortunately traditionally they are called the $2$-nd and the $1$-st Poisson structures respectively.)

An $M\times N$ matrix pseudodifferential operator $A(\partial)=\big(A_{ij}(\partial)\big)$,
where $A_{ij}(\partial)\in\mc V((\partial^{-1}))$,
is called an operator of \emph{Adler type} with respect to a $\lambda$-bracket $\{\cdot\,_\lambda\,\cdot\}$ on $\mc V$
if for every $(i,j),(h,k)\in\{1,\dots,M\}\times\{1,\dots,N\}$ we have \cite{DSKV15,DSKVnew}:
\begin{equation}\label{eq:adler-intro}
\begin{split}
\{A_{ij}(z)_\lambda A_{hk}(w)\}
& = A_{hj}(w+\lambda+\partial)(z\!-\!w\!-\!\lambda\!-\!\partial)^{-1}(A_{ik})^*(\lambda-z)
\\
& - A_{hj}(z)(z\!-\!w\!-\!\lambda\!-\!\partial)^{-1}A_{ik}(w)
\,.
\end{split}
\end{equation}
We shall also say that $A(\partial)$ is of \emph{bi-Adler type} with respect 
to a pencil of $\lambda$-brackets $\{\cdot\,_\lambda\,\cdot\}_\epsilon$, $\epsilon\in\mb F$,
if $A(\partial)+\epsilon\id_N$ is of Adler type
with respect to $\{\cdot\,_\lambda\,\cdot\}_\epsilon$ for every $\epsilon\in\mb F$.

Given an $N\times N$ invertible matrix $A$ over a unital associative (not necessarily commutative)
algebra $R$, its $(i,j)$-\emph{quasideterminant}
is defined as the inverse (if it exists) of the $(j,i)$ entry of $A^{-1}$ \cite{GGRW05}.
This can be generalized by replacing the $(j,i)$ entry of $A^{-1}$ by an $M\times M$ square submatrix.
More generally, given $I\in\Mat_{N\times M}R$ and $J\in\Mat_{M\times N}R$,
for some $M\leq N$,
the $(I,J)$-quasideterminant of $A$ is defined by \cite{DSKVnew}
\begin{equation}\label{eq:intro-1.2}
|A|_{IJ}:=(JA^{-1}I)^{-1}
\,\in\Mat_{M\times M}R
\,,
\end{equation}
provided that $JA^{-1}I$ is an invertible matrix.

The basic family of Adler type $N\times N$ matrix differential operators
is constructed as follows \cite{DSKVnew}.
Let $\mc V(\mf{gl}_N)$ be the algebra of differential polynomials in the indeterminates $q_{ij}$,
$1\leq i,j\leq N$, 
let $Q=\big(q_{ji}\big)_{i,j=1}^N$,
let $S\in\Mat_{N\times N}\mb F$,
and let $\epsilon\in\mb F$ be a parameter.
Then the operator
\begin{equation}\label{eq:intro-1.3}
A_{\epsilon S}(\partial)= \id_N\partial+Q+\epsilon S
\end{equation}
is of Adler type with respect to the $\lambda$-bracket
\begin{equation}\label{eq:intro-1.4}
\{a_\lambda b\}_\epsilon = [a,b]+\tr(ab)\lambda+\epsilon\tr(S[a,b])
\,\,,\,\,\,\, a,b\in\mf{gl}_N\,.
\end{equation}
Here $\mf{gl}_N$ is identified with a subspace of $\mc V(\mf{gl}_N)$ via $E_{ij}\mapsto q_{ij}$.
Formula \eqref{eq:intro-1.4} endows $\mc V(\mf{gl}_N)$ with a pencil of PVA $\lambda$-brackets
and we denote the corresponding family of PVAs by $\mc V_\epsilon(\mf{gl}_N,S)$, $\epsilon\in\mb F$.

The basic property of an $N\times N$ matrix pseudodifferential operator $A(\partial)$
of Adler type with respect to a $\lambda$-bracket $\{\cdot\,_\lambda\,\cdot\}$ on $\mc V$,
is that any of its generalized quasideterminants is again of Adler type with respect 
to the same $\lambda$-bracket \cite{DSKVnew}.
Moreover, let $S\in\Mat_{N\times N}\mb F$ and assume that, for every $\epsilon\in\mb F$, 
$A(\partial)+\epsilon S$ is an operator of Adler type
with respect to a member $\{\cdot\,_\lambda\,\cdot\}_\epsilon$ of a pencil
of $\lambda$-brackets.
Let also $S=IJ$ be a factorization of $S$ with $I\in\Mat_{N\times r}\mb F$ and $J\in\Mat_{r\times N}\mb F$,
where $r$ is the rank of $S$.
Then, the generalized quasideterminant $|A(\partial)|_{IJ}$ is an operator of bi-Adler type
with respect to the same pencil of $\lambda$-brackets $\{\cdot\,_\lambda\,\cdot\}_\epsilon$, $\epsilon\in\mb F$.

The importance of a square matrix Adler type pseudodifferential operator $A(\partial)$
comes from the fact that it provides a hierarchy of compatible Lax equations
\begin{equation}\label{eq:intro-1.5}
\frac{dA(\partial)}{dt_{n,B}}
=
[(B(\partial)^n)_+,A(\partial)]
\,,
\end{equation}
where $B(\partial)$ is a root of $A(\partial)$,
and $n\in\mb Z_+$.
Moreover, this hierarchy admits the following conserved densities in involution:
\begin{equation}\label{eq:intro-1.6}
h_{n',B'}
=
\Res_\partial 
\tr B'(\partial)^{n'}
\,\,,\,\,\,\,
n'\in\mb Z_+
\,,\,\,
B' \,\text{ a root of } A
\,,
\end{equation}
see \cite{DSKV15,DSKVnew}.
Moreover, for a bi-Adler type operator $A(\partial)$
with respect to a pencil of PVA structures,
the hierarchy \eqref{eq:intro-1.5} consists of bi-Hamiltonian equations
(over the differential subalgebra of $\mc V$ generated by the entries of the coefficients of $A(\partial)$),
and the densities \eqref{eq:intro-1.6} satisfy the (generalized) Lenard-Magri recurrence relation.

\subsection{Classical affine $\mc W$-algebras}\label{sec:1.2}

In order to construct an integrable hierarchy of Hamiltonian equations 
for the pencil of affine $\mc W$-algebras $\mc W_\epsilon(\mf{gl}_N,f,S)$, $\epsilon\in\mb F$,
we shall construct an appropriate generalized quasideterminant
of the $N\times N$ matrix $A_{\epsilon S}(\partial)$ defined by \eqref{eq:intro-1.3}.

Recall the construction of the classical affine $\mc W$-algebra
$\mc W_\epsilon(\mf g,f,S)$, $\epsilon\in\mb F$, from \cite{DSKV13},
for the Lie algebra $\mf g=\mf{gl}_N$,
a nilpotent element $f\in\mf g$,
and a certain element $S\in\mf g$ specified below.
(The construction for an arbitrary reductive Lie algebra $\mf g$ is similar.)
The element $f$ can be embedded in an $\mf{sl}_2$-triple $\{f,2x,e\}$,
and we have the corresponding $\ad x$-eigenspace decomposition
\begin{equation}\label{eq:intro-1.7}
\mf g=\bigoplus_{k\in\frac{1}{2}\mb Z}\mf g_{k}
\,\,,\,\,\text{ where }\,\,
\mf g_k=\big\{a\in\mf g\,\big|\,[x,a]=ka\big\}
\,.
\end{equation}
Let $p=(p_1,\dots,p_r)$, with $p_1\geq p_2\geq\dots\geq p_r>0$,
be the partition of $N$ corresponding to $f$.
Then $d=p_1-1$ is the maximal eigenvalue of $\ad x$,
and $r_1$, the multiplicity of $p_1$ in $p$, is the dimension of $\mf g_d$.
Let $S$ be a non-zero element of $\mf g_d$.
For a subspace $\mf a\subset\mf g$,
we denote by $\mc V(\mf a)$ the algebra of differential polynomials over $\mf a$.
Denote by $\rho:\,\mc V(\mf g)\to\mc V(\mf g_{\leq\frac12})$
the differential algebra homomorphism defined by
\begin{equation}\label{eq:intro-1.8}
\rho(a)=\tr(fa)+\pi_{\leq\frac12}a
\,\,,\,\,\,\, a\in\mf g\,,
\end{equation}
where $\pi_{\leq\frac12}$ denotes the projection on $\mf g_{\leq\frac12}$
with respect to the decomposition \eqref{eq:intro-1.7}.
The \emph{classical affine} $\mc W$-\emph{algebra} $\mc W_\epsilon(\mf g,f,S)$
is the differential algebra
\begin{equation}\label{eq:intro-1.9}
\mc W=\mc W(\mf g,f)
=\big\{
w\in\mc V(\mf g_{\leq\frac12})\,\big|\,\rho\{a_\lambda w\}_{\epsilon}=0
\,\text{ for all }\, a\in\mf g_{\geq\frac12} \big\}
\,,
\end{equation}
where $\{a_\lambda w\}_{\epsilon}$ is defined by \eqref{eq:intro-1.4},
endowed with the $\lambda$-bracket 
\begin{equation}\label{eq:intro-1.10}
\{v_\lambda w\}_{\epsilon}^{\mc W}
=
\rho\{v_\lambda w\}_{\epsilon}
\,.
\end{equation}
This pencil of $\lambda$-brackets provides the differential algebra $\mc W$ with a bi-PVA structure.

In order to perform concrete computations,
we choose a convenient slice $U$ to the adjoint orbit of $f$ in $\mf g=\mf{gl}_N$
(different from the Slodowy slice),
so that 
$\mf g
=[f,\mf g]\oplus U
=\mf g^f\oplus U^\perp$.
Hence, we have the decomposition in a direct sum of subspaces
$$
\mc V(\mf g_{\leq\frac12})=\mc V(\mf g^f)\oplus\langle U^\perp\rangle
\,,
$$
where $\langle U^\perp\rangle$ is the differential algebra ideal generated by $U^\perp$.
The corresponding projection of $\mc V(\mf g_{\leq\frac12})$ on $\mc V(\mf g^f)$
induces a differential algebra homomorphism
$\pi:\,\mc W\to\mc V(\mf g^f)$,
and the key fact is that this is an isomorphism \cite{DSKV13,DSKV16} (see Theorem \ref{thm:structure-W}).
Thus, for each element $q\in\mf g^f$ we have a unique element $w(q)\in\mc W$,
and these elements are differential generators of $\mc W$.
The explicit construction of a bi-Adler type operator indicated below
allows us to construct explicitly these generators
and to compute their $\lambda$-brackets.

\subsection{Adler type operators for the $\mc W$-algebras
and associated integrable systems}\label{sec:1.3}

To construct a bi-Adler type operator for the pencil of PVAs $\mc W_\epsilon(\mf g,f,S)$, $\epsilon\in\mb F$,
we first construct
an $r_1\times r_1$ matrix pseudodifferential operator 
$L_1(\partial)$ with entries of coefficients in $\mc W$,
which is of bi-Adler type for the bi-PVA structure of the family $\mc W_\epsilon(\mf g,f,S_1)$, $\epsilon\in\mb F$,
for the matrix $S_1:=I_1J_1$, where
$I_1:\,\mf g_d\hookrightarrow\mf g$ is the inclusion map,
and $J_1:\,\mf g\twoheadrightarrow\mf g_d$ is the projection
with respect to the decomposition \eqref{eq:intro-1.7}.
It is given by the following generalized quasideterminant
\begin{equation}\label{eq:intro-1.11}
L_1(\partial)
=
|\id_N\partial+\rho(Q)|_{I_1J_1}
\,\,,\,\,\text{ where }\,\,
Q=\big(q_{ji}\big)_{i,j=1}^N
\,.
\end{equation}
In our Theorem \ref{thm:L1} we prove that this generalized quasideterminant exists,
and in our most difficult Theorem \ref{prop:L2} (and its Corollary \ref{thm:L2})
we prove that the entries of the coefficients of $L_1(\partial)$ lie in $\mc W$.
Finally, in Theorem \ref{thm:L3} we show that $L_1(\partial)$
is of bi-Adler type
with respect to the bi-PVA structure of the family $\mc W_\epsilon(\mf g,f,S_1)$, $\epsilon\in\mb F$.
The case of arbitrary non-zero $S\in\mf g_d$
is easily reduced to $S_1$, cf. Theorem \ref{thm:L1}(d), Corollary \ref{thm:L2}(b) and Theorem \ref{thm:L3}(b).

In order to compute the Adler type operator $L_1(\partial)$
in terms of a set of generators of the $\mc W$-algebra $\mc W(\mf g,f)$,
we choose in Section \ref{sec:5} a convenient slice $U$ to the adjoint orbit of $f$,
defined by \eqref{eq:basisU}.
With this choice, we define the corresponding set of generators $\{w_{ji:k}\}$ of $\mc W(\mf g,f)$,
indexed by indices $1\leq i,j\leq r$ and $0\leq k\leq\min\{p_i,p_j\}-1$.
We are then able to find an explicit general formula for $L_1(\partial)$:
\begin{equation}\label{eq:intro-L1}
L_1(\partial)
=
-\id_{r_1}(-\partial)^{p_1}+W_1(\partial)
-W_2(\partial)(-(-\partial)^q+W_4(\partial))^{-1}W_3(\partial)
\,,
\end{equation}
where $W_1,W_2,W_3,W_4$ are the four blocks,
of sizes $r_1\times r_1$, $r_1\times(r-r_1)$, $(r-r_1)\times r_1$ and $(r-r_1)\times(r-r_1)$ respectively,
of the matrix
$$
\left(\begin{array}{ll} W_1(\partial) & W_2(\partial) \\ W_3(\partial) & W_4(\partial) \end{array}\right)
=\big(\sum_kw_{ji;k}(-\partial)^k\big)_{1\leq i,j\leq r}
\,,
$$
and $(-\partial)^q$ is the diagonal $(r-r_1)\times(r-r_1)$ matrix
with diagonal entries $(-\partial)^{p_i}$, $r_1<i\leq r$.
Formula \eqref{eq:intro-L1} has a two-fold application.
When combined with the definition \eqref{eq:intro-1.11} of $L_1(\partial)$
it provides an explicit formula for all the generators $w_{ji;k}$ of $\mc W(\mf g,f)$,
as elements of the differential algebra $\mc V(\mf g_{\leq\frac12})$.
On the other hand, when combined with the Adler identity
\eqref{eq:adler-intro} (resp. the bi-Adler identity \eqref{eq:bi-adler}),
it provides explicit formulas for the $0$-th (resp. $1$-st) $\lambda$-brackets between 
all the generators $w_{ji;k}$.

In Section \ref{sec:7} we will demonstrate how this is implemented in several
examples: the case of a principal nilpotent element
(corresponding to the partition $p=N$) in Section \ref{sec:7.1},
of a rectangular nilpotent element 
(corresponding to $p=(p_1,\dots,p_1)$) in Section \ref{sec:7.2},
of a short nilpotent element (corresponding to $p=(2,\dots,2)$) in Section \ref{sec:7.3},
of a minimal nilpotent element (corresponding to $p=(2,1\dots,1)$) in Section \ref{sec:7.4},
and of a vector and matrix ``constrained'' nilpotent element 
(corresponding to $p=(p_1,1,\dots,1)$ and $p=(p_1,\dots,p_1,1,\dots,1)$ respectively)
in Section \ref{sec:7.5}.
In all these examples we will use equation \eqref{eq:intro-L1}
to find the explicit formulas for the generators $w_{ji;k}$ of the $\mc W$-algebra
and for their $0$-th and $1$-st $\lambda$-brackets.
In each case, we shall compare our results with the analogous formulas
which can be found in literature:
such as \cite{GD78,DSKV15,MR15} for the principal and rectangular nilpotents,
\cite{Che92,DSKV14a} for the minimal and short nilpotents.
For the vector and matrix ``constrained'' nilpotents the explicit formulas the $\lambda$-brackets of $\mc W(\mf g,f)$
were, in fact, not known,
and our formulas \eqref{LHS} and \eqref{eq:constraint-1st} 
constitute a new result obtained by our general method.

Our construction encompasses many well known
reductions of the (matrix) KP hierarchy
and automatically provides them with a bi-Poisson structure.
For example, if $f$ is a principal nilpotent element of $\mf{gl}_N$,
then $L_1(\partial)$ is the ``generic'' monic scalar differential operator of order $N$,
and in this case \eqref{eq:intro-1.5}
is the Gelfand-Dickey $N$-th KdV hierarchy.
If $f$ is a rectangular nilpotent,
we similarly obtain the $p_1$-th $r_1\times r_1$ matrix KdV hierarchy.
If $f$ is a vector constrained nilpotent,
we obtain a bi-Hamiltonian hierarchy
whose Dirac reduction is the $(N-p_1)$-vector $p_1$-constrained KP hierarchy
studied by many authors \cite{YO76,Ma81,KSS91,Che92,KS92,SS93,ZC94}.
If $f$ is a matrix constrained nilpotent,
we obtain a matrix generalization of the vector constrained KP hierarchy.
In fact, for every partition $p$ of $N$,
we obtain a reduction of the $r_1\times r_1$ matrix KP hierarchy,
thereby providing all classical affine $\mc W$-algebras associated to $\mf{gl}_N$ with an integrable
bi-Hamiltonian hierarchy.

Our method can be extended to the other classical Lie algebras $\mf g=\mf{so}_N$
and $\mf{sp}_N$.
Moreover, the Adler type operator approach to $\mc W$-algebras
has a natural quantization, related to the notion of Yangians.
We plan to address these questions in forthcoming publications.

\medskip

The first two authors would like to acknowledge
the hospitality of IHES, France,
where this work was completed in the summer of 2015.
The first author is supported by a national FIRB grant,
the second author is supported by an NSF grant,
and the third author is supported by an NSFC 
``Research Fund for International Young Scientists'' grant.

\section{(bi)Adler type matrix pseudodifferential operators and (bi)Hamiltonian hierarchies}\label{sec:2}

In this section we review the main notions and the main results of \cite{DSKVnew},
which will be used in the following sections.
Throughout the paper the base field $\mb F$ is a field of characteristic $0$.

\subsection{(bi)Poisson vertex algebras and (bi)Hamiltonian equations}\label{sec:2.1}

By a differential algebra we mean a commutative associative unital algebra $\mc V$
with a derivation $\partial:\,\mc V\to\mc V$.
A $\lambda$-\emph{bracket} on $\mc V$ 
is a bilinear (over $\mb F$) map $\{\cdot\,_\lambda\,\cdot\}:\,\mc V\times\mc V\to\mc V[\lambda]$ 
satisfying the following axioms ($a,b,c\in\mc V$):
\begin{enumerate}[(i)]
\item
sesquilinearity:
$\{\partial a_\lambda b\}=-\lambda\{a_\lambda b\}$,
$\{a_\lambda\partial b\}=(\lambda+\partial)\{a_\lambda b\}$;
\item
Leibniz rules:
$\{a_\lambda bc\}=\{a_\lambda b\}c+\{a_\lambda c\}b$,
$\{ab_\lambda c\}=\{a_{\lambda+\partial} c\}_\to b
+\{b_{\lambda+\partial} c\}_\to a$,
\end{enumerate}
where $\to$ means that $\partial$ is moved to the right.
We say that $\mc V$ is a \emph{Poisson vertex algebra} (PVA) if 
it is endowed with a $\lambda$-bracket $\{\cdot\,_\lambda\,\cdot\}$
satisfying ($a,b,c\in\mc V$)
\begin{enumerate}[(i)]
\setcounter{enumi}{2}
\item
skewsymmetry:
$\{b_\lambda a\}=-\{a_{-\lambda-\partial} b\}$ (with $\partial$ acting on the coefficients);
\item
Jacobi identity:
$\{a_\lambda \{b_\mu c\}\}-\{b_\mu\{a_\lambda c\}\}
=\{\{a_\lambda b\}_{\lambda+\mu}c\}$.
\end{enumerate}
If an $\mb F[\partial]$-module $R$ is endowed with a sesquilinear map 
$\{\cdot\,_\lambda\,\cdot\}:\,R\times R\to R[\lambda]$ 
satisfying the skewsymmetry and Jacobi identity axioms,
we say that $R$ is a \emph{Lie conformal algebra}.

Let $\mc V$ be a Poisson vertex algebra with $\lambda$-bracket $\{\cdot\,_\lambda\,\cdot\}$.
We have the corresponding Lie algebra on $\mc V/\partial\mc V$
with Lie bracket $\{\tint f,\tint g\}=\tint\{f_\lambda g\}\big|_{\lambda=0}$,
and a representation of the Lie algebra $\mc V/\partial\mc V$ on $\mc V$
given by the action $\{\tint f,g\}=\{f_\lambda g\}\big|_{\lambda=0}$.
Recall that the basic problem in the theory of integrability
is to construct an infinite sequence of elements $\tint h_n\in\mc V/\partial\mc V,\,n\in\mb Z_+$,
called Hamiltonian functionals, in involution, i.e. such that
$$
\{\tint h_m,\tint h_n\}=0
\,\text{ for all } m,n\in\mb Z_+
\,.
$$
In this case we obtain a hierarchy of compatible Hamiltonian equations
$$
\frac{du}{dt_n}=\{\tint h_n,u\}
\,,\,\, u\in\mc V
\,.
$$

Let $\{\cdot\,_\lambda\,\cdot\}_0$ and $\{\cdot\,_\lambda\,\cdot\}_1$
be two $\lambda$-brackets on the same differential algebra $\mc V$.
We can consider the pencil of $\lambda$-brackets
\begin{equation}\label{eq:pencil}
\{\cdot\,_\lambda\,\cdot\}_\epsilon
=
\{\cdot\,_\lambda\,\cdot\}_0
+
\epsilon\{\cdot\,_\lambda\,\cdot\}_1
\quad,\qquad \epsilon\in\mb F
\,.
\end{equation}
We say that $\mc V$ is a \emph{bi}-\emph{PVA} 
if $\{\cdot\,_\lambda\,\cdot\}_\epsilon$ is a PVA $\lambda$-bracket on $\mc V$ for every $\epsilon\in\mb F$.
Clearly, for this it suffices that $\{\cdot\,_\lambda\,\cdot\}_0$,
$\{\cdot\,_\lambda\,\cdot\}_1$ and $\{\cdot\,_\lambda\,\cdot\}_0+\{\cdot\,_\lambda\,\cdot\}_1$
are PVA $\lambda$-brackets.

Let $\mc V$ be a bi-Poisson vertex algebra with $\lambda$-brackets 
$\{\cdot\,_\lambda\,\cdot\}_0$ and $\{\cdot\,_\lambda\,\cdot\}_1$.
A \emph{bi}-\emph{Hamiltonian equation} is an evolution equation 
which can be written in Hamiltonian form with respect to both PVA $\lambda$-brackets
and two Hamiltonian functionals $\tint h_0,\tint h_1\in\mc V/\partial\mc V$:
$$
\frac{du}{dt}
=
\{\tint h_0,u\}_0
=
\{\tint h_1,u\}_1
\,,\,\, u\in\mc V
\,.
$$
The usual way to prove integrability for a bi-Hamiltonian equation
is to solve the so called Lenard-Magri
recurrence relation ($u\in\mc V$):
\begin{equation}\label{eq:LM}
\{\tint h_n,u\}_0
=
\{\tint h_{n+1},u\}_1
\,\,,\,\,\,\,
n\in\mb Z_+
\,.
\end{equation}
In this way, 
we get the corresponding hierarchy of bi-Hamiltonian equations
$$
\frac{du}{dt_n}
=
\{\tint h_n,u\}_0
=
\{\tint h_{n+1},u\}_1
\,,\,\, 
n\in\mb Z_+,\, u\in\mc V
\,.
$$

\subsection{Adler type matrix pseudodifferential operators}\label{sec:2.2}

\begin{definition}[\cite{DSKV15}]\label{def:adler}
An $M\times N$ matrix pseudodifferential operator $A(\partial)$
over a differential algebra $\mc V$
is of \emph{Adler type} with respect to a $\lambda$-bracket $\{\cdot\,_\lambda\,\cdot\}$ on $\mc V$,
if, for every $(i,j),(h,k)\in\{1,\dots,M\}\times\{1,\dots,N\}$, we have
\begin{equation}\label{eq:adler}
\begin{split}
\{A_{ij}(z)_\lambda A_{hk}(w)\}
& = A_{hj}(w+\lambda+\partial)\iota_z(z\!-\!w\!-\!\lambda\!-\!\partial)^{-1}(A_{ik})^*(\lambda-z)
\\
& - A_{hj}(z)\iota_z(z\!-\!w\!-\!\lambda\!-\!\partial)^{-1}A_{ik}(w)
\,.
\end{split}
\end{equation}
In \eqref{eq:adler} $(A_{ik})^*(\partial)$ denotes the formal adjoint of the scalar 
pseudodifferential operator $A_{ik}(\partial)$,
and $(A_{ik})^*(z)$ is its symbol,
and $\iota_z$ denotes the expansion in geometric series for large $z$.
Also, let $S\in\Mat_{M\times N}\mb F$.
We say that $A$ is of $S$-\emph{Adler type}
with respect to two $\lambda$-brackets $\{\cdot\,_\lambda\,\cdot\}_0$
and $\{\cdot\,_\lambda\,\cdot\}_1$
if, for every $\epsilon\in\mb F$,
$A(\partial)+\epsilon S$ is a matrix of Adler type with respect 
to the $\lambda$-bracket $\{\cdot\,_\lambda\,\cdot\}_\epsilon$.
In the case $M=N$,
we also say that $A$ is of \emph{bi}-\emph{Adler type}
if it is of $\id_N$-Adler type.
This is equivalent to saying that $A(\partial)$
is a matrix of Adler type with respect to the $\lambda$-bracket $\{\cdot\,_\lambda\,\cdot\}_0$,
i.e. \eqref{eq:adler} holds, and that 
\begin{equation}\label{eq:bi-adler}
\begin{split}
\{A_{ij}(z)_\lambda A_{hk}(w)\}_1
& = 
\delta_{ik}
\iota_z(z\!-\!w\!-\!\lambda)^{-1}
\big(
A_{hj}(w+\lambda)
- A_{hj}(z)
\big)
\\
& +
\delta_{hj}
\iota_z(z\!-\!w\!-\!\lambda\!-\!\partial)^{-1}
\big(
(A_{ik})^*(\lambda-z)
-A_{ik}(w)
\big)
\,.
\end{split}
\end{equation}
\end{definition}
\begin{example}\label{ex:A}
For a Lie algebra $\mf g$ with a non-degenerate symmetric invariant bilinear form $(\cdot\,|\,\cdot)$
and an element $S\in\mf g$,
we define the corresponding pencil of affine PVAs $\mc V_\epsilon(\mf g,S)$
as follows.
The underlying differential algebra is the algebra $\mc V(\mf g)$ of differential polynomials over $\mf g$.
The PVA $\lambda$-bracket,
depending on a parameter $\epsilon\in\mb F$,
is given on generators by 
\begin{equation}\label{lambda}
\{a_\lambda b\}_\epsilon=[a,b]+(a| b)\lambda+\epsilon(S|[a,b]),
\qquad a,b\in\mf g\,,
\end{equation}
and extended to $\mc V(\mf g)$ by the sesquilinearity axioms and the Leibniz rules.
As shown in \cite[Ex.3.4]{DSKVnew}
we have the following $N\times N$ matrix differential operator
of $S$-Adler type with respect to the bi-PVA structure of $\mc V_\epsilon(\mf g,S)$,
where $\mf g=\mf{gl}_N$ with the trace form $(\cdot\,|\,\cdot)$:
\begin{equation}\label{eq:A}
A(\partial)
=
\id_N\partial+Q
\,,\,\text{ where }\,\,
Q=\sum_{i,j=1}^Nq_{ji}E_{ij}
\,\in\Mat_{N\times N}\mc V(\mf g)
\,.
\end{equation}
Here and further we denote by
$E_{ij}\in\Mat_{N\times N}\mb F$ the elementary matrix with $1$ in position $(ij)$
and $0$ everywhere else,
and we denote by $q_{ij}\in\mf g=\mf{gl}_N$
the same matrix
when viewed as an element of the differential algebra $\mc V(\mf g)$.
Hence, $Q$ in \eqref{eq:A} is the $N\times N$ matrix which, in position $(ij)$,
has entry $q_{ji}\in\mc V(\mf g)$.
\end{example}
One of the main properties of Adler type operator,
which will be used later, is the following:
\begin{theorem}[{\cite[Thm.3.7(c)]{DSKVnew}}]\label{thm:inverse-adler}
Let $\mc V$ be a differential algebra with a $\lambda$-bracket $\{\cdot\,_\lambda\,\cdot\}$.
Let $A(\partial)\in\Mat_{N\times N}\mc V((\partial^{-1}))$ be a matrix 
pseudodifferential operator of Adler type with respect to the $\lambda$-bracket of $\mc V$.
If $A(\partial)$ is invertible in $\Mat_{N\times N}\mc V((\partial^{-1}))$,
then $A^{-1}(\partial)$ is of Adler type with respect to the opposite $\lambda$-bracket
$-\{\cdot\,_\lambda\,\cdot\}$.
\end{theorem}
The relation between operators of $S$-Adler type and Poisson vertex algebras
is described by the following result:
\begin{theorem}[{\cite[thm.6.3]{DSKVnew},\cite{DSKV15}}]\label{thm:main-bi-adler}
Let $A(\partial)\in\Mat_{M\times N}\mc V((\partial^{-1}))$ be an $M\times N$-matrix 
pseudodifferential operator of $S$-Adler type,
for some $S\in\Mat_{M\times N}\mb F$,
with respect to the $\lambda$-brackets 
$\{\cdot\,_\lambda\,\cdot\}_0$ and $\{\cdot\,_\lambda\,\cdot\}_1$ on $\mc V$.
Assume that the coefficients of the entries of the matrix $A(\partial)$
generate $\mc V$ as a differential algebra.
Then $\mc V$ is a bi-PVA with the $\lambda$-brackets $\{\cdot\,_\lambda\,\cdot\}_0$
and $\{\cdot\,_\lambda\,\cdot\}_1$.
\end{theorem}

\subsection{Integrable hierarchy associated to a matrix pseudodifferential operator of Adler type}
\label{sec:2.3}

The following Theorem, proved in \cite[Theorems 5.1 and 6.4]{DSKVnew},
shows how (bi)Adler type operators can be used to construct
integrable (bi)Hamiltonian hierarchies.
\begin{theorem}\label{thm:hn}
Let $\mc V$ be a differential algebra with a $\lambda$-bracket $\{\cdot\,_\lambda\,\cdot\}$.
Let $A(\partial)\in\Mat_{N\times N}\mc V((\partial^{-1}))$
be a matrix pseudodifferential operator of Adler type 
with respect to the $\lambda$-bracket $\{\cdot\,_\lambda\,\cdot\}$,
and assume that $A(\partial)$ is invertible in $\Mat_{N\times N}\mc V((\partial^{-1}))$.
For $B(\partial)\in\Mat_{N\times N}\mc V((\partial^{-1}))$
a $K$-th root of $A$ (i.e. $A(\partial)=B(\partial)^K$ for $K\in\mb Z\backslash\{0\}$)
define the elements $h_{n,B}\in\mc V$, $n\in\mb Z$, by
\begin{equation}\label{eq:hn}
h_{n,B}=
\frac{-K}{|n|}
\Res_z\tr(B^n(z))
\text{ for } n\neq0
\,,\,\,
h_0=0\,.
\end{equation}
Then: 
\begin{enumerate}[(a)]
\item
All the elements $\tint h_{n,B}$ are Hamiltonian functionals in involution:
\begin{equation}\label{eq:invol}
\{\tint h_{m,B},\tint h_{n,C}\}=0
\,\text{ for all } m,n\in\mb Z,\,
B,C \text{ roots of } A
\,.
\end{equation}
\item
The corresponding compatible hierarchy of Hamiltonian equations is
\begin{equation}\label{eq:hierarchy}
\frac{dA(z)}{dt_{n,B}}
=
\{\tint h_{n,B},A(z)\}
=
[(B^n)_+,A](z)
\,,\,\,n\in\mb Z,\,
B \text{ root of } A
\end{equation}
(in the RHS we are taking the symbol of the commutator of matrix pseudodifferential operators),
and the Hamiltonian functionals $\tint h_{n,C}$, $n\in\mb Z_+$, $C$ root of $A$,
are integrals of motion of all these equations.
\item
If, moreover, $A(\partial)$ is of bi-Adler type 
with respect to two $\lambda$-brackets $\{\cdot\,_\lambda\,\cdot\}_0$ 
and $\{\cdot\,_\lambda\,\cdot\}_1$,
then the elements $h_{n,B}\in\mc V$, $n\in\mb Z_+$, given by \eqref{eq:hn}
satisfy the generalized Lenard-Magri recurrence relation:
\begin{equation}\label{eq:LM-K}
\{\tint h_{n,B},A(z)\}_0
=
\{\tint h_{n+K,B},A(z)\}_1
=
[(B^n)_+,A](z)
\,\,,\,\,\,\,
n\in\mb Z
\,.
\end{equation}
Hence, \eqref{eq:hierarchy} is a compatible hierarchy of bi-Hamiltonian equations,
and all the Hamiltonian functionals $\tint h_{n,C}$, $n\in\mb Z_+$, $C$ root of $A$,
are integrals of motion of all the equations of this hierarchy.
\end{enumerate}
\end{theorem}

\subsection{Generalized quasideterminants}\label{sec:2.4}

Following \cite{DSKVnew} introduce the following generalization of quasideterminants, cf. \cite{GGRW05}.
Let $A\in\Mat_{N\times N}R$,
where $R$ is a unital associative algebra,
and let $I\in\Mat_{N\times M}R$, $J\in\Mat_{M\times N}R$,
for some $M\leq N$.
\begin{definition}\label{def:gen-quasidet}
The $(I,J)$-\emph{quasideterminant} of $A$ is
\begin{equation}\label{eq:gen-quasidet}
|A|_{IJ}
=
(JA^{-1}I)^{-1}\,\in\Mat_{M\times M}R
\,,
\end{equation}
assuming that the RHS makes sense, i.e. that $A$ is invertible in $\Mat_{N\times N}R$
and that $JA^{-1}I$ is invertible in $\Mat_{M\times M}R$.
\end{definition}

A special case is when $I$ and $J$ are the following matrices:
\begin{equation}\label{eq:ENM}
I_{NM}=
\left(\begin{array}{l}
\id_{M\times M} \\ 0_{(N-M)\times M}
\end{array}\right)\in\Mat_{N\times M}\mb F
\,,
\end{equation}
and
\begin{equation}\label{eq:EMN}
J_{MN}=
\left(\begin{array}{ll}
\id_{M\times M} & 0_{M\times(N-M)}
\end{array}\right)\in\Mat_{M\times N}\mb F
\,.
\end{equation}
In this case the corresponding quasideterminant
has the following explicit formula (\cite[Prop.4.2]{DSKVnew})
\begin{equation}\label{eq:spec-quasidet}
|A|_{I_{NM}J_{MN}}
=
a-bd^{-1}c\,,
\,
\end{equation}
where $A$ has the block form $A=\left(\begin{array}{ll} a&b \\ c&d \end{array}\right)$,
where $a,\,b,\,c$ and $d$ are matrices of sizes $M\times M$, $M\times(N-M)$, $(N-M)\times M$, and $(N-M)\times(N-M)$
respectively.

Let $I=I_1I_2$ and $J=J_2J_1$, where
$I_1\in\Mat_{N\times M_1}R$, $J_1\in\Mat_{M_1\times N}R$,
$I_2\in\Mat_{M_1\times M_2}R$ and $J_2\in\Mat_{M_2\times M_1}R$.
The following hereditary property of generalized quasideterminants is an obvious consequence of
the definition \eqref{eq:gen-quasidet}:
\begin{equation}\label{eq:hereditary}
|A|_{IJ}=||A|_{I_1J_1}|_{I_2J_2}
\,,
\end{equation}
provided that all generalized quasideterminants involved exist.

The following result, based on Theorem \ref{thm:inverse-adler},
says that identity \eqref{eq:adler}
is preserved under taking generalized quasideterminants.
\begin{theorem}[{\cite[Prop.4.6]{DSKVnew}}]\label{thm:quasidet-adler}
Let $\mc V$ be a differential algebra with a $\lambda$-bracket $\{\cdot\,_\lambda\,\cdot\}$.
Let $A(\partial)\in\Mat_{N\times N}\mc V((\partial^{-1}))$ be a matrix 
pseudodifferential operator of Adler type with respect to the $\lambda$-bracket of $\mc V$.
Then, for every $I\in\Mat_{N\times M}\mb F$
and $J\in\Mat_{M\times N}\mb F$ with $M\leq N$,
the generalized quasideterminant $|A(\partial)|_{IJ}$,
provided that it exists,
is an $M\times M$ matrix pseudodifferential operator of Adler type.
\end{theorem}
A simple but important result,
which we will use in Section \ref{sec:4}, is the following generalization of \cite[Thm.4.5]{DSKVnew}:
\begin{theorem}\label{thm:main-quasidet}
Let $A\in\Mat_{N\times N}R$, 
$I\in\Mat_{N\times M}R$, $J\in\Mat_{M\times N}R$,
and $S_0\in\Mat_{M\times M}R$,for some $M\leq N$.
Let also $S=IS_0J\in\Mat_{N\times N}R$.
Assume that the $(I,J)$-quasideterminant $|A|_{IJ}$ exists
and that the matrix $A+S\in\Mat_{N\times N}R$ is invertible.
Then, the $(I,J)$-quasideterminant of $A+S$ exists, and it is given by
\begin{equation}\label{eq:main-quasidet}
|A+S|_{IJ}=|A|_{IJ}+S_0
\,.
\end{equation}
\end{theorem}
\begin{proof}
It is the same as the proof of \cite[Thm.4.5]{DSKVnew}:
multiplying the identity $A+S=A+S$ by $J(A+S)^{-1}$ on the left and by $A^{-1}I|A|_{IJ}$ on the right,
we get that $S_0+|A|_{IJ}$ is a right inverse of $J(A+S)^{-1}I$,
and multiplying the same identity $A+S=A+S$ by $|A|_{IJ}JA^{-1}$ on the left and by $(A+S)^{-1}I$ on the right,
we get that $S_0+|A|_{IJ}$ is a left inverse of $J(A+S)^{-1}I$.
\end{proof}

\subsection{A new scheme of integrability of bi-Hamiltonian PDE}\label{sec:2.5}

In \cite[Sec.6.3]{DSKVnew} we propose the following scheme of integrability,
based on Theorems \ref{thm:hn} and \ref{thm:main-quasidet}.
Let $\mc V$ be a differential algebra with compatible PVA $\lambda$-brackets 
$\{\cdot\,_\lambda\,\cdot\}_0$ and $\{\cdot\,_\lambda\,\cdot\}_1$.
Let $S\in\Mat_{N\times N}\mb F$ 
and let $A(\partial)\in\Mat_{N\times N}\mc V((\partial^{-1}))$
be an operator of  $S$-Adler type with respect to the
$\lambda$-brackets $\{\cdot\,_\lambda\,\cdot\}_0$ and $\{\cdot\,_\lambda\,\cdot\}_1$.
Assume (without loss of generality) that the differential algebra $\mc V$
is generated by the coefficients of $A(\partial)$.
Then, we obtain an integrable hierarchy of bi-Hamiltonian equations
as follows:
\begin{enumerate}[1.]
\item
consider the canonical factorization $S=IJ$,
where $J:\,\mb F^N\twoheadrightarrow\im(S)$
and $I:\,\im(S)\hookrightarrow\mb F^N$;
\item
assume that the $(I,J)$-quasideterminant $|A|_{IJ}(\partial)$ exists;
then, by Theorem \ref{thm:main-quasidet} and Proposition \ref{thm:quasidet-adler}
$|A|_{IJ}$ is an $M\times M$ matrix pseudodifferential operator
(where $M=\dim\im(S)$) of bi-Adler type
with respect to the same $\lambda$-brackets 
$\{\cdot\,_\lambda\,\cdot\}_0$ and $\{\cdot\,_\lambda\,\cdot\}_1$;
\item
consider the family of local functionals
$\{\tint h_{n,B}\,|\,n\in\mb Z,\,B \text{ a } K\text{-th root of } |A|_{IJ}\}$
given by \eqref{eq:hn};
then, by Theorem \ref{thm:hn}
they are all Hamiltonian functionals in involution with respect to both PVA $\lambda$-brackets
$\{\cdot\,_\lambda\,\cdot\}_0$ and $\{\cdot\,_\lambda\,\cdot\}_1$,
and they satisfy the Lenard-Magri recurrence relation \eqref{eq:LM-K};
\item
we thus get an integrable hierarchy of bi-Hamiltonian equations
\begin{equation}\label{eq:bi-hierarchy}
\frac{du}{dt_{n,B}}
=
\{\tint h_{n,B},u\}_0
=
\{\tint h_{n+K,B},u\}_1
\,,
\end{equation}
provided that the $\tint h_{n,B}$ span an infinite dimensional space.
\end{enumerate}

In the present paper we implement the above scheme
to construct integrable hierarchies associated to the classical affine $\mc W$-algebras
$\mc W(\mf g,f)$ for the Lie algebra $\mf g=\mf{gl}_N$ 
and an arbitrary nilpotent element $f\in\mf g$.

\begin{remark}\label{rem:chg-basis}
The canonical factorization $S=IJ$, with $I\in\Mat_{N\times M}\mb F$ and $J\in\Mat_{M\times N}\mb F$,
is unique up to a choice of basis of $\im S$.
Changing basis leads to a conjugation of the generalized quasideterminant $|A|_{IJ}$
by the change of basis matrix.
Hence, the functionals $\tint h_{n,B}$, $n\in\mb Z_+$, defined by \eqref{eq:hn},
are independent of the choice of basis.
\end{remark}

\subsection{Generic matrices and their properties}
\label{sec:2.6}

Let $\mc V$ be a differential algebra.
Assume that $\mc V$ is an integral domain, and let $\mc K$ be its field of fractions,
which is automatically a differential field.
\begin{definition}\label{def:generic}
A matrix $Q\in\Mat_{N\times\tilde N}\mc V$ is called \emph{generic}
if its entries are differentially independent,
i.e. there is no non-zero differential polynomial over 
the base field $\mb F$
satisfied by the entries of the matrix $Q$.
\end{definition}
\begin{lemma}\label{lem:generic1}
If $Q\in\Mat_{N\times\tilde N}\mc V$ is a generic matrix
then every submatrix obtained by considering some rows and columns of $Q$ is generic.
\end{lemma}
\begin{proof}
Obvious.
\end{proof}
\begin{lemma}\label{lem:generic2}
Let $P_1\in\Mat_{N\times N}\mb F$ and $P_2\in\Mat_{\tilde N\times\tilde N}\mb F$
be invertible matrices with entries in the field of constants $\mb F$.
A matrix $Q\in\Mat_{N\times\tilde N}\mc V$ is generic
if and only if the matrix $P_1QP_2\in\Mat_{N\times\tilde N}\mc V$ is generic.
\end{lemma}
\begin{proof}
The change of variable mapping the entries $q_{ij}$ of the matrix $Q$
to the entries $\tilde q_{ij}$ of the matrix $P_1QP_2$ is invertible.
Hence, if there exists a non-zero differential polynomial over $\mb F$ 
which is zero when evaluated on the $q_{ij}$'s,
after the change of variable we get a non trivial differential polynomial 
which is zero when evaluated on the $\tilde q_{ij}$'s,
and conversely.
\end{proof}
\begin{lemma}\label{lem:generic3}
Let $Q\in\Mat_{N\times\tilde N}\mc V$ be a generic matrix.
Let $J\in\Mat_{M\times N}\mb F$ be a constant matrix of rank $M\, (\leq N)$,
and let $I\in\Mat_{\tilde N\times\tilde M}\mb F$ be a constant matrix of rank $\tilde M\,(\leq\tilde N)$.
Then the matrix $JQI\in\Mat_{M\times\tilde M}\mc V$ is generic.
\end{lemma}
\begin{proof}
By elementary transformations, there exist invertible constant matrices
$P_1\in\Mat_{M\times M}\mb F$, $P_2\in\Mat_{N\times N}\mb F$, 
$P_3\in\Mat_{\tilde N\times\tilde N}\mb F$ and $P_4\in\Mat_{\tilde M\times\tilde M}\mb F$
such that
$$
P_1JP_2=
\left(\begin{array}{ll} \id_{M\times M} & 0_{M\times(N-M)} \end{array}\right)
\,\,,\,\,\,\,
P_3IP_4=
\left(\begin{array}{l} \id_{\tilde M\times\tilde M} \\ 0_{(\tilde N-\tilde M)\times\tilde M} 
\end{array}\right)
\,.
$$
Hence, the matrix
$$
P_1JQIP_4=
\left(\begin{array}{ll} \id_{M\times M} & 0_{M\times(N-M)} \end{array}\right)
P_2^{-1}QP_3^{-1}
\left(\begin{array}{l} \id_{\tilde M\times\tilde M} \\ 0_{(\tilde N-\tilde M)\times\tilde M} 
\end{array}\right)
$$
coincides with the upper left $M\times\tilde M$ block of the matrix $P_2^{-1}QP_3^{-1}$,
which is generic by Lemmas \ref{lem:generic1} and \ref{lem:generic2}.
Therefore, by Lemma \ref{lem:generic2} it follows that the matrix $JQI$
is generic as well.
\end{proof}
\begin{lemma}\label{lem:generic4}
If $Q\in\Mat_{N\times N}\mc K$ is a generic matrix
then it is invertible.
\end{lemma}
\begin{proof}
The determinant of $Q$, being a non-zero polynomial in the entries of $Q$,
cannot be zero.
\end{proof}
\begin{lemma}\label{lem:generic5}
Let $A(\partial)\in\Mat_{N\times N}\mc K((\partial^{-1}))$
be a matrix pseudodifferential operator with the block form
$$
A(\partial)
=
\left(\begin{array}{ll} 
A_{11}(\partial) & A_{12}(\partial) \\
A_{21}(\partial) & A_{22}(\partial) 
\end{array}\right)
\,,
$$
where the submatrices
$A_{11}(\partial)\in\Mat_{r\times r}\mc K((\partial^{-1}))$,
$A_{12}(\partial)\in\Mat_{r\times(N-r)}\mc K((\partial^{-1}))$,
$A_{21}(\partial)\in\Mat_{(N-r)\times r}\mc K((\partial^{-1}))$,
$A_{22}(\partial)\in\Mat_{(N-r)\times(N-r)}\mc K((\partial^{-1}))$,
are pseudodifferential operators of orders 
(= maximal orders of their entries) 
$n_{11}$, $n_{12}$, $n_{21}$ and $n_{22}\,\in\mb Z$ respectively,
such that $n_{11}+n_{22}>n_{12}+n_{21}$.
Assume moreover that the square matrices $A_{11}(\partial)$ and $A_{22}(\partial)$ 
have invertible leading coefficients.
Then the matrix $A(\partial)$ is invertible.
\end{lemma}
\begin{proof}
Under our assumptions, the matrices  $A_{11}(\partial)$ and $A_{22}(\partial)$,
having invertible leading coefficients, are invertible,
and their inverses have order $-n_{11}$ and $-n_{22}$ respectively.
Moreover, under the conditions on the orders,
the matrices
$$
A_{11}(\partial)-A_{12}(\partial)A_{22}(\partial)^{-1}A_{21}(\partial)
\,\,\text{ and }\,\,
A_{22}(\partial)-A_{21}(\partial)A_{11}(\partial)^{-1}A_{12}(\partial)
$$
have order $n_{11}$ and $n_{22}$ respectively, 
and they have the same (invertible) leading coefficients
of $A_{11}(\partial)$ and $A_{22}(\partial)$ respectively.
Hence, they are invertible as well.
But then the inverse matrix of $A(\partial)$
exists since it has the block form
$$
A(\partial)^{-1}
=
\left(\begin{array}{ll} 
B_{11}(\partial) & B_{12}(\partial) \\
B_{21}(\partial) & B_{22}(\partial) 
\end{array}\right)
\,,
$$
where
\begin{equation*}
\begin{split}
& B_{11}(\partial)
=
(A_{11}(\partial)-A_{12}(\partial)A_{22}(\partial)^{-1}A_{21}(\partial))^{-1}
\,,\\
& B_{12}(\partial)
=
-A_{11}(\partial)^{-1}A_{12}(\partial)
(A_{22}(\partial)-A_{21}(\partial)A_{11}(\partial)^{-1}A_{12}(\partial))^{-1}
\,,\\
& B_{21}(\partial)
=
-A_{22}(\partial)^{-1}A_{21}(\partial)
(A_{11}(\partial)-A_{12}(\partial)A_{22}(\partial)^{-1}A_{21}(\partial))^{-1}
\,,\\
& B_{22}(\partial)
=
(A_{22}(\partial)-A_{21}(\partial)A_{11}(\partial)^{-1}A_{12}(\partial))^{-1}
\,.
\end{split}
\end{equation*}
\end{proof}
\begin{proposition}\label{prop:generic}
Let $A(\partial)=A_n\partial^n+A_{n-1}\partial^{n-1}+\dots\in\Mat_{N\times N}\mc K((\partial^{-1}))$
be a matrix pseudodifferential operator
such that $A_n\in\Mat_{N\times N}\mb F$ 
and the matrix $A_{n-1}\in\Mat_{N\times N}\mc K$ is generic.
Then $A(\partial)$ is invertible in $\Mat_{N\times N}\mc K((\partial^{-1}))$.
\end{proposition}
\begin{proof}
After multiplying $A(\partial)$ on the left and on the right
by invertible matrices in $\Mat_{N\times N}\mb F$,
we can assume that $A_n$ has block form
$$
A_n
=
\left(\begin{array}{ll} 
\id_r & 0 \\
0 & 0 
\end{array}\right)
\,,
$$
where $r$ is the rank of $A_n$.
In this case, the matrix $A(\partial)$ has the block form
$$
A(\partial)
=
\left(\begin{array}{ll} 
\id_r\partial^n+A_{n-1;11}\partial^{n-1}+\dots & A_{n-1;12}\partial^{n-1}+\dots \\
A_{n-1;21}\partial^{n-1}+\dots & A_{n-1;22}\partial^{n-1}+\dots 
\end{array}\right)
\,.
$$
The matrix $A_{n-1;22}$ is generic by Lemma \ref{lem:generic1},
and therefore it is invertible by Lemma \ref{lem:generic4}.
Hence, the above block form of the matrix $A(\partial)$ satisfies 
all the assumptions of Lemma \ref{lem:generic5}.
\end{proof}
\begin{example}\label{ex:generic}
Proposition \ref{prop:generic} is false if we replace the assumption that $A_{n-1}$ 
is generic by the assumption that $A_{n-1}$ is invertible. 
For example, the matrix
$\left(\begin{array}{rrr}
0 & a' & -a \\
0 & 0 & 1 \\
1 & a & 0
\end{array}\right)$,
$a\in\mc K$,
is non-degenerate provided that $a'\neq0$,
but the matrix
$\left(\begin{array}{rrr}
1 & 0 & 0 \\
0 & 1 & 0 \\
0 & 0 & 0
\end{array}\right)
\partial
+\left(\begin{array}{rrr}
0 & a' & -a \\
0 & 0 & 1 \\
1 & a & 0
\end{array}\right)$
is degenerate for every $a$.
\end{example}

\section{Classical affine \texorpdfstring{$\mc W$}{W}-algebras and associated bi-Poisson structures}\label{sec:3}

\subsection{Definition of the classical affine $\mc W$-algebra $\mc W_\epsilon(\mf g,f,S)$}\label{sec:3.1}

We review here the construction of the classical affine $\mc W$-algebra
following \cite{DSKV13}.
Let $\mf g$ be a reductive Lie algebra with a non-degenerate symmetric 
invariant bilinear form $(\cdot\,|\,\cdot)$,
and let $\{f,2x,e\}\subset\mf g$ be an $\mf{sl}_2$-triple in $\mf g$.
We have the corresponding $\ad x$-eigenspace decomposition
\begin{equation}\label{eq:grading}
\mf g=\bigoplus_{k\in\frac{1}{2}\mb Z}\mf g_{k}
\,\,\text{ where }\,\,
\mf g_k=\big\{a\in\mf g\,\big|\,[x,a]=ka\big\}
\,,
\end{equation}
so that $f\in\mf g_{-1}$, $x\in\mf g_{0}$ and $e\in\mf g_{1}$.
We let $d$ be the \emph{depth} of the grading, i.e. the maximal eigenvalue of $\ad x$.
For a subspace $\mf a\subset\mf g$ we denote by $\mc V(\mf a)$
the algebra of differential polynomials over $\mf a$,
i.e. $\mc V(\mf a)=S(\mb F[\partial]\mf a)$.

Consider the pencil of affine PVAs $\mc V_\epsilon(\mf g,S)$ defined in Example \ref{ex:A}.
We shall assume that $S$ lies in $\mf g_d$.
In this case the $\mb F[\partial]$-submodule
$\mb F[\partial]\mf g_{\geq\frac12}\subset\mc V(\mf g)$ 
is a Lie conformal subalgebra of $\mc V_\epsilon(\mf g,S)$
with the $\lambda$-bracket $\{a_\lambda b\}_\epsilon=[a,b]$, $a,b\in\mf g_{\geq\frac12}$
(it is independent of $\epsilon$, since $S$ commutes with $\mf g_{\geq\frac12}$).
Consider the differential subalgebra
$\mc V(\mf g_{\leq\frac12})$ of $\mc V(\mf g)$,
and denote by $\rho:\,\mc V(\mf g)\twoheadrightarrow\mc V(\mf g_{\leq\frac12})$,
the differential algebra homomorphism defined on generators by
\begin{equation}\label{rho}
\rho(a)=\pi_{\leq\frac12}(a)+(f| a),
\qquad a\in\mf g\,,
\end{equation}
where $\pi_{\leq\frac12}:\,\mf g\to\mf g_{\leq\frac12}$ denotes 
the projection with kernel $\mf g_{\geq1}$.
We have a representation of the Lie conformal algebra $\mb F[\partial]\mf g_{\geq\frac12}$ 
on the differential subalgebra $\mc V(\mf g_{\leq\frac12})\subset\mc V(\mf g)$,
with the action of $a\in\mf g_{\geq\frac12}$ on $g\in\mc V(\mf g_{\leq\frac12})$
given by $\rho\{a_\lambda g\}_\epsilon$
(note that the RHS is independent of $\epsilon$ since, by assumption, $S\in\mf g_d$).

The \emph{classical} $\mc W$-\emph{algebra} $\mc W_\epsilon(\mf g,f,S)$ is, by definition,
the differential algebra
\begin{equation}\label{20120511:eq2}
\mc W=\mc W(\mf g,f)
=\big\{w\in\mc V(\mf g_{\leq\frac12})\,\big|\,\rho\{a_\lambda w\}_\epsilon=0\,
\text{ for all }a\in\mf g_{\geq\frac12}\}\,,
\end{equation}
endowed with the following pencil of PVA $\lambda$-brackets
\begin{equation}\label{20120511:eq3}
\{v_\lambda w\}^{\mc W}_{\epsilon}=\rho\{v_\lambda w\}_\epsilon,
\qquad v,w\in\mc W\,.
\end{equation}
With a slight abuse of notation,
we shall denote by $\mc W(\mf g,f)$ also the $\mc W$-algebra $\mc W_\epsilon(\mf g,f,S)$
for $\epsilon=0$ (or, equivalently, $S=0$).
\begin{theorem}[{\cite[Lem.3.1, Lem.3.2, Cor.3.3]{DSKV13}}]\phantomsection\label{daniele1}
\begin{enumerate}[(a)]
\item
$\mc W\subset\mc V(\mf g_{\leq\frac12})$ is a differential subalgebra
and, for every $v,w\in\mc W$, we have $\rho\{v_\lambda w\}_\epsilon\in\mc W[\lambda]$.
Hence, the $\lambda$-bracket 
$\{\cdot\,_\lambda\,\cdot\}^{\mc W}_{\epsilon}:\,\mc W\otimes\mc W\to\mc W[\lambda]$,
given by \eqref{20120511:eq3}, defines a pencil of PVA structures on $\mc W$.
\item
For $g,h\in\mc V(\mf g)$ such that $\rho(g),\rho(h)\in\mc W$, we have 
$\{\rho(g)_\lambda\rho(h)\}^{\mc W}_\epsilon=\rho\{g_\lambda h\}_\epsilon$.
\end{enumerate}
\end{theorem}
\begin{remark}\label{rem:generalizations}
The definition of the $\mc W$-algebra can be generalized
to the case of an arbitrary good grading $\mf g=\oplus_j\mf g_j$ 
such that $f\in\mf g_{-1}$ (not necessarily the Dynkin grading) \cite{EK05},
and to arbitrary isotropic subspace $\ell\subset\mf g_{\frac12}$
(not necessarily $\ell=0$, as above) cf. \cite{DSKV13}.
In fact, it can be proved that the ``second'' Poisson structure $\{\cdot\,_\lambda\,\cdot\}_0$
is independent of the choice of good grading and isotropic subspace \cite{BG07}.
On the other hand, the ``first'' Poisson structure $\{\cdot\,_\lambda\,\cdot\}_1$
may vary with these choices, and so the corresponding bi-Hamiltonian integrable hierarchies
as described in Section \ref{sec:4} may be different.
In this paper, for simplicity, we stick to the traditional choice
of Dynkin grading and isotropic subspace $\ell=0$.
However, it should be interesting to investigate how the choices of good grading and
isotropic subspaces affect the corresponding bi-Hamiltonian hierarchies.
\end{remark}

\subsection{Structure Theorem for classical affine $\mc W$-algebras}
\label{sec:3.2}

Fix a subspace $U\subset\mf g$ complementary to $[f,\mf g]$,
which is compatible with the grading \eqref{eq:grading}.
For example, we could take $U=\mf g^e$, the Slodowy slice, as we did in \cite{DSKV13} and \cite{DSKV16},
however, in Section \ref{sec:5.1} we will make a different, more convenient, choice for $U$.
Since $\ad f:\,\mf g_{j}\to\mf g_{j-1}$ is surjective for $j\leq\frac12$, 
we have $\mf g_{\leq-\frac12}\subset[f,\mf g]$.
In particular, we have the direct sum decomposition
\begin{equation}\label{eq:U}
\mf g_{\geq-\frac12}=[f,\mf g_{\geq\frac12}]\oplus U\,.
\end{equation}
Note that, by the non-degeneracy of $(\cdot\,|\,\cdot)$, the orthocomplement to $[f,\mf g]$
is $\mf g^f$, the centralizer of $f$ in $\mf g$.
Hence, the direct sum decomposition dual to \eqref{eq:U}
has the form
\begin{equation}\label{eq:Uperp}
\mf g_{\leq\frac12}=U^\perp\oplus\mf g^f\,.
\end{equation}
As a consequence of \eqref{eq:Uperp}
we have the decomposition in a direct sum of subspaces
\begin{equation}\label{eq:decomp}
\mc V(\mf g_{\leq\frac12})=\mc V(\mf g^f)\oplus\langle U^\perp\rangle\,,
\end{equation}
where $\langle U^\perp\rangle$
is the differential algebra ideal of $\mc V(\mf g_{\leq\frac12})$ generated by $U^\perp$.
Let $\pi_{\mf g^f}:\,\mc V(\mf g_{\leq\frac12})\twoheadrightarrow\mc V(\mf g^f)$
be the canonical quotient map, with kernel $\langle U^\perp\rangle$.

As an immediate consequence of \cite[Thm.3.14(c)]{DSKV13} and \cite[Cor.4.1]{DSKV16},
we get the following:
\begin{theorem}\label{thm:structure-W}
The map $\pi_{\mf g^f}$ restricts to a differential algebra isomorphism
$$
\pi:=\pi_{\mf g^f}|_{\mc W}:\,\mc W\,\stackrel{\sim}{\longrightarrow}\,\mc V(\mf g^f)
\,,
$$
hence we have the inverse differential algebra isomorphism
$$
w=:\,\mc V(\mf g^f)\,\stackrel{\sim}{\longrightarrow}\,\mc W
\,,
$$
which associates to every element $q\in\mf g^f$ the (unique) element $w(q)\in\mc W$
of the form $w(q)=q+r$, with $r\in\langle U^\perp\rangle$.
\end{theorem}
\begin{remark}\label{rem:general-U}
In \cite[Cor.4.1]{DSKV16}
the analogue of Theorem \ref{thm:structure-W}
is stated with the choice $U=\mf g^e$.
However the proof there works verbatim for every choice of a subspace $U\subset\mf g$
complementary to $[f,\mf g]$ and compatible with the grading \eqref{eq:grading}.
\end{remark}

\subsection{$\mc W$-algebras as limit of Dirac reductions}
\label{sec:3.3}

Let us briefly recall the definition of the Dirac reduction of a PVA,
following \cite{DSKV14b}.
Let $\mc V$ be a Poisson vertex algebra with PVA $\lambda$-bracket $\{\cdot\,_\lambda\,\cdot\}$.
Assume that $\mc V$ is a domain with differential field of fractions $\mc K$.
Let $\langle \theta_\alpha\rangle_{\alpha=1}^\ell\subset\mc V$
be the differential algebra ideal generated 
by elements $\theta_1,\dots,\theta_\ell\in\mc V$.
Assume that the $\ell\times\ell$ matrix differential operator
$C(\partial)=\big(C_{\alpha\beta}(\partial)\big)_{\alpha,\beta=1}^\ell$
with symbol
\begin{equation}\label{eq:dirac3}
C_{\beta\alpha}(\lambda)=\{{\theta_\alpha}_\lambda{\theta_\beta}\}
\end{equation}
is non-degenerate, i.e. it is invertible in the ring
$\Mat_{\ell\times\ell}\mc K((\partial^{-1}))$.
Then, the Dirac reduction of the PVA $\mc V$ by the constraints $\{\theta_\alpha\}_{\alpha=1}^\ell$
is the quotient differential algebra $\mc V^D=\mc V/\langle \theta_\alpha\rangle_{\alpha=1}^\ell$,
endowed with the following Dirac reduced (non-local) PVA $\lambda$-bracket,
\begin{equation}\label{eq:dirac1}
\{\pi(a)_\lambda \pi(b)\}^D
=
\pi\Big(
\{a_\lambda b\}
- \sum_{\alpha,\beta=1}^\ell
\{{\theta_\beta}_{\lambda+\partial}b\}_\to (C^{-1})_{\beta\alpha}(\lambda+\partial)\{a_\lambda\theta_\alpha\}
\Big)
\,,
\end{equation}
where $\pi:\,\mc V\to\mc V^D$ is the canonical quotient map.
For a definition of non-local $\lambda$-brackets and non-local Poisson vertex algebra, see \cite{DSK13}.

In this section we show that the $\mc W$-algebra $\mc W_\epsilon(\mf g,f,S)$
can be obtained as a limit of Dirac reductions of the affine vertex algebra 
$\mc V_\epsilon(\mf g,S)$.
Let $\{u_\alpha\}_{\alpha=1}^\ell$ be a basis of $\mf g_{\geq1}$
and let $\{u^\alpha\}_{\alpha=1}^\ell$ be the dual basis of $\mf g_{\leq-1}$.
Consider the elements, depending on the parameter $t\in\mb F$,
\begin{equation}\label{eq:dirac2}
\theta_\alpha(t)=u_\alpha-(f|u_\alpha)+\frac12tu^\alpha
\,\,,\,\,\,\,
\alpha=1,\dots,\ell
\,.
\end{equation}
Denote by $\pi_t:\,\mc V(\mf g)\to\mc V(\mf g)/\langle \theta_\alpha(t)\rangle_{\alpha=1}^\ell$
the canonical quotient map.
Note that, for $t=0$, it coincides with the map
$\rho:\,\mc V(\mf g)\to\mc V(\mf g)/\langle \theta_\alpha(0)\rangle_{\alpha=1}^\ell
\simeq\mc V(\mf g_{\leq\frac12})$
given by \eqref{rho}.

By the definition \eqref{lambda} of the $\lambda$-bracket on $\mc V_\epsilon(\mf g,S)$,
the matrix $C(\partial)$ defined by \eqref{eq:dirac3} has entries
\begin{equation}\label{eq:dirac4}
C_{\beta\alpha}(\partial)
=
t\delta_{\alpha\beta}\partial
+[\theta_\alpha(t),\theta_\beta(t)]
+\epsilon(S|[\theta_\alpha(t),\theta_\beta(t)])
\,.
\end{equation}
In particular, for every $t\neq0$,
it is a matrix differential operator of order $1$ with leading coefficient $t\id_\ell$,
and so it is invertible in $\Mat_{\ell\times\ell}\mc V(\mf g)$.
Then, we can consider the Dirac reduction of the PVA $\mc V_\epsilon(\mf g,S)$
by the constraints $\{\theta_\alpha(t)\}_{\alpha=1}^\ell$.

Note that the quotient differential algebra $\mc V(\mf g)/\langle \theta_\alpha(t)\rangle_{\alpha=1}^\ell$
is canonically isomorphic to $\mc V(\mf g_{\leq\frac12})$.
Hence, the $\mc W$-algebra $\mc W_\epsilon(\mf g,f,S)\subset\mc V(\mf g_{\leq\frac12})$, 
as a differential algebra,
is a subalgebra of 
$\mc V_\epsilon(\mf g,S)^D=\mc V(\mf g)/\langle \theta_\alpha(t)\rangle_{\alpha=1}^\ell
\simeq\mc V(\mf g_{\leq\frac12})$.
We claim that the $\mc W$-algebra $\lambda$-bracket on $\mc W_\epsilon(\mf g,f,S)$
can be obtained as the limit for $t\to0$ of the Dirac reduced $\lambda$-bracket 
on $\mc V_\epsilon(\mf g,S)^D$, restricted to $\mc W_\epsilon(\mf g,f,S)$.
Indeed, for $v,w\in\mc W_\epsilon(\mf g,f,S)$ we have,
\begin{equation}\label{eq:dirac5}
\{v_\lambda w\}_\epsilon^D
=
\pi_t\{v_\lambda w\}_\epsilon
- \sum_{\alpha,\beta=1}^\ell
\pi_t\{{\theta_\beta(t)}_{\lambda+\partial}w\}_{\epsilon,\to} 
\pi_t\big( C_{\beta\alpha}(t;\lambda+\partial) \big)
\pi_t \{v_\lambda\theta_\alpha(t)\}_\epsilon
\,.
\end{equation}
But, for $t\to0$, we have $\theta_\alpha(t)=u_\alpha-(f|u_\alpha)+O(t)$,
and therefore, by the definition \eqref{20120511:eq2} of the $\mc W$-algebra,
we have
$\pi_t\{{\theta_\beta(t)}_\lambda w\}_\epsilon
=\rho\{{u_\beta}_\lambda w\}_\epsilon+O(t)
=O(t)$,
and similarly, $\pi_t \{v_\lambda\theta_\alpha(t)\}_\epsilon=O(t)$.
On the other hand, $\pi_t\big( C_{\beta\alpha}(t;\lambda+\partial) \big)=O(\frac1t)$.
Therefore, the summation in the RHS of \eqref{eq:dirac5}
vanishes in the limit $t\to0$.
In conclusion,
\begin{equation}\label{eq:dirac6}
\{v_\lambda w\}_\epsilon^D
\stackrel{t\to0}{\longrightarrow}
\rho\{v_\lambda w\}_\epsilon
=\{v_\lambda w\}_\epsilon^{\mc W}
\,.
\end{equation}

\section{Operator of bi-Adler type for \texorpdfstring{$\mc W_\epsilon(\mf{gl}_N,f,S)$}{W\_e(gl\_N,f,S)}
and the corresponding integrable bi-Hamiltonian hierarchies}\label{sec:4}

\subsection{Setup and notation}\label{sec:4.1}

We fix a convenient basis of $\mf g=\mf{gl}_N$,
associated to a nilpotent element $f\in\mf g$ and the corresponding Dynkin grading.

Let $p=(p_1,\dots,p_r)$, with $p_1\geq\dots\geq p_r>0$, be a partition of $N$.
We associate to it a symmetric (with respect to the $y$-axis) pyramid,
with boxes indexed by $(i,h)$ in the set (of cardinality $N$)
\begin{equation}\label{eq:J}
\mc J=\big\{(i,h)\in\mb Z_+^2\,\big|\,1\leq i\leq r,\,1\leq h\leq p_i\big\}
\,,
\end{equation}
with $i$ and $h$ being respectively the row index (starting from the bottom) 
and the column index (starting from the right).
For example, for the partition $(9,7,4,4)$ of $24$, we have the pyramid in Figure \ref{fig:pyramid}.

\begin{figure}[h]
\setlength{\unitlength}{0.14in}
\centering
\begin{picture}(30,12)
\put(5,4){\framebox(2,2){(19)}}
\put(7,4){\framebox(2,2){(18)}}
\put(9,4){\framebox(2,2){(17)}}
\put(11,4){\framebox(2,2){(16)}}
\put(13,4){\framebox(2,2){(15)}}
\put(15,4){\framebox(2,2){(14)}}
\put(17,4){\framebox(2,2){(13)}}
\put(19,4){\framebox(2,2){(12)}}
\put(21,4){\framebox(2,2){(11)}}

\put(7,6){\framebox(2,2){(27)}}
\put(9,6){\framebox(2,2){(26)}}
\put(11,6){\framebox(2,2){(25)}}
\put(13,6){\framebox(2,2){(24)}}
\put(15,6){\framebox(2,2){(23)}}
\put(17,6){\framebox(2,2){(22)}}
\put(19,6){\framebox(2,2){(21)}}

\put(10,8){\framebox(2,2){(34)}}
\put(12,8){\framebox(2,2){(33)}}
\put(14,8){\framebox(2,2){(32)}}
\put(16,8){\framebox(2,2){(31)}}


\put(10,10){\framebox(2,2){(44)}}
\put(12,10){\framebox(2,2){(43)}}
\put(14,10){\framebox(2,2){(42)}}
\put(16,10){\framebox(2,2){(41)}}

\put(4,2){\vector(1,0){22}}
\put(26,1){$x$}

\put(6,1.6){\line(0,1){0.8}}
\put(7,1.8){\line(0,1){0.4}}
\put(8,1.6){\line(0,1){0.8}}
\put(9,1.8){\line(0,1){0.4}}
\put(10,1.6){\line(0,1){0.8}}
\put(11,1.8){\line(0,1){0.4}}
\put(12,1.6){\line(0,1){0.8}}
\put(13,1.8){\line(0,1){0.4}}
\put(14,1.6){\line(0,1){0.8}}
\put(15,1.8){\line(0,1){0.4}}
\put(16,1.6){\line(0,1){0.8}}
\put(17,1.8){\line(0,1){0.4}}
\put(18,1.6){\line(0,1){0.8}}
\put(19,1.8){\line(0,1){0.4}}
\put(20,1.6){\line(0,1){0.8}}
\put(21,1.8){\line(0,1){0.4}}
\put(22,1.6){\line(0,1){0.8}}

\put(13.8,0.6){0}
\put(15.8,0.6){1}
\put(17.8,0.6){2}
\put(19.8,0.6){3}
\put(21.8,0.6){4}
\put(11.6,0.6){-1}
\put(9.6,0.6){-2}
\put(7.6,0.6){-3}
\put(5.6,0.6){-4}

\put(14.8,-0.3){\tiny{$\frac12$}}
\put(16.8,-0.3){\tiny{$\frac32$}}
\put(18.8,-0.3){\tiny{$\frac52$}}
\put(20.8,-0.3){\tiny{$\frac72$}}
\put(12.6,-0.3){\tiny{$-\frac12$}}
\put(10.6,-0.3){\tiny{$-\frac32$}}
\put(8.6,-0.3){\tiny{$-\frac52$}}
\put(6.6,-0.3){\tiny{$-\frac72$}}

\end{picture}
\caption{} 
\label{fig:pyramid}
\end{figure}

\noindent
We also let $r_1$ be the number or rows of maximal length $p_1$
(i.e. the multiplicity of $p_1$ in the partition $p$).
For example in Figure \ref{fig:pyramid} we have $r_1=1$.

Let $V$ be the $N$-dimensional vector space over $\mb F$ with basis 
$\{e_{ih}\}_{(i,h)\in\mc J}$.
The Lie algebra $\mf g=\mf{gl}(V)$ has a basis 
consisting of the elementary matrices $E_{(ih),(jk)}$, $(ih),(jk)\in\mc J$.
The nilpotent element $f\in\mf g$ associated to the partition $p$
is the ``shift'' operator: 
$f(e_{ih})=e_{i,h+1}$, for $h<p_i$, and $f(e_{i,p_i})=0$.
In terms of elementary matrices, 
\begin{equation}\label{eq:f}
f=\sum_{(ih)\in\mc J\,|\,h<p_i}E_{(i,h+1),(ih)}
\,.
\end{equation}
If we order the indices $(ih)$ lexicographically,
$f$ has Jordan form with nilpotent Jordan blocks of sizes $p_1,\dots,p_r$.
The elementary matrix $E_{(jk),(ih)}$ in $\mf g$ can be depicted by an arrow
going from the center of the box $(ih)$ to the center of the box $(jk)$.
In particular, $f$ is depicted as the sum of all the arrows pointing
from each box to the next one on the left.

We also let $x\in\mf g$ be the diagonal endomorphism of $V$
whose eigenvalue on $e_{ih}$ is $\frac12(p_i+1-2h)$,
i.e. the $1$-st coordinate of the center of the corresponding box (see Figure \ref{fig:pyramid}).
It follows that the elementary matrices $E_{(ih),(jk)}$ are eigenvectors
with respect to the adjoint action of $x$:
\begin{equation}\label{eq:adx}
(\ad x)E_{(ih),(jk)}=\Big(\frac12(p_i-p_j)-(h-k)\Big)E_{(ih),(jk)}
\,.
\end{equation}
This defines a $\frac12\mb Z$-gradation of $\mf g$,
given by the $\ad x$-eigenvalues as in \eqref{eq:grading}.
The depth of this gradation is $d=p_1-1$.

\subsection{Canonical factorization in $\mf g_d$}\label{sec:4.1b}

Note that the $\ad x$-eigenspace of maximal degree is
\begin{equation}\label{eq:gd}
\mf g_d
=
\Span_{\mb F}
\big\{
E_{(i1),(jp_1)}\,\big|\,
i,j=1,\dots,r_1
\big\}
\,.
\end{equation}
We have a natural bijection $\mf g_d\stackrel{\sim}{\longrightarrow}\Mat_{r_1\times r_1}\mb F$
given by
\begin{equation}\label{eq:s0}
S=\sum_{i,j=1}^{r_1}
s_{ij}E_{(i1),(jp_1)}
\mapsto
\bar S=\big(s_{ij}\big)_{i,j=1}^{r_1}
\,.
\end{equation}
For example, the element of $\mf g_d$, corresponding to $\id_{r_1}\in\Mat_{r_1\times r_1}\mb F$,
is the matrix
\begin{equation}\label{eq:s1}
S_1=\sum_{i=1}^{r_1}
E_{(i1),(ip_1)}
\,.
\end{equation}

Let $S\in\mf g_d$, and let $\bar r\leq r_1$ be its rank.
The following proposition gives
an explicit description of its canonical factorization $S=IJ$,
where $I\in\Mat_{N\times\bar r}\mb F$ and $J\in\Mat_{\bar r\times N}\mb F$
are the matrices associated (in some basis of $\im(S)$)
to the maps $I:\,\im(S)\hookrightarrow\mb F^N$, $X\mapsto X$,
and $J:\,\mb F^N\twoheadrightarrow\im(S)$, $X\mapsto S(X)$. 
\begin{proposition}\phantomsection\label{prop:factor}
\begin{enumerate}[(a)]
\item
The canonical factorization of the matrix $S_1$ in \eqref{eq:s1} is
$S_1=I_1J_1$, where
\begin{equation}\label{eq:factor1}
I_1=\sum_{i=1}^{r_1}
E_{(i1),i}
\,\in\Mat_{N\times r_1}\mb F
\,\,,\,\,\,\,
J_1=\sum_{i=1}^{r_1}
E_{i,(ip_1)}
\,\in\Mat_{r_1\times N}\mb F
\,.
\end{equation}
\item
If $\bar S=\bar I\bar J$, with $\bar I\in\Mat_{r_1\times\bar r}\mb F$ and $\bar J\in\Mat_{\bar r\times r_1}\mb F$,
is the canonical factorization of $\bar S\in\Mat_{r_1\times r_1}\mb F$
(where $\bar r\leq r_1$ is the rank of $\bar S$),
then the canonical factorization of the element $S\in\mf g_d$
corresponding to $\bar S$ via \eqref{eq:s0} is $S=IJ$,
where
\begin{equation}\label{eq:factor2}
I=I_1\bar I\in\Mat_{N\times\bar r}\mb F
\,\,,\,\,
J=\bar JJ_1\in\Mat_{\bar r\times N}\mb F
\,.
\end{equation}
\end{enumerate}
\end{proposition}
\begin{proof}
Clearly, $I_1$ and $J_1$ in \eqref{eq:factor1} are rectangular matrices of maximal rank $r_1$,
and it is immediate to check that $I_1J_1=S_1$. This proves part (a).
Moreover, if $I$ and $J$ are as in \eqref{eq:factor2}, then
$$
IJ
=
I_1\bar I\bar JJ_1
=
I_1\bar SJ_1
=
\sum_{i,j=1}^{r_1}
E_{(i1),i}
\bar S
E_{j,(jp_1)}
=
\sum_{i,j=1}^{r_1}
s_{ij}E_{(i1),(jp_1)}
=S\,,
$$
proving (b).
\end{proof}

\subsection{Operator of bi-Adler type for $\mc W_\epsilon(\mf g,f,S)$, $\epsilon\in\mb F$,
and corresponding bi-Hamiltonian integrable hierarchy}
\label{sec:4.2}

Let $\mf g=\mf{gl}_N$. 
Consider the pencil of affine Poisson vertex algebras $\mc V=\mc V_\epsilon(\mf g,S)$, $\epsilon\in\mb F$,
defined in Example \ref{ex:A},
depending on the matrix $S\in\Mat_{N\times N}\mb F$.
Recall that the matrix differential operator $A(\partial)\in\Mat_{N\times N}\mc V[\partial]$ in \eqref{eq:A},
is of $S$-Adler type with respect to the bi-PVA structure of $\mc V_\epsilon(\mf g,S)$,
for every $S\in\Mat_{N\times N}\mb F$.

We want to find an analogous operator for the affine $\mc W$-algebras.
Fix a non-zero nilpotent element $f\in\mf g$,
associated to the partition $p_1\geq p_2\geq\dots\geq p_r>0$ of $N$,
and consider the corresponding pencil of $\mc W$-algebras $\mc W_\epsilon(\mf g,f,S)$,
$\epsilon\in\mb F$,
depending on $S\in\mf g_d$,
where $d=p_1-1$ is the depth of the gradation \eqref{eq:grading}.
We will construct a matrix pseudodifferential operator 
$L_1(\partial)\in\Mat_{r_1\times r_1}\mc W((\partial^{-1}))$,
where $r_1$ is the multiplicity of $p_1$ in the partition,
which is of $\bar S$-Adler type with respect to the bi-PVA structure of $\mc W_\epsilon(\mf g,f,S)$,
where $S\in\mf g_d$ and $\bar S\in\Mat_{r_1\times r_1}\mb F$ are related by \eqref{eq:s0}.

The operator $L_1(\partial)$ is constructed as follows:
\begin{equation}\label{eq:La}
L_1(\partial)
=
|\rho(A(\partial))|_{I_1J_1}
=
|\id_N\partial+\rho(Q)|_{I_1J_1}
\,,
\end{equation}
where 
$I_1$ and $J_1$ are the matrices \eqref{eq:factor1},
given by the canonical factorization of $S_1\in\mf g_d$,
and $\rho:\,\mc V(\mf g)\to\mc V(\mf g_{\leq\frac12})$ is the map defined by \eqref{rho}.

In Section \ref{sec:4.2b} 
we prove that the generalized quasideterminant \eqref{eq:La} exists.
In fact, we show that,
for every $\bar S\in\Mat_{r_1\times r_1}\mb F$ and its canonical factorization $\bar S=\bar I\bar J$,
the generalized quasideterminant $|L_1(\partial)|_{\bar I\bar J}$ exists
over the field of fractions of $\mc V(\mf g_{\leq\frac12})$.
In Section \ref{sec:4.3}
we show that the entries of the coefficients of $L_1(\partial)$ actually lie in the $\mc W$-algebra 
$\mc W_\epsilon(\mf g,f,S)$.
Finally, in Section \ref{sec:4.4} we prove that,
if $S\in\mf g_d$ and $\bar S\in\Mat_{r_1\times r_1}\mb F$ are related by \eqref{eq:s0},
then $L_1(\partial)$ is a matrix pseudodifferential operator
of $\bar S$-Adler type with respect to the bi-PVA structure of $\mc W_\epsilon(\mf g,f,S)$, $\epsilon\in\mb F$.

\bigskip

Using the above stated results and following the scheme described in Section \ref{sec:2.5},
we will be able to construct a bi-Hamiltonian integrable hierarchy 
for the bi-Poisson structure of $\mc W_\epsilon(\mf{gl}_N,f,S)$,
for every nilpotent element $f\in\mf{gl}_N$ and every non-zero element $S\in\mf g_d$.
Such a hierarchy was constructed by Drinfeld and Sokolov \cite{DS85} 
for a principal nilpotent element $f$ of an arbitrary simple Lie algebra $\mf g$,
and their argument was generalized in different directions by many authors,
all under restrictive assumptions on the nilpotent element $f$
\cite{dGHM92,FHM93,BdGHM93,DF95,FGMS95,FGMS96,DSKV13,DSKV14a}.

Our idea is very simple.
Take the matrix $\bar S\in\Mat_{r_1\times r_1}\mb F$ corresponding to $S\in\mf g_d$
via \eqref{eq:s0},
take its canonical factorization $\bar S=\bar I\bar J$,
with $\bar I\in\Mat_{r_1\times\bar r}\mb F$ and $\bar J\in\Mat_{\bar r\times r_1}\mb F$
($\bar r$ is the rank of $S$ or, equivalently, of $\bar S$),
and consider the generalized quasideterminant
\begin{equation}\label{eq:Lb}
L(\partial)
=
|L_1(\partial)|_{\bar I\bar J}
=
|\rho(A(\partial))|_{IJ}
\,.
\end{equation}
The second equality holds, for $I=I_1\bar I$ and $J=\bar JJ_1$,
by the hereditary property \eqref{eq:hereditary} of generalized quasideterminants
and by Proposition \ref{prop:factor}(b).
By the results of Section \ref{sec:4.2b}, the generalized quasideterminant \eqref{eq:Lb} exists
and, by the results of Sections \ref{sec:4.3} and \ref{sec:4.4}, $L(\partial)$ is a matrix pseudodifferential operator
with coefficients in the field of fractions $\mc K$ of $\mc W\,\big(=\mc W(\mf g,f)\big)$,
of bi-Adler type with respect to the bi-Poisson structure of $\mc W_\epsilon(\mf g,f,S)$, $\epsilon\in\mb F$.
Hence, following the scheme described in Section \ref{sec:2.5},
we get that the Hamiltonian densities (cf. \eqref{eq:hn})
\begin{equation}\label{eq:hn-W}
h_n
=
-\frac1n\tr\Res_\partial L(\partial)^n
\,\,,\,\,\,\,n\geq1\,,
\end{equation}
are in involution with respect to both $\lambda$-brackets 
$\{\cdot\,_\lambda\,\cdot\}_0^{\mc W}$ and $\{\cdot\,_\lambda\,\cdot\}_1^{\mc W}$
of the bi-PVA $\mc W_\epsilon(\mf g,f,S)$, $\epsilon\in\mb F$,
they satisfy the Lenard-Magri recurrence relation $\{\tint h_n,w\}_0=\{\tint h_{n+1},w\}_1$
in the bi-PVA subalgebra $\mc W_1\subset\mc K$ generated by the coefficients of $L(\partial)$,
and they thus define an integrable hierarchy of (bi)Hamiltonian equations ($w\in\mc W$)
\begin{equation}\label{eq:hier-W}
\frac{dw}{dt_n}
=
\{\tint h_n,w\}_0^{\mc W}
\,\,\Big(
=
\{\tint h_{n+1},w\}_1^{\mc W}
\,\,\text{ if } w\in\mc W_1
\Big)
\,.
\end{equation}
More generally,
by Theorem \ref{thm:hn},
we have a bigger family of integrals of motion in involution $\{\tint h_{n,B}\}$,
parametrized by $n\in\mb Z_+$
and all possible roots $B\in\Mat_{\bar r\times\bar r}\mc K(\mf g,f)((\partial^{-1}))$ of $L(\partial)$.

In Section \ref{sec:6}
we will compute explicitly the matrix pseudodifferential operator $L_1(\partial)$.
This will be used in two ways:
to find an explicit formula for the generators of the $\mc W$-algebra $\mc W(\mf g,f)$,
and to compute explicitly, in Section \ref{sec:7}, the hierarchies of bi-Hamiltonian equations
for the $\mc W$-algebra $\mc W_\epsilon(\mf g,f,S)$, for various choices 
of the nilpotent element $f\in\mf{gl}_N$ and the element $S\in\mf g_d$.

\subsection{$L(\partial)$ exists}\label{sec:4.2b}

\begin{theorem}\label{thm:L1}
Let, as before, $\mc V(\mf g_{\leq\frac12})$ be the algebra of differential polynomials 
over $\mf g_{\leq\frac12}$,
and let $\mc K(\mf g_{\leq\frac12})$ be its field of fractions.
\begin{enumerate}[(a)]
\item
The matrix differential operator $\rho(A(\partial))=\id_N\partial+\rho(Q)$ in \eqref{eq:A} is invertible
in $\Mat_{N\times N}\mc V(\mf g_{\leq\frac12})((\partial^{-1}))$.
\item
Let $S_1$ be as in \eqref{eq:s1}, 
with its canonical factorization $S_1=I_1J_1$ defined by \eqref{eq:factor1}.
The matrix
$J_1(\id_N\partial+\rho(Q))^{-1}I_1
\in\Mat_{r_1\times r_1}\mc V(\mf g_{\leq\frac12})((\partial^{-1}))$
has an expansion
\begin{equation}\label{eq:thmL1-1}
J_1(\id_N\partial+\rho(Q))^{-1}I_1
=
(-1)^{p_1-1}\id_{r_1}\partial^{-p_1}
+(-1)^{p_1}\overline{Q}\partial^{-p_1-1}
+\dots
\,,
\end{equation}
where $\overline{Q}$ is a generic $r_1\times r_1$ matrix with entries in $\mc V(\mf g_{\leq\frac12})$.
\item
The generalized quasideterminant $L_1(\partial)=|\id_N\partial+\rho(Q)|_{I_1J_1}$ exists,
and it lies in the algebra $\Mat_{r_1\times r_1}\mc V(\mf g_{\leq\frac12})((\partial^{-1}))$.
\item
Let $S\in\mf g_d$ be a non-zero element of rank $\bar r\leq r_1$,
and let $S=IJ$ be its canonical factorization,
with $I\in\Mat_{N\times\bar r}\mb F$ and $J\in\Mat_{\bar r\times N}\mb F$.
Then
$J(\id_N\partial+\rho(Q))^{-1}I$ is invertible in
$\Mat_{\bar r\times\bar r}\mc K(\mf g_{\leq\frac12})$.
In particular, the generalized quasideterminant $L(\partial)=|\id_N\partial+\rho(Q)|_{IJ}$ exists
for every non-zero $S\in\mf g_d$,
and it lies in $\Mat_{\bar r\times\bar r}\mc K(\mf g_{\leq\frac12})((\partial^{-1}))$.
\end{enumerate}
\end{theorem}
\begin{proof}
The matrix differential operator $\rho(A(\partial))$
is of order one with leading coefficient $\id_N$.
Hence it is invertible in the algebra $\Mat_{N\times N}\mc V(\mf g_{\leq\frac12})((\partial^{-1}))$,
and its inverse can be computed as geometric series expansion:
\begin{equation}\label{eq:thm1-pr1}
\rho(A(\partial))^{-1}
=
\sum_{\ell=0}^\infty (-1)^\ell \big(\partial^{-1}\circ\rho(Q)\big)^\ell\partial^{-1}
\,.
\end{equation}
This proves part (a). 
Next, we prove part (b).
By the definition \eqref{eq:factor1} of the matrices $I_1$ and $J_1$,
$J_1(\id_N\partial+\rho(Q))^{-1}I_1$ is an $r_1\times r_1$-matrix
with entry in row $i$ and column $j$ (with $i,j\in\{1,\dots,r_1\}$)
given by
\begin{equation}\label{eq:thm1-pr2}
\begin{split}
& \big(J_1(\id_N\partial+\rho(Q))^{-1}I_1\big)_{ij}
=
\sum_{\ell=0}^\infty (-1)^\ell 
\Big(
\partial^{-1}\circ\rho(Q)\dots\partial^{-1}\circ\rho(Q)\partial^{-1}
\Big)_{(ip_1),(j1)} \\
& =
\sum_{\ell=0}^\infty (-1)^\ell 
\sum
_{
\substack{
(i_0h_0),(i_1h_1),\dots,(i_\ell h_\ell)\in\mc J \\
(i_0h_0)=(ip_1),\,(i_\ell h_\ell)=(j1)
}} \\
& \qquad
\partial^{-1}\circ
\rho(q_{(i_1h_1),(i_0h_0)})
\partial^{-1}\circ
\rho(q_{(i_2h_2),(i_1h_1)})
\dots
\partial^{-1}\circ
\rho(q_{(i_\ell h_\ell),(i_{\ell-1}h_{\ell-1})})
\partial^{-1}
\,.
\end{split}
\end{equation}
Let $x_\alpha=\frac12(p_{i_\alpha}+1-2h_\alpha)\in\frac12\mb Z$, $\alpha=0,1,\dots,\ell$.
In particular,
\begin{equation}\label{eq:thm1-pr4}
x_0=-\frac12(p_1-1)=-\frac12 d 
\,\,,\,\,\,\, 
x_\ell=\frac12(p_1-1)=\frac12 d
\,.
\end{equation}
By the definition \eqref{rho} of the map $\rho$,
the summand in the RHS of \eqref{eq:thm1-pr2} vanishes
unless the indices $(i_0h_0),(i_1h_1),\dots,(i_\ell h_\ell)\in\mc J$
satisfy the following conditions:
\begin{equation}\label{eq:thm1-pr3}
x_1-x_0\leq1
\,\,,\,\,\,\,
x_2-x_1\leq1
\,,\,\,\dots\,\,,\,
x_\ell-x_{\ell-1}\leq1 
\,.
\end{equation}
Moreover, 
by the definition \eqref{eq:f} of $f$,
\begin{equation}\label{eq:thm1-pr5}
\text{ if } 
x_\alpha=x_{\alpha-1}+1
\,\,,\,\,\,\,\text{ then }\,
i_\alpha=i_{\alpha-1}
\,\text{ and }\,
h_\alpha=h_{\alpha-1}-1
\,.
\end{equation}
Clearly, from \eqref{eq:thm1-pr4} and \eqref{eq:thm1-pr3}
we get that necessarily $\ell\geq p_1-1$.
Moreover, by \eqref{eq:thm1-pr5} we also have that if $\ell=p_1-1$ 
then necessarily 
$$
i_0=i_1=\dots=i_{p_1-1}
\,\text{ and }\,
h_0=p_1,\, h_1=p_1-1,\dots, h_{p_1-1}=1
\,,
$$
and, in this case,
$$
\rho(q_{(i_1h_1),(i_0h_0)})
=\rho(q_{(i_2h_2),(i_1h_1)})
=\dots
=\rho(q_{(i_\ell h_\ell),(i_{\ell-1}h_{\ell-1})})=1
\,.
$$
This proves that
the pseudodifferential operator $\big(J_1(\id_N\partial+\rho(Q))^{-1}I_1\big)_{ij}$
has order $\leq-p_1$,
and the coefficient of $\partial^{-p_1}$ is $(-1)^{p_1-1}\delta_{ij}$.
In order to prove (b),
we are left to prove that the coefficients $\overline{Q}_{ij}$ of $(-1)^{p_1}\partial^{-p_1-1}$
in $\big(J_1(\id_N\partial+\rho(Q))^{-1}I_1\big)_{ij}$
form a generic matrix $\overline{Q}$ (according to Definition \ref{def:generic}).
By the above observations, the only contributions to $\overline{Q}_{ij}$
come from the term with $\ell=p_1$ in the RHS of \eqref{eq:thm1-pr2}:
\begin{equation}\label{eq:thm1-pr6}
\begin{split}
\overline{Q}_{ij}
=
\!\!\!\!\!\!
\sum
_{\substack{
(i_0h_0),\dots,(i_{p_1}h_{p_1})\in\mc J \\
(i_0h_0)=(ip_1),(i_{p_1}\!h_{p_1})=(j1)
}}
\!\!\!\!\!\!
\rho(q_{(i_1h_1),(i_0h_0)})
\rho(q_{(i_2h_2),(i_1h_1)})
\dots
\rho(q_{(i_{p_1} h_{p_1}),(i_{p_1-1}h_{p_1-1})})
\,.
\end{split}
\end{equation}
There are only two types of contributions to the RHS of \eqref{eq:thm1-pr6}:

\noindent\emph{Type 1.}
The terms with
$$
x_\alpha
=
\left\{\begin{array}{l}
\vphantom{\Big(}
-\frac12 d+\alpha
\,,\,\text{ for } \alpha=0,\dots,s\, \\
\vphantom{\Big(}
-\frac12 d+\alpha-1
\,,\,\text{ for } \alpha=s+1,\dots,p_1\,
\end{array}
\right.
$$
for some $s=0,\dots,p_1-1$.
In this case we have,
by \eqref{eq:thm1-pr4}, \eqref{eq:thm1-pr3} and \eqref{eq:thm1-pr5},
$$
(i_\alpha,h_\alpha)
=
\left\{\begin{array}{l}
\vphantom{\Big(}
(i,p_1-\alpha)
\,,\,\text{ for } \alpha=0,\dots,s\, \\
\vphantom{\Big(}
(j,p_1+1-\alpha)
\,,\,\text{ for } \alpha=s+1,\dots,p_1\,
\end{array}
\right.
$$
so that, by the definition \eqref{rho} of the map $\rho$,
$$
\rho(q_{(i_\alpha h_\alpha),(i_{\alpha-1}h_{\alpha-1})})=1
\,\text{ for }\,\alpha\neq s+1
\,,
$$
and
$$
\rho(q_{(i_{s+1}h_{s+1}),(i_sh_s)})=q_{(j,p_1-s),(i,p_1-s)}\in\mf g_0
\,.
$$

\noindent\emph{Type 2.}
The terms with
$$
x_\alpha
=
\left\{\begin{array}{l}
\vphantom{\Big(}
-\frac12 d+\alpha
\,,\,\text{ for } \alpha=0,\dots,s\, \\
\vphantom{\Big(}
-\frac12 d+\alpha-\frac12
\,,\,\text{ for } \alpha=s+1,\dots,t\, \\
\vphantom{\Big(}
-\frac12 d+\alpha-1
\,,\,\text{ for } \alpha=t+1,\dots,p_1\,
\end{array}
\right.
$$
for some $0\leq s<t< p_1$.
In this case we have,
by \eqref{eq:thm1-pr4}, \eqref{eq:thm1-pr3} and \eqref{eq:thm1-pr5},
\begin{equation*}
(i_\alpha,h_\alpha)
=
\left\{\begin{array}{l}
\vphantom{\Big(}
(i,p_1-\alpha)
\,,\,\text{ for } \alpha=0,\dots,s\, \\
\vphantom{\Big(}
(\tilde{i},\tilde{h}+1+s-\alpha)
\,,\,\text{ for } \alpha=s+1,\dots,t\, \\
\vphantom{\Big(}
(j,p_1+1-\alpha)
\,,\,\text{ for } \alpha=t+1,\dots,p_1\,
\end{array}
\right.
\end{equation*}
for some $(\tilde{i},\tilde{h})\in\mc J$ such that
\begin{equation}\label{eq:thm1-pr8}
p_{\tilde{i}}-2\tilde{h}=-p_1+1+2s
\,.
\end{equation}
Hence, by the definition \eqref{rho} of the map $\rho$, we have
$$
\rho(q_{(i_\alpha h_\alpha),(i_{\alpha-1}h_{\alpha-1})})=1
\,\text{ for }\,\alpha\not\in\{s+1,t+1\}
\,,
$$
and
$$
\rho(q_{(i_{s+1}h_{s+1}),(i_sh_s)})=q_{(i,p_1-s),(\tilde{i},\tilde{h})}
\,,\,\,
\rho(q_{(i_{t+1}h_{t+1}),(i_th_t)})=q_{(j,p_1-t),(\tilde{i},\tilde{h}+s+1-t)}
\,\in\mf g_{\frac12}
\,.
$$

It follows that
\begin{equation}\label{eq:thm1-pr7}
\begin{split}
\overline{Q}_{ij}
=
\sum_{s=0}^{p_1-1}q_{(j,p_1-s),(i,p_1-s)}
+
\sum_{0\leq s<t\leq p_1}
\!\!\!\!\!\!
\sum_{\substack{
(\tilde{i},\tilde{h})\in\mc J \\
\text{ s.t. } \eqref{eq:thm1-pr8} \text{ holds}
}}
\!\!\!\!\!\!
q_{(i,p_1-s),(\tilde{i},\tilde{h})}
q_{(j,p_1-t),(\tilde{i},\tilde{h}+s+1-t)}
\,.
\end{split}
\end{equation}
We then observe that the matrix $\overline{Q}$ in \eqref{eq:thm1-pr7}
is generic since,
for example, by letting all the variables in $\mf g_{\frac12}$ equal to $0$
and all the variables $q_{(jk),(ih)}$ with $h\neq 1$ equal to $0$,
we are left with the matrix $\big(q_{(j1),(i1)}\big)_{i,j=1}^{r_1}$,
which is clearly generic.

Part (c) follows from part (b) by taking geometric series expansion 
of $L_1(\partial)=(J_1(\id_N\partial+\rho(Q))^{-1}I_1)^{-1}$
using \eqref{eq:thmL1-1},
and part (d) is an immediate consequence of part (b) and
Propositions \ref{prop:generic} and \ref{prop:factor}(b).
\end{proof}

\subsection{$L_1(\partial)$ has coefficients with entries in $\mc W$}\label{sec:4.3}

The following key result is the only one which requires quite involved computations.
\begin{theorem}\label{prop:L2}
Consider the matrix pseudodifferential operator
\begin{equation}\label{eq:L2-1}
L_1^{-1}(\partial)=J_1(\id_N\partial+\rho(Q))^{-1}I_1
\,\in\Mat_{r_1\times r_1}\mc V(\mf g_{\leq\frac12})((\partial^{-1}))
\,,
\end{equation}
where $I_1,J_1$ are as in \eqref{eq:factor1} 
and $\rho$ is defined by \eqref{rho}.
Then, 
\begin{equation}\label{eq:L2-2}
\rho\{a_\lambda L_1^{-1}(z)_{ij}\}_\epsilon=0
\,\,\text{ for every }\,\,
i,j\in\{1,\dots,r_1\}
\,\,\text{ and }\,\,
a\in\mf g_{\geq\frac12}
\,.
\end{equation}
In particular, the entries of the coefficients of $L_1^{-1}(\partial)$ lie in
the differential algebra $\mc W(\mf g,f)\subset\mc V(\mf g_{\leq\frac12})$ 
(defined in \eqref{20120511:eq2}).
\end{theorem}
\begin{proof}
As in \eqref{eq:thm1-pr2}, we can expand $L_1^{-1}(\partial)$ in geometric series.
Recalling that 
$f_{(jk),(ih)}=\delta_{(jk),(i,h+1)}$,
we get
\begin{equation}\label{eq:L2-pr1}
\begin{split}
& L_1^{-1}(\partial)_{ij} 
=
\sum_{\ell=0}^\infty (-1)^\ell 
\Big(\big(\partial^{-1}\circ(f+\pi_{\leq\frac12}Q)\big)^\ell
\partial^{-1}
\Big)_{(ip_1),(j1)} 
\\
& =
\sum_{\ell=0}^\infty (-1)^\ell 
\!\!\!\!\!\!
\sum
_{(i_0h_0),(i_1h_1),\dots,(i_\ell h_\ell)\in\mc J} 
\!\!\!\!\!\!
\delta_{(i_0h_0)(ip_1)}\delta_{(i_\ell h_\ell)(j1)}
\\
& \qquad
\partial^{-1}\circ
(\delta_{(i_1h_1),(i_0,h_0-1)}+\pi_{\leq\frac12}q_{(i_1h_1),(i_0h_0)})
\\
& \qquad
\partial^{-1}\circ
(\delta_{(i_2h_2),(i_1,h_1-1)}+\pi_{\leq\frac12}q_{(i_2h_2),(i_1h_1)})
\dots \\
& \qquad
\dots
\partial^{-1}\circ
(\delta_{(i_\ell h_\ell),(i_{\ell-1},h_{\ell-1}-1)}
+\pi_{\leq\frac12}q_{(i_\ell h_\ell),(i_{\ell-1}h_{\ell-1})})
\partial^{-1}
\,.
\end{split}
\end{equation}
By grouping the terms with the same number of factors $\pi_{\leq\frac12}q$, 
we can rewrite equation \eqref{eq:L2-pr1} as
\begin{equation}\label{eq:L2-pr2}
L_1^{-1}(\partial)_{ij} 
=
\delta_{ij}(-1)^{p_1-1}\partial^{-p_1}+
\sum_{s=1}^\infty X_{s;ij}(\partial) 
\,,
\end{equation}
where
\begin{equation}\label{eq:L2-pr3}
\begin{split}
& 
X_{s;ij}(\partial)
=
\sum_{n_0,n_1,\dots,n_s=0}^\infty
(-1)^{n_0+n_1+\dots+n_s+s}
\!\!\!\!\!\!\!\!\!\!\!\!
\sum
_{
(i_0h_0),(i_1h_1),\dots,(i_s h_s)\in\mc J
}
\!\!\!\!\!\!
\delta_{(i_0h_0)(ip_1)}\delta_{(i_s,h_s-n_s)(j1)}
\\
&\vphantom{\Big(} \qquad
\partial^{-n_0-1}\circ\pi_{\leq\frac12}q_{(i_1h_1),(i_0,h_0-n_0)}
\,\,
\partial^{-n_1-1}\circ\pi_{\leq\frac12}q_{(i_2h_2),(i_1,h_1-n_1)}
\dots
\\
&\vphantom{\Big(} \qquad 
\dots
\partial^{-n_{s-1}-1}\circ\pi_{\leq\frac12}q_{(i_sh_s),(i_{s-1},h_{s-1}-n_{s-1})}
\partial^{-n_s-1}
\\
&\vphantom{\Big(} \qquad
=
\sum_{n_0,n_1,\dots,n_s=0}^\infty
(-1)^{n_0+n_1+\dots+n_s+s}
\!\!\!\!\!\!\!\!\!\!\!\!
\sum
_{
(i_0h_0),(i_1h_1),\dots,(i_s h_s)\in\mc J
}
\!\!\!\!\!\!
\delta_{(i_0h_0)(ip_1)}\delta_{(i_s,h_s-n_s)(j1)}
\\
&\vphantom{\Big(} \qquad\qquad\quad
\prod_{r=1}^s
\Big(
\partial^{-n_{r-1}-1}\circ\pi_{\leq\frac12}q_{(i_rh_r),(i_{r-1},h_{r-1}-n_{r-1})}
\Big)
\partial^{-n_s-1}
\,.
\end{split}
\end{equation}
In order to prove \eqref{eq:L2-2},
we need to compute $\rho\{a_\lambda X_{s;ij}(z)\}_\epsilon$ for $a\in\mf g_{\geq\frac12}$.
Recall that, for $a\in\mf g_{\geq\frac12}$ and $q\in\mf g$, we have
$$
\rho\{a_\lambda\pi_{\leq\frac12}q\}_\epsilon
=
\rho\{a_\lambda q\}_\epsilon
=
\pi_{\leq\frac12}[a,q]
+(a|q)\lambda+(f|[a,q])
\,.
$$
Hence, 
by \eqref{eq:L2-pr3}, the sesquilinearity axioms and the Leibniz rules, we have
\begin{equation}\label{eq:L2-pr4}
\rho\{a_\lambda X_{s;ij}(z)\}_\epsilon
=
\sum_{\ell=1}^s(Y_{s\ell;ij}(z)+Z_{s,\ell;ij}(z))
\,,
\end{equation}
where
\begin{equation}\label{eq:L2-pr5}
\begin{split}
& 
Y_{s\ell;ij}(z)
=
\sum_{n_0,\dots,n_s=0}^\infty
(-1)^{n_0+\dots+n_s+s}
\!\!\!\!\!\!\!\!\!\!\!\!
\sum
_{
(i_0h_0),(i_1h_1),\dots,(i_s h_s)\in\mc J
}
\!\!\!\!\!\!
\delta_{(i_0h_0)(ip_1)}\delta_{(i_s,h_s-n_s)(j1)}
\\
&\vphantom{\Big(} \qquad
\prod_{r=1}^{\ell-1}
\Big(
(z+\lambda+\partial)^{-n_{r-1}-1}\pi_{\leq\frac12}q_{(i_rh_r),(i_{r-1},h_{r-1}-n_{r-1})}
\Big)
\\
&\vphantom{\Big(} \qquad
(z+\lambda+\partial)^{-n_{\ell-1}-1}\pi_{\leq\frac12}[a,q_{(i_\ell h_\ell),(i_{\ell-1},h_{\ell-1}-n_{\ell-1})}]
\\
&\vphantom{\Big(} \qquad
\prod_{r=\ell+1}^{s}
\Big(
(z+\partial)^{-n_{r-1}-1}\pi_{\leq\frac12}q_{(i_r h_r),(i_{r-1},h_{r-1}-n_{r-1})}
\Big)
(z+\partial)^{-n_s-1}
\,,
\end{split}
\end{equation}
and
\begin{equation}\label{eq:L2-pr6}
\begin{split}
& 
Z_{s\ell;ij}(z)
=
\sum_{n_0,\dots,n_s=0}^\infty
(-1)^{n_0+\dots+n_s+s}
\!\!\!\!\!\!\!\!\!\!\!\!
\sum
_{
(i_0h_0),(i_1h_1),\dots,(i_s h_s)\in\mc J
}
\!\!\!\!\!\!
\delta_{(i_0h_0)(ip_1)}\delta_{(i_s,h_s-n_s)(j1)}
\\
&\vphantom{\Big(} \qquad
\prod_{r=1}^{\ell-1}
\Big(
(z+\lambda+\partial)^{-n_{r-1}-1}\pi_{\leq\frac12}q_{(i_rh_r),(i_{r-1},h_{r-1}-n_{r-1})}
\Big)
\\
&\vphantom{\Big(} \qquad
(z\!+\!\lambda\!+\!\partial)^{-n_{\ell-1}-1}
\Big(
(a|q_{(i_\ell h_\ell),(i_{\ell-1},h_{\ell-1}-n_{\ell-1})})\lambda
+(f|[a,q_{(i_\ell h_\ell),(i_{\ell-1},h_{\ell-1}-n_{\ell-1})}])
\Big)
\\
&\vphantom{\Big(} \qquad
(z+\partial)^{-n_\ell-1}
\prod_{r=\ell+1}^{s}
\Big(
\pi_{\leq\frac12}q_{(i_r h_r),(i_{r-1},h_{r-1}-n_{r-1})}
(z+\partial)^{-n_r-1}
\Big)
\,.
\end{split}
\end{equation}
Let $a=q_{(\tilde j,\tilde k),(\tilde i,\tilde h)}\in\mf g_{\geq\frac12}$.
Then we have
$$
(a|q_{(i_\ell h_\ell),(i_{\ell-1},h_{\ell-1}-n_{\ell-1})})
=\delta_{(i_\ell h_\ell),(\tilde i,\tilde h)}\delta_{(i_{\ell-1}h_{\ell-1}),(\tilde j,\tilde k+n_{\ell-1})}
\,,
$$
and
\begin{equation*}
\begin{split}
\vphantom{\Big(}
(f|[a,q_{(i_\ell h_\ell),(i_{\ell-1},h_{\ell-1}-n_{\ell-1})}])
=
\delta_{(i_\ell h_\ell),(\tilde i \tilde h)}\delta_{(i_{\ell-1}h_{\ell-1}),(\tilde j,\tilde k+n_{\ell-1}+1)}
\\
\vphantom{\Big(}
-\delta_{(i_\ell h_\ell),(\tilde i,\tilde h-1)}\delta_{(i_{\ell-1}h_{\ell-1}),(\tilde j,\tilde k+n_{\ell-1})}
\,.
\end{split}
\end{equation*}
Hence, \eqref{eq:L2-pr6} becomes
\begin{equation}\label{eq:L2-pr7}
\begin{split}
& 
Z_{s\ell;ij}(z)
=
\sum_{n_0,\dots,n_s=0}^\infty
(-1)^{n_0+\dots+n_s+s}
\!\!\!\!\!\!\!\!\!\!\!\!
\sum
_{
(i_0h_0),(i_1h_1),\dots,(i_s h_s)\in\mc J
}
\!\!\!\!\!\!
\delta_{(i_0h_0)(ip_1)}\delta_{(i_s,h_s-n_s)(j1)}
\\
&\vphantom{\Big(} \qquad
\prod_{r=1}^{\ell-1}
\Big(
(z+\lambda+\partial)^{-n_{r-1}-1}\pi_{\leq\frac12}q_{(i_rh_r),(i_{r-1},h_{r-1}-n_{r-1})}
\Big)
(z\!+\!\lambda\!+\!\partial)^{-n_{\ell-1}-1}
\\
&\vphantom{\Big(} \qquad
\Big(
\delta_{(i_\ell h_\ell),(\tilde i,\tilde h)}\delta_{(i_{\ell-1}h_{\ell-1}),(\tilde j,\tilde k+n_{\ell-1})}
\lambda
+
\delta_{(i_\ell h_\ell),(\tilde i \tilde h)}\delta_{(i_{\ell-1}h_{\ell-1}),(\tilde j,\tilde k+n_{\ell-1}+1)}
\\
&\vphantom{\Big(} \qquad
-
\delta_{(i_\ell h_\ell),(\tilde i,\tilde h-1)}\delta_{(i_{\ell-1}h_{\ell-1}),(\tilde j,\tilde k+n_{\ell-1})}
\Big)
(z+\partial)^{-n_\ell-1}
\\
&\vphantom{\Big(} \qquad
\prod_{r=\ell+1}^{s}
\Big(
\pi_{\leq\frac12}q_{(i_r h_r),(i_{r-1},h_{r-1}-n_{r-1})}
(z+\partial)^{-n_r-1}
\Big)
\,.
\end{split}
\end{equation}
The RHS of \eqref{eq:L2-pr7} is sum of three terms,
according to the three terms in the middle parenthesis.
We then make the following change of variables:
we replace $n_{\ell-1}+1$ by $n_{\ell-1}$ in the second summand,
and $h_\ell+1$ by $h_\ell$ and $n_\ell+1$ by $n_\ell$ in the third summand.
As a result we get
\begin{equation}\label{eq:L2-pr8}
\begin{split}
& 
Z_{s\ell;ij}(z)
=
\!\!\!
\sum_{n_0,\dots,n_s=0}^\infty
(-1)^{n_0+\dots+n_s+s}
\!\!\!\!\!\!\!\!\!\!\!\!
\!\!\!\!\!\!
\sum
_{
(i_0h_0),(i_1h_1),\dots,(i_s h_s)\in\mc J
}
\!\!\!\!\!\!\!\!\!\!\!\!
\delta_{(i_0h_0)(ip_1)}\delta_{(i_s,h_s-n_s)(j1)}
\delta_{(i_\ell h_\ell),(\tilde i,\tilde h)}
\\
&\vphantom{\Big(} \qquad
\delta_{(i_{\ell-1}h_{\ell-1}),(\tilde j,\tilde k+n_{\ell-1})}
\prod_{r=1}^{\ell-1}
\Big(
(z+\lambda+\partial)^{-n_{r-1}-1}\pi_{\leq\frac12}q_{(i_rh_r),(i_{r-1},h_{r-1}-n_{r-1})}
\Big)
\\
&\vphantom{\Big(} \qquad
(z\!+\!\lambda\!+\!\partial)^{-n_{\ell-1}-1}
\Big(
\lambda
-
(1-\delta_{n_{\ell-1},0})
(z+\lambda+\partial)
+
(1-\delta_{n_\ell,0})
(z+\partial)
\Big)
\\
&\vphantom{\Big(} \qquad
(z+\partial)^{-n_\ell-1}
\prod_{r=\ell+1}^{s}
\Big(
\pi_{\leq\frac12}q_{(i_r h_r),(i_{r-1},h_{r-1}-n_{r-1})}
(z+\partial)^{-n_r-1}
\Big)
\,.
\end{split}
\end{equation}
Since $\lambda-(z+\lambda+\partial)-(z+\partial)=0$,
the RHS of \eqref{eq:L2-pr8} is the sum of the following two contributions:
\begin{equation}\label{eq:L2-pr9}
\begin{split}
& 
\sum_{n_0,\stackrel{\ell-1}{\check{\dots}},n_s=0}^\infty
(-1)^{n_0+\stackrel{\ell-1}{\check{\dots}}+n_s+s}
\!\!\!\!\!\!\!\!\!\!\!\!
\sum
_{
(i_0h_0),(i_1h_1),\dots,(i_s h_s)\in\mc J
}
\!\!\!\!\!\!
\delta_{(i_0h_0)(ip_1)}\delta_{(i_s,h_s-n_s)(j1)}
\delta_{(i_\ell h_\ell),(\tilde i,\tilde h)}
\\
&\vphantom{\Big(} \qquad
\delta_{(i_{\ell-1}h_{\ell-1}),(\tilde j,\tilde k)}
\prod_{r=1}^{\ell-1}
\Big(
(z+\lambda+\partial)^{-n_{r-1}-1}\pi_{\leq\frac12}q_{(i_rh_r),(i_{r-1},h_{r-1}-n_{r-1})}
\Big)
\\
&\vphantom{\Big(} \qquad
(z+\partial)^{-n_\ell-1}
\prod_{r=\ell+1}^{s}
\Big(
\pi_{\leq\frac12}q_{(i_r h_r),(i_{r-1},h_{r-1}-n_{r-1})}
(z+\partial)^{-n_r-1}
\Big)
\,,
\end{split}
\end{equation}
and
\begin{equation}\label{eq:L2-pr10}
\begin{split}
& 
-\sum_{n_0,\stackrel{\ell}{\check{\dots}},n_s=0}^\infty
(-1)^{n_0+\stackrel{\ell}{\check{\dots}}+n_s+s}
\!\!\!\!\!\!\!\!\!\!\!\!
\sum
_{
(i_0h_0),(i_1h_1),\dots,(i_s h_s)\in\mc J
}
\!\!\!\!\!\!
\delta_{(i_0h_0)(ip_1)}\delta_{(i_s,h_s-n_s)(j1)}
\delta_{(i_\ell h_\ell),(\tilde i,\tilde h)}
\\
&\vphantom{\Big(} \qquad
\delta_{(i_{\ell-1}h_{\ell-1}),(\tilde j,\tilde k+n_{\ell-1})}
\prod_{r=1}^{\ell-1}
\Big(
(z+\lambda+\partial)^{-n_{r-1}-1}\pi_{\leq\frac12}q_{(i_rh_r),(i_{r-1},h_{r-1}-n_{r-1})}
\Big)
\\
&\vphantom{\Big(} \qquad
(z\!+\!\lambda\!+\!\partial)^{-n_{\ell-1}-1}
\prod_{r=\ell+1}^{s}
\Big(
\pi_{\leq\frac12}q_{(i_r h_r),(i_{r-1},h_{r-1}-n_{r-1})}
(z+\partial)^{-n_r-1}
\Big)
\,.
\end{split}
\end{equation}
Note that, for $\ell=1$, \eqref{eq:L2-pr9} becomes $0$,
since we get a factor $\delta_{(\tilde j\tilde k),(ip_1)}$,
and $q_{(ip_1),(\tilde i\tilde h)}\in\mf g_{\leq0}$ for every $(\tilde i\tilde h)\in\mc J$,
contrary to the assumption that $a\in\mf g_{\geq\frac12}$.
Similarly, for $\ell=s$, \eqref{eq:L2-pr10} becomes $0$,
since we get a factor $\delta_{(\tilde i\tilde h),(j1)}$,
and $q_{(\tilde j\tilde k),(j1)}\in\mf g_{\leq0}$ for every $(\tilde i\tilde h)\in\mc J$,
contrary to the assumption that $a\in\mf g_{\geq\frac12}$.
For $2\leq\ell\leq s$, we can rewrite \eqref{eq:L2-pr9} as
\begin{equation}\label{eq:L2-pr11}
\begin{split}
& 
\sum_{n_0,\stackrel{\ell-1}{\check{\dots}},n_s=0}^\infty
(-1)^{n_0+\stackrel{\ell-1}{\check{\dots}}+n_s+s}
\!\!\!\!\!\!\!\!\!\!\!\!
\sum
_{
(i_0h_0),\stackrel{\ell-1}{\check{\dots}},(i_s h_s)\in\mc J
}
\!\!\!\!\!\!
\delta_{(i_0h_0)(ip_1)}\delta_{(i_s,h_s-n_s)(j1)}
\\
&\vphantom{\Big(} \qquad
\prod_{r=1}^{\ell-2}
\Big(
(z+\lambda+\partial)^{-n_{r-1}-1}\pi_{\leq\frac12}q_{(i_rh_r),(i_{r-1},h_{r-1}-n_{r-1})}
\Big)
\\
&\vphantom{\Big(} \qquad
(z+\lambda+\partial)^{-n_{\ell-2}-1}
\pi_{\leq\frac12}q_{(\tilde j,\tilde k),(i_{\ell-2},h_{\ell-2}-n_{\ell-2})}
\delta_{(i_\ell h_\ell),(\tilde i,\tilde h)}
\\
&\vphantom{\Big(} \qquad
\prod_{r=\ell+1}^{s}
\Big(
(z+\partial)^{-n_{r-1}-1}
\pi_{\leq\frac12}q_{(i_r h_r),(i_{r-1},h_{r-1}-n_{r-1})}
\Big)
(z+\partial)^{-n_s-1}
\,,
\end{split}
\end{equation}
while, for $1\leq\ell\leq s-1$, we can rewrite \eqref{eq:L2-pr10} as
\begin{equation}\label{eq:L2-pr12}
\begin{split}
& 
-\sum_{n_0,\stackrel{\ell}{\check{\dots}},n_s=0}^\infty
(-1)^{n_0+\stackrel{\ell}{\check{\dots}}+n_s+s}
\!\!\!\!\!\!\!\!\!\!\!\!
\sum
_{
(i_0h_0),\stackrel{\ell}{\check{\dots}},(i_s h_s)\in\mc J
}
\!\!\!\!\!\!
\delta_{(i_0h_0)(ip_1)}\delta_{(i_s,h_s-n_s)(j1)}
\\
&\vphantom{\Big(} \qquad
\prod_{r=1}^{\ell-1}
\Big(
(z+\lambda+\partial)^{-n_{r-1}-1}\pi_{\leq\frac12}q_{(i_rh_r),(i_{r-1},h_{r-1}-n_{r-1})}
\Big)
\\
&\vphantom{\Big(} \qquad
(z\!+\!\lambda\!+\!\partial)^{-n_{\ell-1}-1}
\pi_{\leq\frac12}q_{(i_{\ell+1} h_{\ell+1}),(\tilde i,\tilde h)}
\delta_{(i_{\ell-1}h_{\ell-1}-n_{\ell-1}),(\tilde j,\tilde k)}
\\
&\vphantom{\Big(} \qquad
\prod_{r=\ell+2}^{s}
\Big(
(z+\partial)^{-n_{r-1}-1}
\pi_{\leq\frac12}q_{(i_r h_r),(i_{r-1},h_{r-1}-n_{r-1})}
\Big)
(z+\partial)^{-n_s-1}
\,.
\end{split}
\end{equation}
Summing \eqref{eq:L2-pr11} over all values of $\ell=2,\dots,s$
we get, after a shift of indices
\begin{equation}\label{eq:L2-pr13}
\begin{split}
& 
\sum_{\ell=2}^{s}
\sum_{n_0,\dots,n_{s-1}=0}^\infty
(-1)^{n_0+\dots+n_{s-1}+s}
\!\!\!\!\!\!\!\!\!\!\!\!
\sum
_{
(i_0h_0),\dots,(i_{s-1} h_{s-1})\in\mc J
}
\!\!\!\!\!\!
\delta_{(i_0h_0)(ip_1)}\delta_{(i_{s-1},h_{s-1}-n_{s-1})(j1)}
\\
&\vphantom{\Big(} \qquad
\prod_{r=1}^{\ell-2}
\Big(
(z+\lambda+\partial)^{-n_{r-1}-1}\pi_{\leq\frac12}q_{(i_rh_r),(i_{r-1},h_{r-1}-n_{r-1})}
\Big)
\\
&\vphantom{\Big(} \qquad
(z+\lambda+\partial)^{-n_{\ell-2}-1}
\pi_{\leq\frac12}q_{(\tilde j,\tilde k),(i_{\ell-2},h_{\ell-2}-n_{\ell-2})}
\delta_{(i_{\ell-1} h_{\ell-1}),(\tilde i,\tilde h)}
\\
&\vphantom{\Big(} \qquad
\prod_{r=\ell}^{s-1}
\Big(
(z+\partial)^{-n_{r-1}-1}
\pi_{\leq\frac12}q_{(i_r h_r),(i_{r-1},h_{r-1}-n_{r-1})}
\Big)
(z+\partial)^{-n_{s-1}-1}
\,,
\end{split}
\end{equation}
and similarly, 
summing \eqref{eq:L2-pr12} over all values of $\ell=1,\dots,s-1$
we get,
\begin{equation}\label{eq:L2-pr14}
\begin{split}
&
-
\sum_{\ell=1}^{s-1}
\sum_{n_0,\dots,n_{s-1}=0}^\infty
(-1)^{n_0+\dots+n_{s-1}+s}
\!\!\!\!\!\!\!\!\!\!\!\!
\sum
_{
(i_0h_0),\dots,(i_{s-1} h_{s-1})\in\mc J
}
\!\!\!\!\!\!
\delta_{(i_0h_0)(ip_1)}\delta_{(i_{s-1},h_{s-1}-n_{s-1})(j1)}
\\
&\vphantom{\Big(} \qquad
\prod_{r=1}^{\ell-1}
\Big(
(z+\lambda+\partial)^{-n_{r-1}-1}\pi_{\leq\frac12}q_{(i_rh_r),(i_{r-1},h_{r-1}-n_{r-1})}
\Big)
\\
&\vphantom{\Big(} \qquad
(z\!+\!\lambda\!+\!\partial)^{-n_{\ell-1}-1}
\pi_{\leq\frac12}q_{(i_{\ell} h_{\ell}),(\tilde i,\tilde h)}
\delta_{(i_{\ell-1}h_{\ell-1}-n_{\ell-1}),(\tilde j,\tilde k)}
\\
&\vphantom{\Big(} \qquad
\prod_{r=\ell+1}^{s-1}
\Big(
(z+\partial)^{-n_{r-1}-1}
\pi_{\leq\frac12}q_{(i_r h_r),(i_{r-1},h_{r-1}-n_{r-1})}
\Big)
(z+\partial)^{-n_{s-1}-1}
\,.
\end{split}
\end{equation}
Combining \eqref{eq:L2-pr13} and \eqref{eq:L2-pr14}
we finally get, recalling \eqref{eq:L2-pr5},
\begin{equation}\label{eq:L2-pr15}
\begin{split}
& 
\sum_{\ell=1}^s
Z_{s\ell;ij}(z)
\\
&\vphantom{\Big(}
=
\sum_{\ell=1}^{s-1}
\sum_{n_0,\dots,n_{s-1}=0}^\infty
(-1)^{n_0+\dots+n_{s-1}+s}
\!\!\!\!\!\!\!\!\!\!\!\!
\sum
_{
(i_0h_0),\dots,(i_{s-1} h_{s-1})\in\mc J
}
\!\!\!\!\!\!
\delta_{(i_0h_0)(ip_1)}\delta_{(i_{s-1},h_{s-1}-n_{s-1})(j1)}
\\
&\vphantom{\Big(} \qquad
\prod_{r=1}^{\ell-1}
\Big(
(z+\lambda+\partial)^{-n_{r-1}-1}\pi_{\leq\frac12}q_{(i_rh_r),(i_{r-1},h_{r-1}-n_{r-1})}
\Big)
\\
&\vphantom{\Big(} \qquad
(z+\lambda+\partial)^{-n_{\ell-1}-1}
\pi_{\leq\frac12}
[q_{(\tilde j,\tilde k),(\tilde i,\tilde h)},q_{(i_{\ell} h_{\ell}),(i_{\ell-1},h_{\ell-1}-n_{\ell-1})}]
\\
&\vphantom{\Big(} \qquad
\prod_{r=\ell+1}^{s-1}
\Big(
(z+\partial)^{-n_{r-1}-1}
\pi_{\leq\frac12}q_{(i_r h_r),(i_{r-1},h_{r-1}-n_{r-1})}
\Big)
(z+\partial)^{-n_{s-1}-1}
\\
&\vphantom{\Big(}
=
-\sum_{\ell=1}^{s-1}
Y_{s-1,\ell;ij}(z)
\,.
\end{split}
\end{equation}
In conclusion, recalling \eqref{eq:L2-pr4},
we get $\rho\{a_\lambda\sum_{s=1}^\infty X_{s;ij}(z)\}_\epsilon=0$, as claimed.
\end{proof}
\begin{remark}\label{cor:daniele-proof}
After submitting the paper, in trying to quantize the result of the present paper,
we discovered a simpler, more conceptual proof of Theorem \ref{prop:L2},
which we present in Appendix \ref{sec:appendix}.
\end{remark}
\begin{corollary}\phantomsection\label{thm:L2}
\begin{enumerate}[(a)]
\item
The matrix pseudodifferential operator $L_1(\partial)$
defined by \eqref{eq:La} has coefficients with entries in the differential algebra $\mc W(\mf g,f)$.
\item
The matrix pseudodifferential operator $L(\partial)$
defined by \eqref{eq:Lb} has coefficients with entries in the field of fractions $\mc K(\mf g,f)$
of the differential algebra $\mc W(\mf g,f)$.
\end{enumerate}
\end{corollary}
\begin{proof}
By Theorem \ref{thm:L1}(b), $L_1^{-1}(\partial)$ has an expansion as in \eqref{eq:thmL1-1},
and by Theorem \ref{prop:L2} 
its coefficients have entries in the differential algebra $\mc W(\mf g,f)$.
Then $L_1(\partial)$ can be obtained by the geometric series expansion of the inverse of \eqref{eq:thmL1-1},
and therefore its coefficients will still have entries in $\mc W(\mf g,f)$.
This proves part (a).
By Proposition \ref{prop:generic},
$\bar JL_1^{-1}(\partial)\bar I$ is invertible
and its inverse has coefficients with entries in the field of fractions of $\mc W(\mf g,f)$.
On the other hand, by Proposition \ref{prop:factor}(b), 
$L(\partial)$ coincides with the inverse of $\bar JL_1^{-1}(\partial)\bar I$,
proving (b).
\end{proof}

\subsection{$L_1(\partial)$ is of $\bar S$-Adler type for $\mc W_\epsilon(\mf g,f,S)$}\label{sec:4.4}

\begin{theorem}\label{thm:L3}
Let $S\in\mf g_d$, and let $\bar S\in\Mat_{r_1\times r_1}\mb F$
be the corresponding matrix via \eqref{eq:s0}.
\begin{enumerate}[(a)]
\item
The matrix pseudodifferential operator $L_1(\partial)\in\Mat_{r_1\times r_1}\mc W(\mf g,f)((\partial^{-1}))$
defined by \eqref{eq:La} (cf. Corollary \ref{thm:L2}(a)) is of $\bar S$-Adler type
with respect to the compatible $\lambda$-brackets 
$\{\cdot\,_\lambda\,\cdot\}_0^{\mc W}$ and $\{\cdot\,_\lambda\,\cdot\}_1^{\mc W}$
of the family of PVA's $\mc W_\epsilon(\mf g,f,S)$, $\epsilon\in\mb F$.
\item
The matrix pseudodifferential operator $L(\partial)\in\Mat_{\bar r\times\bar r}\mc K(\mf g,f)((\partial^{-1}))$
defined by \eqref{eq:Lb} (cf. Corollary \ref{thm:L2}(b)) is of bi-Adler type
with respect to the compatible $\lambda$-brackets 
$\{\cdot\,_\lambda\,\cdot\}_0^{\mc W}$ and $\{\cdot\,_\lambda\,\cdot\}_1^{\mc W}$
of the family of PVA's $\mc W_\epsilon(\mf g,f,S)$, $\epsilon\in\mb F$.
\end{enumerate}
\end{theorem}
\begin{proof}
Since the matrix $S$ has constant entries,
it follows by Theorem \ref{prop:L2} (and the geometric series expansion)
that the matrix
\begin{equation}\label{eq:L3-pr1}
L_{1,\epsilon}^{-1}(\partial)
=
J_1(\id_N\partial+\rho(Q)+\epsilon S)^{-1}I_1
\end{equation}
lies in $\Mat_{r_1\times r_1}\mc W(\mf g,f)((\partial^{-1}))$
for every $\epsilon$.
Since the map 
$\rho:\,\mc V(\mf g)\to\mc V(\mf g_{\leq\frac12})$
is a homomorphism of differential algebras,
we can rewrite $L_{1,\epsilon}^{-1}(\partial)$ as
\begin{equation}\label{eq:L3-pr2}
L_{1,\epsilon}^{-1}(\partial)
=
\rho(J_1A_{\epsilon S}^{-1}(\partial)I_1)
\,\,\text{ where }\,\,
A_{\epsilon S}(\partial)=\id_N\partial+Q+\epsilon S
\,.
\end{equation}
Recall from Example \ref{ex:A}
that $A_{\epsilon S}(\partial)$ is of Adler type with respect 
to the $\lambda$-bracket $\{\cdot\,_\lambda\,\cdot\}_\epsilon$ defined by \eqref{lambda}.
Hence, by Theorem \ref{thm:inverse-adler},
$A_{\epsilon S}^{-1}(\partial)$ is of Adler type with respect to the opposite $\lambda$-bracket
$-\{\cdot\,_\lambda\,\cdot\}_\epsilon$.
It follows that
\begin{equation}\label{eq:L3-pr3}
\begin{split}
& -\{{L_{1,\epsilon}^{-1}(z)_{ij}}_\lambda{L_{1,\epsilon}^{-1}(w)_{hk}}\}_\epsilon^{\mc W}
=
-\{{\rho(J_1A_{\epsilon S}^{-1}(z)I_1)_{ij}}_\lambda
{\rho(J_1A_{\epsilon S}^{-1}(w)I_1)_{hk}}\}_\epsilon^{\mc W} \\
& =
-\rho\{{(J_1A_{\epsilon S}^{-1}(z)I_1)_{ij}}_\lambda
{(J_1A_{\epsilon S}^{-1}(w)I_1)_{hk}}\}_\epsilon \\
& =
-\rho\{{A_{\epsilon S}^{-1}(z)_{(ip_1),(j1)}}_\lambda
{A_{\epsilon S}^{-1}(w)_{(hp_1),(k1)}}\}_\epsilon \\
& =
\rho
A_{\epsilon S}^{-1}(w+\lambda+\partial)_{(hp_1),(j1)}
\iota_z(z-w-\lambda-\partial)^{-1}
((A_{\epsilon S}^{-1})_{(ip_1),(k1)})^*(\lambda-z) \\
&\quad
- 
\rho
A_{\epsilon S}^{-1}(z)_{(hp_1),(j1)}
\iota_z(z-w-\lambda-\partial)^{-1}
A_{\epsilon S}^{-1}(w)_{(ip_1),(k1)} \\
& =
(L_{1,\epsilon}^{-1})_{hj}(w+\lambda+\partial)
\iota_z(z-w-\lambda-\partial)^{-1}
((L_{1,\epsilon}^{-1})_{ik})^*(\lambda-z) \\
&\quad
- 
(L_{1,\epsilon}^{-1})_{hj}(z)
\iota_z(z-w-\lambda-\partial)^{-1}
(L_{1,\epsilon}^{-1})_{ik}(w)
\,.
\end{split}
\end{equation}
In the second equality we used Theorem \ref{daniele1}(b)
(and the above observation that $L_{1,\epsilon}^{-1}(\partial)$ has
coefficients with entries in $\mc W(\mf g,f)$),
while in the third and the fifth equality we used the definition \eqref{eq:factor1}
of the matrices $I_1$ and $J_1$.
It follows by \eqref{eq:L3-pr3} that $L_{1,\epsilon}^{-1}$
is of Adler type with respect to the negative of the $\lambda$-bracket
of the PVA $\mc W_\epsilon(\mf g,f,S)$.
Therefore, by Theorem \ref{thm:inverse-adler},
its inverse 
\begin{equation}\label{eq:L3-pr4}
L_{1,\epsilon}(\partial)
=
(J_1\rho(A_{\epsilon S})^{-1}(\partial)I_1)^{-1}
=
|\rho A_{\epsilon S}|_{I_1J_1}(\partial)
\,,
\end{equation}
is of Adler type with respect to the $\lambda$-bracket $\{\cdot\,_\lambda\,\cdot\}_\epsilon^{\mc W}$.
Note that in \eqref{eq:L3-pr4} we can take map $\rho$ out of the generalized quasideterminant,
since it is a differential algebra homomorphism.
If $S=IJ\in\mf g_d$ and $\bar S=\bar I\bar J\in\Mat_{r_1\times r_1}\mb F$ are as in Proposition \ref{prop:factor},
we then get, by Theorem \ref{thm:main-quasidet}, that
\begin{equation}\label{eq:L3-pr5}
L_{1,\epsilon}(\partial)
=
\rho|A(\partial)+\epsilon S|_{I_1J_1}
=
\rho|A(\partial)|_{I_1J_1}+\epsilon \bar S
=
|\rho A(\partial)|_{I_1J_1}+\epsilon\bar S
=
L_1(\partial)+\epsilon\bar S
\,.
\end{equation}
Hence, 
$L_1(\partial)$ is of $\bar S$-Adler type with respect to the pencil
of $\lambda$-brackets $\{\cdot\,_\lambda\,\cdot\}_\epsilon^{\mc W}$, $\epsilon\in\mb F$,
proving (a).
Furthermore, by the hereditary property \eqref{eq:hereditary} 
of generalized quasideterminants and Theorem \ref{thm:main-quasidet} again,
we also have
\begin{equation}\label{eq:L3-pr5b}
\begin{split}
& |L_{1,\epsilon}(\partial)|_{\bar I\bar J}
=
||\rho A(\partial)|_{I_1J_1}+\epsilon\bar S|_{\bar I\bar J}
=
||\rho A(\partial)|_{I_1J_1}|_{\bar I\bar J}+\epsilon \id_{r} \\
& =
|\rho A(\partial)|_{IJ}+\epsilon \id_{\bar r}
=
L(\partial)+\epsilon \id_{\bar r}
\,.
\end{split}
\end{equation}
Hence, by \eqref{eq:L3-pr5b} and Theorem \ref{thm:quasidet-adler},
we conclude that $L(\partial)$
is of bi-Adler type with respect to the pencil of $\lambda$-brackets 
$\{\cdot\,_\lambda\,\cdot\}_\epsilon^{\mc W}$, $\epsilon\in\mb F$,
proving the claim.
\end{proof}

\section{Explicit form of \texorpdfstring{$L$}{L}}\label{sec:5}

\subsection{A choice of a cross section to a nilpotent orbit in \texorpdfstring{$\mf{gl}_N$}{gl\_N}}\label{sec:5.1}

Let $\mf g=\mf{gl}_N$ and let $f\in\mf g$ be a non-zero nilpotent element, corresponding to a partition $p$.
In terms of the basis and notation introduced in Section \ref{sec:4.1}
we have the following result:
\begin{proposition}\label{20150601:prop}
For any partition $p$ we have
$\mf g=[f,\mf g]\oplus U$, where
\begin{equation}\label{eq:basisU}
U=\Span\Big\{E_{(j1),(i,p_i-k)}
\,\,,\,\,\text{ where }\,\,
1\leq i,j\leq r
\,\,\text{ and }\,\,
0\leq k\leq\min\{p_i,p_j\}-1
\Big\}
\,.
\end{equation}
\end{proposition}
\begin{proof}
Given an elementary matrix $E_{(jk),(ih)}$, we have
$[f,E_{(jk),(ih)}]=E_{(j,k+1),(ih)}-E_{(jk),(i,h-1)}$, 
which is depicted as
%
$$
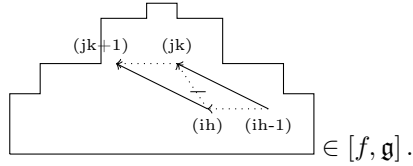
\begin{tikzpicture}[scale=0.4]
\draw (0,0)--(10,0)--(10,2)--(9,2)--(9,3)--(7,3)--(7,4.5)--(5.5,4.5)--(5.5,5)--(4.5,5)--(4.5,4.5)--(3,4.5)--(3,3)--(1,3)--(1,2)--(0,2)--(0,0);
\draw [dotted,->] (6.5,1.5)--(5.5,3);
\draw [dotted,->] (5.5,3)--(3.5,3);
\draw [dotted,->] (8.5,1.5)--(6.5,1.5);
\draw [->] (6.5,1.5)--(3.5,3);
\draw [->] (8.5,1.5)--(5.5,3);
\node [below] at (6.5,1.5) {\tiny{(ih)}};
\node [above] at (3,3) {\tiny{(jk+1)}};
\node [below] at (8.5,1.5) {\tiny{(ih-1)}};
\node [above] at (5.5,3) {\tiny{(jk)}};
\node [above] at (6.2,1.5) {{\bf --}};
\end{tikzpicture}\in[f,\mf g]\,.
$$
%
Hence, in the quotient space $\mf g/[f,\mf g]$, two arrows are equivalent
if one is obtained from the other by a horizontal shift to the left or to the right.
%
%
Moreover, we have $E_{(jk),(i1)}=[f,E_{(jk-1),(i1)}]\in[f,\mf g]$, for $k>1$,
and $E_{(jp_j),(ih)}=-[f,E_{(jp_j),(i,h+1)}]\in[f,\mf g]$, for $h<p_i$,
namely, if an arrow has the tail at the center of the foremost right box, then it lies in $[f,\mf g]$,
and similarly, if an arrow has the head at the center of the foremost left box, then it lies in $[f,\mf g]$:
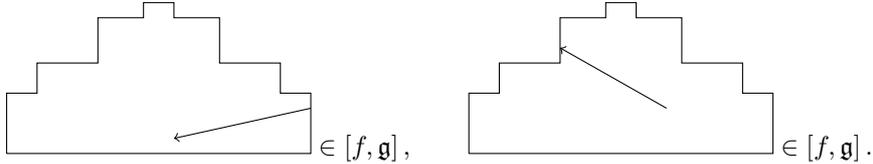
$$
\begin{tikzpicture}[scale=0.4]
\draw (0,0)--(10,0)--(10,2)--(9,2)--(9,3)--(7,3)--(7,4.5)--(5.5,4.5)--(5.5,5)--(4.5,5)--(4.5,4.5)--(3,4.5)--(3,3)--(1,3)--(1,2)--(0,2)--(0,0);
\draw [->] (10,1.5)--(5.5,0.5);
\end{tikzpicture}
\in [f,\mf g]\,,
\qquad
\begin{tikzpicture}[scale=0.4]
\draw (0,0)--(10,0)--(10,2)--(9,2)--(9,3)--(7,3)--(7,4.5)--(5.5,4.5)--(5.5,5)--(4.5,5)--(4.5,4.5)--(3,4.5)--(3,3)--(1,3)--(1,2)--(0,2)--(0,0);
\draw [->] (6.5,1.5)--(3,3.5);
\end{tikzpicture}
\in [f,\mf g]\,.
$$
Hence, for the quotient space $\mf g/[f,\mf g]$, we can take
as representatives the arrows with the head in the foremost right box (of the corresponding row),
with the property that, when shifted to the left, they have the tail in the foremost left box:
$$
E_{(jk),(ih)}=
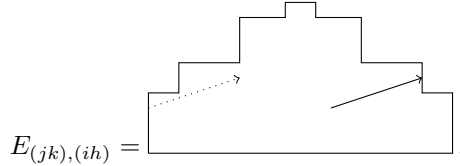
\begin{tikzpicture}[scale=0.4]
\draw (0,0)--(10,0)--(10,2)--(9,2)--(9,3)--(7,3)--(7,4.5)--(5.5,4.5)--(5.5,5)--(4.5,5)--(4.5,4.5)--(3,4.5)--(3,3)--(1,3)--(1,2)--(0,2)--(0,0);
\draw [->] (6,1.5)--(9,2.5);
\draw [dotted,->] (0,1.5)--(3,2.5);
\end{tikzpicture}
\,.
$$
These are the matrices $E_{(j1),(i,p_i-\ell)}$ with $\ell$ satisfying $0\leq\ell\leq\min\{p_i,p_j\}-1$.
By definition, $U\subset\mf g$ is the linear span of all these matrices.
We have proved that $U+[f,\mf g]=\mf g$.
We are left to prove that this sum is a direct sum.
Indeed, $\dim U\geq\dim\mf g-\dim[f,\mf g]=\dim\mf g^f$.
Let $\widetilde{e}\in\mf g$ be the operator of ``shift'' to the right:
$$
\widetilde{e}=\sum_{(ih)\in\mc J\,|\,h<p_i}E_{(ih),(i,h+1)}
\,.
$$
By an obvious symmetry argument,
we have $\dim \mf g^{\widetilde{e}}=\dim\mf g^f$.
On the other hand, we have the injective linear map $U\to\mf g^{\widetilde{e}}$ given by
\begin{equation}\label{eq:ge}
E_{(j1),(i,p_i-\ell)}\mapsto\sum_{k=0}^\ell E_{(j,k+1),(i,p_i+k-\ell)}
\,\,,\qquad
0\leq\ell\leq\min(p_i,p_j)-1
\,.
\end{equation}
Hence, $\dim U\leq\dim\mf g^{\widetilde{e}}=\dim\mf g^f$.
This proves that $\dim U=\dim\mf g^f$, and therefore that $\mf g=U\oplus[f,\mf g]$.
\end{proof}

\subsection{Description of $\mf g^f$ and $\mf g_0^f$}\label{sec:5.2a}

Recall from Section \ref{sec:3.2} that any subspace $U\subset\mf g$ complementary to $[f,\mf g]$
is dual to $\mf g^f$.
In particular, consider the space $U$ defined in Proposition \ref{20150601:prop}
and its basis defined in \eqref{eq:basisU}.
We can find the corresponding dual basis of $\mf g^f$.
\begin{proposition}\label{prop:gf}
The basis of $\mf g^f$, dual to the basis $\{E_{(j1),(i,p_i-k)}\,|\,1\leq i,j\leq r,\,0\leq k\leq\min\{p_i,p_j\}-1\}$
of $U$, is (cf. \eqref{eq:ge}):
\begin{equation}\label{eq:basisgf}
f_{ij;k}
:=
\sum_{h=0}^k
E_{(i,p_i+h-k),(j,h+1)}
\,\,,\,\,\,\,
1\leq i,j\leq r\,,\,\, 0\leq k\leq\min\{p_i,p_j\}-1
\,.
\end{equation}
\end{proposition}
\begin{proof}
It is immediate to check that $[f,f_{ij;k}]=0$ for every $1\leq i,j\leq r$ and $0\leq k\leq\min\{p_1,p_j\}-1$,
and that $\tr(E_{(j1),(i,p_i-k)}f_{i'j';k'})=\delta_{ii'}\delta_{jj'}\delta_{kk'}$.
\end{proof}
It is useful to have an explicit description of $\mf g_0^f$:
\begin{corollary}\label{cor:g0f}
The space $\mf g_0^f$ is spanned by the elements
$f_{ij;p_i-1}$ with $1\leq i,j\leq r$ such that $p_i=p_j$.
\end{corollary}
\begin{proof}
By equation \eqref{eq:adx},
the $\ad x$-eigenvalue of $f_{ij;k}$
is $\delta=\frac12(p_i-p_j)-(p_i-k-1)$.
Recalling that $k\leq\min\{p_i,p_j\}-1$,
we get that $\delta=0$ if and only if $p_i=p_j$ and $k=p_i-1$.
\end{proof}
\begin{corollary}\label{cor:s1g0f}
The element $S_1$ in \eqref{eq:s1} commutes with $\mf g_0^f$.
\end{corollary}
\begin{proof}
We have
\begin{equation*}
\begin{split}
& [S_1,f_{ij;k}]
=
\sum_{a=1}^{r_1}
\sum_{h=0}^k
[E_{(a1),(ap_1)},E_{(i,p_i+h-k),(j,h+1)}] \\
& =
\delta_{p_i,p_1}E_{(i1),(j,k+1)}-\delta_{p_j,p_1}E_{(i,p_i-k),(j,p_1)}
\,.
\end{split}
\end{equation*}
To conclude, we observe that the RHS of the above equation is zero for $p_i=p_j$ and $k=p_i-1$.
\end{proof}
By Theorem \ref{thm:structure-W},
the $\mc W$-algebra $\mc W=\mc W(\mf g,f)$
is, as a differential algebra, the algebra of differential polynomials
in the corresponding set of generators,
\begin{equation}\label{eq:basisW}
w_{ij;k}
:=
w(f_{ij;k})
\,\,,\,\,\,\,
1\leq i,j\leq r\,,\,\, 0\leq k\leq\min\{p_i,p_j\}-1
\,.
\end{equation}
As a consequence of Corollary \ref{cor:s1g0f}, we have
\begin{corollary}\label{cor:casimirs}
For $1\leq i,j\leq r$ such that $p_i=p_j$,
the elements $w_{ij;p_i-1}$ are central for the $1$-st $\lambda$-bracket $\{\cdot\,_\lambda\,\cdot\}_1^{\mc W}$
of the family of $\mc W$-algebras $\mc W_\epsilon(\mf{gl}_N,f,S_1)$, $\epsilon\in\mb F$.
\end{corollary}
\begin{proof}
It follows from \cite[Eq.6.2]{DSKV16}.
\end{proof}

\subsection{Explicit form of $L(\partial)$}\label{sec:5.2}

For $i,j\in\{1,\dots,r\}$, let
\begin{equation}\label{eq:Wij}
W_{ij}(\partial)
=
\sum_{k=0}^{\min\{p_i,p_j\}-1}w_{ij;k}(-\partial)^k
\in\mc W[\partial]
\,.
\end{equation}
Denote by $W(\partial)$ the $r\times r$ matrix differential operator with entries \eqref{eq:Wij}
(transposed):
\begin{equation}\label{eq:Wd}
W(\partial)
=
\sum_{i,j=1}^r 
W_{ij}(\partial)
E_{ji}
=
\left(\begin{array}{cccc}
W_{11}(\partial) & W_{21}(\partial) & \dots & W_{r1}(\partial) \\
W_{12}(\partial) & W_{22}(\partial) & \dots & W_{r2}(\partial) \\
\vdots & \vdots & \ddots & \vdots \\
W_{1r}(\partial) & W_{2r}(\partial) & \dots & W_{rr}(\partial)
\end{array}\right)
\,,
\end{equation}
and by $(-\partial)^p$ the diagonal $r\times r$ matrix
with diagonal entries $(-\partial)^{p_i}$, $i=1,\dots,r$:
\begin{equation}\label{eq:dp}
(-\partial)^p
=
\sum_{i=1}^r (-\partial)^{p_i}E_{ii}
=
\left(\begin{array}{ccc}
(-\partial)^{p_1} & & 0 \\
& \ddots& \\
0 & & (-\partial)^{p_r}
\end{array}\right)
\,.
\end{equation}
\begin{theorem}\label{thm:L1-explicit}
The matrix pseudodifferential operator $L_1(\partial)\in\Mat_{r_1\times r_1}\mc W((\partial^{-1}))$
defined by \eqref{eq:La} is equal to
\begin{equation}\label{eq:L1-explicit}
L_1(\partial)
:=
|\id_N\partial+\rho(Q)|_{I_1J_1}
=
|-(-\partial)^p+W(\partial)|_{I_{rr_1}J_{r_1r}}
\,,
\end{equation}
where $I_{rr_1}\in\Mat_{r\times r_1}\mb F$ and $J_{r_1r}\in\Mat_{r_1\times r}\mb F$
are as in \eqref{eq:ENM} and \eqref{eq:EMN} respectively.
\end{theorem}
\begin{proof}
According to Corollary \ref{thm:L2},
the matrix pseudodifferential $L_1(\partial)$ has coefficients with entries in $\mc W(\mf g,f)$.
Hence, by Theorem \ref{thm:structure-W},
$L_1(\partial)$ is unchanged if we apply first the map 
$\pi_{\mf g^f}:\,\mc V(\mf g_{\leq\frac12})\to\mc V(\mf g^f)$
and then the map $w:\,\mc V(\mf g^f)\to\mc W(\mf g,f)$
to the entries of its coefficients:
\begin{equation}\label{eq:expl-pr1}
L_1(\partial)=w(\pi_{\mf g^f}L_1(\partial))
\,.
\end{equation}
Since $\pi_{\mf g^f}$ and $w$ are homomorphisms of differential algebras,
they commute with taking generalized quasideterminants.
Hence, \eqref{eq:expl-pr1} can be rewritten as
\begin{equation}\label{eq:expl-pr2}
L_1(\partial)
=
\big|
\id_N\partial+f+
w(\pi_{\mf g^f}\pi_{\leq\frac12}Q)
\big|_{I_1J_1}
\,.
\end{equation}
Here we used the definition \eqref{rho} of the map $\rho:\,\mc V(\mf g)\to\mc V(\mf g_{\leq\frac12})$.
By the definition \eqref{eq:basisW} of the generators $w_{ij;k}\in\mc W$,
and since the bases \eqref{eq:basisgf} of $\mf g^f$
and \eqref{eq:basisU} of $U$ are dual to each other,
we immediately get that
\begin{equation}\label{eq:expl-pr3}
w(\pi_{\mf g^f}\pi_{\leq\frac12}Q)
=
\sum_{i,j=1}^r\sum_{k=0}^{\min\{p_1,p_j\}-1}w_{ij;k} E_{(j1),(i,p_i-k)}
\,.
\end{equation}
Combining \eqref{eq:expl-pr2} and \eqref{eq:expl-pr3}, and recalling \eqref{eq:f}, we get
\begin{equation}\label{eq:expl-pr4}
L_1(\partial)
=
\Big|
\id_N\partial+
\sum_{i=1}^r\sum_{h=1}^{p_1-1}E_{(i,h+1),(ih)}
+
\sum_{i,j=1}^r\sum_{k=0}^{\min\{p_1,p_j\}-1} w_{ij;k} E_{(j1),(i,p_i-k)}
\Big|_{I_1J_1}
\,.
\end{equation}
In the proof of Theorem \ref{prop:L2}
we computed the entries of the inverse of the matrix $L_1(\partial)$,
which are given by equations \eqref{eq:L2-pr2} and \eqref{eq:L2-pr3}. 
According to equation \eqref{eq:expl-pr4},
the same formulas hold if we replace, 
in the definition of $L_1(\partial)$,
the matrix $\pi_{\leq\frac12}Q$ 
by the matrix $w(\pi_{\mf g^f}\pi_{\leq\frac12}Q)$, given by \eqref{eq:expl-pr3}.
This means replacing in the RHS of \eqref{eq:L2-pr3}
the expression $\pi_{\leq\frac12}q_{(i_rh_r),(i_{r-1},h_{r-1}-n_{r-1})}$
by $\delta_{h_{r-1}-n_{r-1},1}
w_{i_ri_{r-1};p_{i_r}-h_r}$
if $p_{i_r}-h_r\leq\min\{p_{i_r},p_{i_{r-1}}\}-1$,
and by $0$ otherwise.
Hence, we get, after some algebraic manipulations
\begin{equation}\label{eq:L2-pr2b}
L_1^{-1}(\partial)_{ij} 
=
-\delta_{ij}(-\partial)^{-p_1}+
\sum_{s=1}^\infty \bar X_{s;ij}(\partial) 
\,,
\end{equation}
where
\begin{equation}\label{eq:L2-pr3b}
\begin{split}
& 
\bar X_{s;ij}(\partial)
=
\!\!
-\sum_{i_1,\dots,i_{s-1}=1}^r
\!\!
\sum_{k_1=0}^{\min\{p_{i_1},p_i\}-1}
\!
\sum_{k_2=0}^{\min\{p_{i_2},p_{i_1}\}-1}
\!\!
\dots
\!\!
\sum_{k_s=0}^{\min\{p_j,p_{i_{s-1}}\}-1}
\!\!
\\&\qquad
(-\partial)^{-p_1}
w_{i_1i;k_1}
(-\partial)^{k_1-p_{i_1}}
w_{i_2i_1;k_2}
(-\partial)^{k_2-p_{i_2}}
\dots\\&\qquad\dots
w_{i_{s-1}i_{s-2};k_{s-1}}
(-\partial)^{k_{s-1}-p_{i_{s-1}}}
w_{ji_{s-1};k_s}
(-\partial)^{k_s-p_1}
\\&
=
-\sum_{i_1,\dots,i_{s-1}=1}^r
(-\partial)^{-p_1}
W_{i_1i}(\partial)
(-\partial)^{-p_{i_1}}
W_{i_2i_1}(\partial)
(-\partial)^{-p_{i_2}}
\dots\\&\qquad\dots
W_{i_{s-1}i_{s-2}}(\partial)
(-\partial)^{-p_{i_{s-1}}}
W_{ji_{s-1}}(\partial)
(-\partial)^{-p_1}
\,.
\end{split}
\end{equation}
On the other hand, the RHS of \eqref{eq:L2-pr2b}, combined with \eqref{eq:L2-pr3b},
is exactly the $(ij)$-entry of the matrix
$J_1(-(-\partial)^p+W(\partial))^{-1}I_1$,
computed using the geometric series expansion.
The claim follows.
\end{proof}
We can write the matrix $W(\partial)$ in block form as
$W(\partial)
=
\left(\begin{array}{ll}
W_1(\partial) & W_2(\partial) \\
W_3(\partial) & W_4(\partial)
\end{array}\right)$,
where
\begin{equation}\label{eq:blockW}
\begin{split}
& W_1(\partial)
=
\big(
W_{ji}(\partial)
\big)_{1\leq i,j\leq r_1}
\,\,,\,\,\,\,
W_2(\partial)
=
\big(
W_{ji}(\partial)
\big)_{1\leq i\leq r_1<j\leq r}
\,,\\
& W_3(\partial)
=
\big(
W_{ji}(\partial)
\big)_{1\leq j\leq r_1<i\leq r}
\,\,,\,\,\,\,
W_4(\partial)
=
\big(
W_{ji}(\partial)
\big)_{r_1<i,j\leq r}
\,.
\end{split}
\end{equation}
Then, by \cite[Prop.4.2]{DSKVnew} (cf. formula \eqref{eq:spec-quasidet}),
we can rewrite equation \eqref{eq:L1-explicit}
as the following explicit formula for the operator $L_1(\partial)$:
\begin{equation}\label{eq:L1-explicit2}
L_1(\partial)
=
-\id_{r_1}(-\partial)^{p_1}+W_1(\partial)
-W_2(\partial)(-(-\partial)^{q}+W_4(\partial))^{-1}W_3(\partial)
\,,
\end{equation}
where $q=(p_{r_1+1}\geq\dots\geq p_r>0)$ is the partition of $N-r_1p_1$,
obtained by removing from the partition $p$ all the maximal parts.

\section{Summary: 
explicit generators and \texorpdfstring{$\lambda$}{lambda}-brackets for the
\texorpdfstring{$\mc W$}{W}-algebra,
and explicit algorithm for the associated bi-Hamiltonian hierarchy}
\label{sec:6}

Let $f\in\mf{gl}_N$ be a non-zero nilpotent element,
and let $S\in\mf g_d$ be a non-zero element of maximal degree
with respect to the Dynkin grading \eqref{eq:grading} of $\mf{gl}_N$, associated to $f$.
We have the corresponding pencil of Poisson vertex algebras $\mc W_\epsilon(\mf{gl}_N,f,S)$, $\epsilon\in\mb F$,
defined in Section \ref{sec:3.1}.
We summarize below the main results of the paper.

First, as a differential algebra, the $\mc W$-algebra $\mc W(\mf{gl}_N,f)\subset\mc V(\mf g_{\leq\frac12})$
is the algebra of differential polynomials on the generators $\{w_{ij;k}\}$,
parametrized by $1\leq i,j\leq r$, $0\leq k\leq\min\{p_i,p_j\}-1$.
This set of generators is related to the choice of the subspace $U\subset\mf g$ complementary to $[f,\mf g]$
made in Section \ref{sec:5.1}.
Theorem \ref{thm:L1-explicit}
provides a method of computing explicitly all these generators $w_{ij;k}$,
which is obtained by combining equations \eqref{eq:La} and \eqref{eq:L1-explicit} for the operator $L_1(\partial)$
given by
\begin{equation}\label{eq:generators}
\begin{split}
& L_1(\partial) =
|-(-\partial)^p\id_r+W(\partial)|_{I_{rr_1}J_{r_1r}} \\
& =
|\partial\id_N+f+\pi_{\leq\frac12}Q|_{I_1J_1}
\,\in\Mat_{r_1\times r_1}\mc V(\mf g_{\leq\frac12})((\partial^{-1}))
\,,
\end{split}
\end{equation}
where $W(\partial)$ is the matrix (cf. \eqref{eq:Wij} and \eqref{eq:Wd})
$$
W(\partial)
=
\Bigg(
\sum_{k=0}^{\min\{p_i,p_j\}-1}w_{ji;k}(-\partial)^k
\Bigg)_{1\leq i,j\leq r}
\,,
$$
$(-\partial)^p$ is the diagonal matrix with diagonal entries $(-\partial)^{p_i}$, $i=1,\dots,r$,
the matrices $I_{rr_1}$ and $J_{r_1r}$ are as in \eqref{eq:ENM} and \eqref{eq:EMN},
and the matrices $I_1,J_1$ are as in \eqref{eq:factor1}.
Equation \eqref{eq:generators}
defines uniquely the generators $w_{ij;k}$ as elements of the differential algebra $\mc V(\mf g_{\leq\frac12})$
if it is combined with the additional information that $w_{ij;k}-f_{ij;k}$
lies in the differential ideal $\langle U^\perp\rangle\subset\mc V(\mf g_{\leq\frac12})$
generated by $U^\perp$ (cf. Theorem \ref{thm:structure-W} and equations \eqref{eq:basisgf} and \eqref{eq:basisW}).

Furthermore, 
we have a method for computing explicitly 
the $0$-th $\lambda$-bracket $\{\cdot\,_\lambda\,\cdot\}^{\mc W}_0$
between all the generators $w_{ij;k}$'s of the $\mc W$-algebra.
This is provided by the Adler identity \eqref{eq:adler}
for the operator $L_1(\partial)\in\Mat_{r_1\times r_1}\mc W((\partial^{-1}))$
defined by the first equation in \eqref{eq:generators}, namely
\begin{equation}\label{eq:adler-L1}
\begin{split}
\{(L_1)_{ij}(z)_\lambda (L_1)_{hk}(w)\}^{\mc W}_0
& = (L_1)_{hj}(w+\lambda+\partial)\iota_z(z\!-\!w\!-\!\lambda\!-\!\partial)^{-1}((L_1)_{ik})^*(\lambda-z)
\\
& - (L_1)_{hj}(z)\iota_z(z\!-\!w\!-\!\lambda\!-\!\partial)^{-1}(L_1)_{ik}(w)
\,.
\end{split}
\end{equation}
Equation \eqref{eq:adler-L1} is an implicit formula for all the $\lambda$-brackets
between the generators $w_{ij;k}$.
Similarly,
we have an implicit formula for the $1$-st $\lambda$-bracket $\{\cdot\,_\lambda\,\cdot\}^{\mc W}_1$,
depending on the choice of the element $S\in\mf g_d$,
obtained by the condition that the operator $L_1(\partial)+\epsilon\bar S$
is of Adler type with respect to the $\lambda$-bracket
$\{\cdot\,_\lambda\,\cdot\}^{\mc W}_\epsilon=\{\cdot\,_\lambda\,\cdot\}^{\mc W}_0+\epsilon\{\cdot\,_\lambda\,\cdot\}^{\mc W}_1$,
for every $\epsilon\in\mb F$,
where $S\in\mf g_d$ and $\bar S\in\Mat_{r_1\times r_1}\mb F$ are related by \eqref{eq:s0}:
\begin{equation}\label{eq:bi-adler-L1}
\begin{split}
\{(L_1)_{ij}(z)_\lambda (L_1)_{hk}(w)\}^{\mc W}_1
& = 
\bar S_{ik}
\iota_z(z\!-\!w\!-\!\lambda)^{-1}
\big(
(L_1)_{hj}(w+\lambda)
- (L_1)_{hj}(z)
\big)
\\
& +
\bar S_{hj}
\iota_z(z\!-\!w\!-\!\lambda\!-\!\partial)^{-1}
\big(
((L_1)_{ik})^*(\lambda-z)-(L_1)_{ik}(w)
\big)
\,.
\end{split}
\end{equation}

Finally, 
we have the following algorithm to construct an integrable hierarchy of bi-Hamiltonian equations
for the bi-Poisson structure of the family $\mc W_\epsilon(\mf{gl}_N,f,S)$, $\epsilon\in\mb F$:

\begin{enumerate}[1.]
\item
Let $p=(p_1,\dots,p_r)$, with $p_1\geq\dots\geq p_r\ge1$,
be the partition of $N$ associated to the Jordan form of $f$.
Let $r_1\leq r$ be the multiplicity of the largest part $p_1$.
As a differential algebra,
the affine $\mc W$-algebra is
$$
\mc W:=\mc W(\mf{gl}_N,f)
\simeq
\mb F[w_{ij;k}^{(n)}\,\big|\, 1\leq i,j\leq r,\,0\leq k\leq\min\{p_i,p_j\}-1,\,n\in\mb Z_+\big]\,,
$$
the algebra of differential polynomials in the variables $w_{ij;k}$.
\item
Let $W_{ji}(\partial)=\sum_{k=0}^{\min\{p_i,p_j\}-1}w_{ji;k}(-\partial)^k\in\mc W[\partial]$,
for $1\leq i,j\leq r$,
and let 
$W_1(\partial)=\big(W_{ji}(\partial)\big)_{1\leq i,j\leq r_1}$,
$W_2(\partial)=\big(W_{ji}(\partial)\big)_{1\leq i\leq r_1<j\leq r}$,
$W_3(\partial)=\big(W_{ji}(\partial)\big)_{1\leq j\leq r_1<i\leq r}$,
$W_4(\partial)=\big(W_{ji}(\partial)\big)_{r_1<i,j\leq r}$.
Denote by $(-\partial)^q$ the $(r-r_1)\times(r-r_1)$ diagonal matrix
with diagonal entries $(-\partial)^{p_i}$, $r_1<i\leq r$.
Then
\begin{equation}\label{eq:L1-explicit2-b}
L_1(\partial)
=
-\id_{r_1}(-\partial)^{p_1}+W_1(\partial)
-W_2(\partial)(-(-\partial)^q+W_4(\partial))^{-1}W_3(\partial)
\end{equation}
is an $r_1\times r_1$-matrix pseudodifferential operator
of $\bar S$-Adler type for every $\bar S\in\Mat_{r_1\times r_1}\mb F$.
More precisely,
if $S\in\mf g_d$ and $\bar S\in\Mat_{r_1\times r_1}\mb F$
are related by \eqref{eq:s0},
then $L_1(\partial)$ is of $\bar S$-Adler type with respect to the compatible $\lambda$-brackets
of the family of PVAs $\mc W_\epsilon(\mf{gl}_N,f,S)$, $\epsilon\in\mb F$.
\item
Let $\bar r\leq r_1$ be the rank of $\bar S$,
and let $\bar S=\bar I\bar J$ be a factorization of $\bar S$,
with $\bar I\in\Mat_{r_1\times\bar r}\mb F$ and $\bar J\in\Mat_{\bar r\times r_1}\mb F$.
Then the generalized quasideterminant
$$
L(\partial)=|L_1(\partial)|_{\bar I\bar J}
$$
is an $r\times r$ matrix pseudodifferential operator (with entries of coefficients in the field of fractions $\mc K$ of $\mc W$)
of bi-Adler type with respect to the bi-Poisson structure 
of the family $\mc W_\epsilon(\mf{gl}_N,f,S)$, $\epsilon\in\mb F$.
\item
The Hamiltonian densities
$$
h_{0,B}=0
\,\,,\,\,\,\, 
h_{n,B}=-\frac{K}{n}\Res_\partial\tr B(\partial)^n
\,\,,\,\,\,\, n\neq0
\,, 
$$
indexed by $n\in\mb Z_+$
and $B(\partial)\in\Mat_{r\times r}\mc K((\partial^{-1}))$ 
such that $B(\partial)^K=L(\partial)$, $K\geq1$,
are in involution with respect to the $\lambda$-brackets
of the family $\mc W_\epsilon(\mf{gl}_N,f,S)$, $\epsilon\in\mb F$:
$$
\{\tint h_{m,B},\tint h_{n,C}\}_0^{\mc W}=\{\tint h_{m,B},\tint h_{n,C}\}_1^{\mc W}=0
\,\,,\,\,\,\,
m,n\in\mb Z_+,\,B,C \,\text{ roots of }\, L(\partial)\,,
$$
and they satisfy the generalized Lenard-Magri scheme 
on the bi-PVA subalgebra $\mc W_1\subset\mc K$
generated by the coefficients of $L(\partial)$:
$$
\{\tint h_{n,B},w\}_0^{\mc W}=\{\tint h_{n+K,B},w\}_1^{\mc W}
\,\,,\,\,\,\,
\text{ for }\,w\in\mc W_1,\,n\in\mb Z_+
\,\text{ and }\, B(\partial)^K=L(\partial)\,.
$$
Hence, we get the corresponding hierarchy of Hamiltonian equations
$$
\frac{dw}{dt_{n,B}}
=
\{\tint h_{n,B},w\}_0
\,\,,\,\,\,\,w\in\mc W
$$
(bi-Hamiltonian on $\mc W_1$),
which is integrable, provided that the Hamiltonian functionals $\tint h_{n,B}$
span an infinite-dimensional space.
\end{enumerate}

\section{Examples}\label{sec:7}

In this section we show how to apply the methods described in the present paper
by working out in all detail a few examples:
the case of the principal nilpotent element (cf. Section \ref{sec:7.1}),
of a rectangular nilpotent element (cf. Section \ref{sec:7.2}),
of a short nilpotent element (cf. Section \ref{sec:7.3}),
of a minimal nilpotent element (cf. Section \ref{sec:7.4}),
and of a vector and matrix constrained nilpotent element (cf. Section \ref{sec:7.5}).
We find explicit expressions for the generators $w_{ij;k}$ of the $\mc W$-algebras
and their $\lambda$-brackets,
and we describe the associated bi-Hamiltonian integrable hierarchy
for some choice of the element $S\in\mf g_d$.
In all cases we compare our methods and results with the the previous known formulas
which appeared in literature.

\subsection{Example 1: principal nilpotent $f$ and $N$-th KdV hierarchy}\label{sec:7.1}

Consider the trivial partition $p=N$,
corresponding to the principal nilpotent element 
$f^{\text{pr}}=\sum_{i=1}^{N-1}E_{i+1,i}\in\mf{gl}_N$.
For this partition $r=r_1=1$.
Hence, formula \eqref{eq:L1-explicit2-b} gives
(denoting $w_{11;k}=w_k$)
\begin{equation}\label{eq:L1-principal}
L_1(\partial)
=
-(-\partial)^N
+\sum_{k=0}^{N-1}w_k(-\partial)^k
\,,
\end{equation}
the ``generic'' differential operator of order $N$ (which in \cite{DSKVnew} we denoted, up to a sign, by $L_{(N)}(\partial)$),
over the algebra of differential polynomials $\mc W=\mb F[w_k^{(n)}\,|\,k=0,\dots,N-1,\,n\in\mb Z_+]$.

The generators $w_k$, $k=0,\dots,N-1$,
can be computed explicitly, as elements of $\mc V(\mf g_{\leq\frac12})$,
by equation \eqref{eq:generators}.
We get that
\begin{equation}\label{eq:principal-gener}
\begin{split}
L_1(\partial)
=|\id_N\partial+f+\pi_{\leq\frac12}Q|_{1N}
=
q_{N1}-
\left(\begin{array}{llll}
\partial+q_{11}&q_{21}&\dots&q_{N-1,1}
\end{array}\right) \\
\circ
\left(\begin{array}{cccll}
1&\partial\!+\!q_{22}&q_{32}&\dots&q_{N-1,2} \\
0&1&\ddots&\ddots&\vdots \\
\vdots&\ddots&\ddots&\ddots&q_{N-1,N-2} \\
\vdots&&\ddots&\ddots&\partial\!+\!q_{N-1,N-1} \\
0&\dots&\dots&0&1
\end{array}\right)^{-1}
\circ
\left(\begin{array}{l}
q_{N2} \\ \vdots \\ q_{N,N-1} \\ \partial+q_{NN}
\end{array}\right)
\,.
\end{split}
\end{equation}
Here we used the usual formula for the quasideterminant of a matrix,
cf. \cite[Prop.4.2]{DSKVnew}.
We can expand the inverse matrix in the RHS of \eqref{eq:principal-gener}
in geometric series, to get the following more explicit equation
for all the generators $w_k$ of the $\mc W$-algebra $\mc W(\mf{gl}_N,f^{\text{pr}})\subset\mc V(\mf g_{\leq\frac12})$:
\begin{equation}\label{eq:principal-gener-explicit}
\begin{split}
& -(-z)^N+\sum_{k=1}^{N-1}w_k(-z)^k 
=
q_{N1}+
\sum_{s=1}^{N-1}(-1)^s \\
& \sum_{2\leq h_1<\dots<h_s\leq N}
(\delta_{h_1-1,1}(z+\partial)+q_{h_1-1,1})
(\delta_{h_2-1,h_1}(z+\partial)+q_{h_2-1,h_1})
\dots \\
& \qquad\qquad\dots
(\delta_{h_s-1,h_{s-1}}(z+\partial)+q_{h_s-1,h_{s-1}})
(\delta_{N,h_s}z+q_{N,h_s})
\,.
\end{split}
\end{equation}
For example, $w_{N-1}=\tr Q=q_{11}+\dots+q_{NN}$.
This formula for the generators of $\mc W(\mf{gl}_N,f^{\text{pr}})$
agrees with the results in \cite{MR15}.

Furthermore, 
the compatible $\lambda$-brackets among the generators $w_k$, $k=0,\dots,N-1$, 
of the $\mc W$-algebra $\mc W(\mf{gl}_N,f^{\text{pr}})$ 
are obtained by the Adler identities \eqref{eq:adler-L1} and \eqref{eq:bi-adler-L1}.
They are
\begin{equation}\label{eq:20150911-1}
\begin{split}
& \{{w_i}_\lambda{w_j}\}^{\mc W}_0
=
\sum_{n=0}^{N-i-1}
\sum_{a=\max\{0,n-j\}}^{N+n-j}
\Bigg(
\binom{n}{a}(-1)^{a}w_{n+i+1}(\lambda+\partial)^{a}w_{j+a-n} \\
& +\!\!\!
\sum_{b=0}^{N-n-i-1}
\!\!\!
\binom{j+a}{a}
\binom{i\!+\!n\!+\!b\!+\!1}{b}
(-1)^{a+1}
w_{j+a-n}(\lambda+\partial)^{a+b}w_{i+n+b+1}
\Bigg)\,,
\end{split}
\end{equation}
and
$$
\{{w_i}_\lambda{w_j}\}^{\mc W}_1
=
\sum_{n=0}^{N-i-j-1}
\bigg(
\binom{n+j}{j}(-\lambda)^n
-
\binom{n+i}{i}(\lambda+\partial)^n
\bigg)
w_{i+j+n+1}
\,,
$$
where, in the RHS of both formulas, we let $w_N=-1$.
These formulas agree
with the corresponding formulas in \cite[Sec.2.4]{DSKV15},
after the change of variable $w_i=(-1)^iu_{-i-1}$,
and up to the overall sign related to the different choice of sign in the definition of Adler type operators.

The corresponding integrable hierarchy is the $N$-th Gelfand-Dickey hierarchy \cite{GD76}
$$
\frac{dL_1(\partial)}{dt_k}
=
[L_1(\partial)^{\frac kN}_+,L_1(\partial)]
\,,\,\, k\in\mb Z_+\,.
$$
It is easy to show
that these equations with $k\notin N\mb Z_+$ are linearly independent, see \cite[Sec.3.2]{DSKV15}.

\subsection{Example 2: rectangular nilpotent and matrix $n$-th KdV hierarchy}\label{sec:7.2}

Consider the partition $p=(p_1,\dots,p_1)$ of $N$,
consisting of $r_1$ equal parts of size $p_1$.
It corresponds to the so called rectangular nilpotent element $f$.
For this choice \eqref{eq:L1-explicit2} becomes
\begin{equation}\label{eq:L1-rectangular}
L_1(\partial)
=
-\id_{r_1}(-\partial)^{p_1}
+\sum_{k=0}^{p_1-1}W_k(-\partial)^k
\,\,,\,\,\,\,
W_k=\big(w_{ji;k}\big)_{1\leq i,j\leq r_1}\in\Mat_{r_1\times r_1}\mc W
\,,
\end{equation}
the ``generic'' $r_1\times r_1$ matrix differential operator of order $p_1$ 
(which, up to a sign, in \cite{DSKVnew} we denoted by $L_{(p_1r_1)}(\partial)$, cf. \cite[Ex.3.5]{DSKVnew}).

The generators $w_{ji;k}$, $1\leq i,j\leq r_1$, $0\leq k\leq p_1-1$,
can be computed explicitly, as elements of $\mc V(\mf g_{\leq\frac12})$,
by equation \eqref{eq:generators}.
We identify 
\begin{equation}\label{eq:rectangular-factors}
\Mat_{N\times N}\mb F
\simeq
\Mat_{p_1\times p_1}\mb F\otimes\Mat_{r_1\times r_1}\mb F
\,,
\end{equation}
by mapping $E_{(ih),(jk)}\mapsto E_{hk}\otimes E_{ij}$.
Under this identification, we have
\begin{equation*}
\begin{split}
& \id_N\mapsto\id_{p_1}\otimes\id_{r_1}
\,\,,\,\,\,\,
f\mapsto \sum_{k=1}^{p_1-1}E_{k+1,k}\otimes \id_{r_1}
\,,\\
& \pi_{\leq\frac12}Q\mapsto
\sum_{i,j=1}^{r_1}\sum_{1\leq h\leq k\leq p_1}q_{(jk),(ih)}E_{hk}\otimes E_{ij}
\,.
\end{split}
\end{equation*}
Hence, according to \eqref{eq:generators}, we have, by the usual formulas for quasideterminants
(cf. \cite[Prop.4.2]{DSKVnew}),
\begin{equation}\label{eq:rectangular-gener}
\begin{split}
& L_1(\partial)
=
\Big|
(\id_{p_1}\otimes\id_{r_1})\partial
+\sum_{k=1}^{p_1-1}E_{k+1,k}\otimes\id_{r_1}
+\sum_{i,j=1}^{r_1}\sum_{1\leq h\leq k\leq p_1}q_{(jk),(ih)}E_{hk}\otimes E_{ij}
\Big|_{I_1J_1} \\
& =\!
\sum_{i,j=1}^{r_1}
q_{(jp_1),(i1)}E_{ij}
-
\Bigg(\!
\sum_{i,j=1}^{r_1}
\left(\begin{array}{llll}
\!\!\delta_{ij}\partial\!+\!q_{(j1),(i1)}&q_{(j2),(i1)}&\dots&q_{(j,p_1-1),(i1)}\!\!
\end{array}\right)\otimes E_{ij}
\!\Bigg)
\\
& \circ
\Bigg(\!
\id_{p_1-1}\!\otimes\!\id_{r_1}+
\sum_{i,j=1}^{r_1}
\left(\begin{array}{rccl}
0&\delta_{ij}\partial\!\!+\!\!q_{(j2),(i2)}&\dots&q_{(j,p_1-1),(i2)} \\
\vdots&\ddots&\ddots&\vdots \\
\vdots&&\ddots&\delta_{ij}\partial\!\!+\!\!q_{(j,p_1-1),(i,p_1-1)} \\
0&\dots&\dots&0
\end{array}\right)\!\otimes\! E_{ij}
\!\Bigg)^{-1} \\
& \circ
\Bigg(
\sum_{i,j=1}^{r_1}
\left(\begin{array}{l}
q_{(jp_1),(i2)} \\ \vdots \\ q_{(jp_1),(i,p_1-1)} \\ \delta_{ij}\partial+q_{(jp_1),(ip_1)}
\end{array}\right)
\otimes E_{ij}
\Bigg)
\,.
\end{split}
\end{equation}
As we did in Section \ref{sec:7.1},
we expand the inverse matrix in the RHS of \eqref{eq:rectangular-gener}
in geometric series, to get a more explicit equation
for all the generators $w_{ji;k}$, $1\leq i,j\leq r_1$, $0\leq k\leq p_1-1$,
of the $\mc W$-algebra $\mc W(\mf{gl}_N,f^{\text{rect}})\subset\mc V(\mf g_{\leq\frac12})$.
Equating the $(ij)$-entry of the RHS's of equations \eqref{eq:L1-rectangular} and \eqref{eq:rectangular-gener},
we get
\begin{equation}\label{eq:rectangular-gener-explicit}
\begin{split}
& -(-z)^{p_1}\delta_{ij}+\sum_{k=0}^{p_1-1}w_{ji;k}(-z)^k 
=
q_{(jp_1),(i1)}+
\sum_{s=1}^{p_1-1}(-1)^s
\sum_{i_1,\dots,i_s=1}^{r_1}
\sum_{2\leq h_1<\dots<h_s\leq p_1} \\
& (\delta_{i_1,i}\delta_{h_1-1,1}(z+\partial)+q_{(i_1,h_1-1),(i1)})
(\delta_{i_2,i_1}\delta_{h_2-1,h_1}(z+\partial)+q_{(i_2,h_2-1),(i_1h_1)})
\dots \\
& \dots
(\delta_{i_s,i_{s-1}}\delta_{h_s-1,h_{s-1}}(z+\partial)+q_{(i_s,h_s-1),(i_{s-1},h_{s-1})})
(\delta_{i_s,j}\delta_{p_1,h_s}z+q_{(jp_1),(i_sh_s)})
\,.
\end{split}
\end{equation}
For example, $w_{ji;p_1-1}=q_{(j1),(i1)}+\dots+q_{(jp_1),(ip_1)}$.

Furthermore, 
we can compute the compatible $\lambda$-brackets between the generators $w_{ji;k}$, 
$1\leq i,j\leq r_1$, $k=0,\dots,p_1-1$, 
of the $\mc W$-algebra $\mc W(\mf{gl}_N,f^{\text{pr}})$.
By \eqref{eq:L1-rectangular}, we have
$$
\{{w_{\beta\alpha;h}}_\lambda{w_{\delta\gamma;k}}\}^{\mc W}_{\epsilon}
=
\Res_z\Res_w(-1)^{h+k}z^{-h-1}w^{-k-1}
\{(L_1)_{\alpha\beta}(z)_\lambda(L_1)_{\gamma\delta}(w)\}^{\mc W}_{\epsilon}
\,.
$$
Hence, by the Adler identities \eqref{eq:adler-L1} 
we get the $0$-th $\lambda$-bracket
($1\leq\alpha,\beta,\gamma,\delta\leq r_1$, $0\leq h,k\leq p_1-1$):
\begin{equation}\label{eq:rectangular-0}
\begin{split}
& \{{w_{\beta\alpha;h}}_\lambda{w_{\delta\gamma;k}}\}^{\mc W}_0 
=
\sum_{n=0}^{p_1-h-1}
\sum_{a=\max\{0,n-k\}}^{p_1+n-k}
\Bigg(
\binom{n}{a}(-1)^{a}w_{\beta\gamma;h+n+1}(\lambda+\partial)^{a}w_{\delta\alpha;k+a-n} \\
& + \sum_{b=0}^{p_1-n-h-1}
\!\!\!
\binom{h\!+\!n\!+\!b\!+\!1}{b}\binom{k+a}{a}(-1)^{a+1}
w_{\beta\gamma;k+a-n}(\lambda+\partial)^{a+b}w_{\delta\alpha;h+n+b+1}
\Bigg)\,,
\end{split}
\end{equation}
where, in the RHS, we let $w_{ji;p_1}=-\delta_{i,j}$.
To compute the $1$-st Poisson structure of the family of PVAs $\mc W_\epsilon(\mf{gl}_N,f,S)$, $\epsilon\in\mb F$,
we fix an element $S\in\mf g_d$ or, equivalently, 
a matrix $\bar S\in\Mat_{r_1\times r_1}\mb F$.
The corresponding $1$-st $\lambda$-bracket is obtained by \eqref{eq:bi-adler-L1}:
\begin{equation}\label{eq:rectangular-1}
\begin{split}
& \{{w_{\beta\alpha;h}}_\lambda{w_{\delta\gamma;k}}\}^{\mc W}_1
=
\sum_{n=0}^{p_1-h-k-1}
\binom{n+k}{k}
\bar S_{\alpha\delta}w_{\beta\gamma;h+k+n+1}
(-\lambda)^n \\
& -
\sum_{n=0}^{p_1-h-k-1}
\binom{n+h}{h}(\lambda+\partial)^n
\bar S_{\gamma\beta}w_{\delta\alpha;h+k+n+1}
\,.
\end{split}
\end{equation}
These formulas agree
with the corresponding formulas in \cite[Sec.4.2]{DSKV15}
(but there we only considered the case $\bar S=\id_{r_1}$).

The integrable hierarchy corresponding to the Adler type operator \eqref{eq:L1-rectangular}
is the $r_1\times r_1$ matrix $p_1$-th KdV hierarchy.
This is also a well known hierarchy, see e.g. \cite{DSKV15}
for the bi-Poisson structure and the bi-Hamiltonian equations,
similar to the scalar case $r_1=1$, discussed in the previous section.

However, for $r_1>1$ there are more possibilities for constructing
bi-Adler type operators, choosing different $S\in\mf g_d$:
if $\bar S\in\Mat_{r_1\times r_1}\mb F$ is the matrix corresponding to $S$ via \eqref{eq:s0}
and $\bar S=\bar I\bar J$ is its canonical factorization,
with $\bar I\in\Mat_{r_1\times\bar r}\mb F$ and $\bar J\in\Mat_{\bar r\times r_1}\mb F$
($\bar r$ is the rank of $\bar S$),
then $L(\partial)=|L_1(\partial)|_{\bar I\bar J}$ is the operator of bi-Adler type
with respect to the compatible $\lambda$-brackets of the family of PVAs $\mc W_\epsilon(\mf{gl}_N,f,S)$, $\epsilon\in\mb F$.
If $\bar S$ is non-degenerate (i.e. $\bar r=r_1$), then $L(\partial)=S^{-1}L_1(\partial)$
(up to conjugation by a constant non-degenerate matrix),
while if $\bar S$ is degenerate, then $L(\partial)$ is a non-trivial generalized quasideterminant.

\subsection{Example 3: short nilpotent $f$}\label{sec:7.3}

The short nilpotent element $f$ in $\mf g=\mf{gl}_N$, with even $N$,
is associated to the partition $p=(2,\dots,2)$, consisting of $r=\frac{N}{2}$ parts equal to $2$.
So, it is the special case of Example 2 from Section \ref{sec:7.2}
when $p_1=2$.
In this case, the Adler type operator $L_1(\partial)\in\Mat_{r\times r}\mc W[\partial]$ 
in \eqref{eq:L1-explicit2} has the form
\begin{equation}\label{eq:L1-short}
L_1(\partial)
=
-\id_{r}\partial^2-W_1\partial+W_0
\,\,,\,\,\,\,
W_k=\big(w_{ji;k}\big)_{1\leq i,j\leq r}\in\Mat_{r\times r}\mc W
\,,
\end{equation}

Equation \eqref{eq:rectangular-gener-explicit} for the generators $w_{ji;k}$ of $\mc W(\mf g,f)$,
$i,j=1,\dots,r$, $k=0,1$, as elements of the differential algebra $\mc V(\mf g_{\leq\frac12})$ gives,
in this case,
\begin{equation}\label{eq:short-gener-explicit}
\begin{split}
& w_{ji;1}
=
q_{(j1),(i1)}+q_{(j2),(i2)} \,,\\
& w_{ji;0}
=
q_{(j2),(i1)}-q_{(j2),(i2)}^\prime
-\sum_{k=1}^rq_{(k1),(i1)}q_{(j2),(k2)}
\,.
\end{split}
\end{equation}

In order to compare these generators with the formulas
for the generators of $\mc W(\mf g,f)$ found in \cite[Th.4.2]{DSKV14a},
we first introduce some notation.
For $h,k=1,2$, we define the maps 
$\bar q_{kh}:\,\mf{gl}_r\hookrightarrow\mf{gl}_N$
given by
$$
\bar q_{kh}(q_{ji})=q_{(jk),(ih)}
\,\,,\,\,\,\,i,j=1,\dots,r\,.
$$
We also denote $\bar\id=\bar q_{11}+\bar q_{22}$ and $\bar h=\bar q_{11}-\bar q_{22}$.
With this notation, we have
$$
\mf g_{-1}=\bar q_{21}(\mf{gl}_r)
\,\,,\,\,\,\,
\mf g_{1}=\bar q_{12}(\mf{gl}_r)
\,\,,\,\,\,\,
\mf g^f_0=\bar\id(\mf{gl}_r)
\,\,,\,\,\,\,
[e,\mf g_{-1}]=[f,\mf g_1]=\bar h(\mf{gl}_r)
\,.
$$

Recall that, in the present paper, we consider the subspace $U\subset\mf g$ complementary to $[f,\mf g]$
as in \eqref{eq:basisU}.
We have the corresponding direct sum decomposition
$\mf g=\mf g^f\oplus U^{\perp}$, where
$$
\mf g^f
=
\bar q_{21}(\mf{gl}_r)\oplus\bar\id(\mf{gl}_r)
\,\text{ and }\,
U^{\perp}
=
\bar q_{12}(\mf{gl}_r)\oplus\bar q_{22}(\mf{gl}_r)
\,.
$$
The corresponding quotient map $\mf g\to\mf g^f$
induces the differential algebra homomorphism $\pi:\,\mc V(\mf g_{\leq\frac12})\to\mc V(\mf g^f)$
which, by Theorem \ref{thm:structure-W},
restricts to an isomorphism $\mc W(\mf g,f)\stackrel{\sim}{\rightarrow}\mc V(\mf g^f)$,
and we denote by $w:\,\mc V(\mf g^f)\stackrel{\sim}{\rightarrow}\mc W(\mf g,f)$
the inverse map.
It will be convenient to denote $w_1=w\circ\bar\id$ and $w_0=w\circ\bar q_{21}:\,\mf{gl}_r\to\mc W$,
so that 
$$
w_k(q_{ji})=w_{ji;k}
\,\text{ for every } i,j=1,\dots,r\,,\,\,\,\,k=0,1
\,.
$$
On the other hand,
in \cite{DSKV14a}, as a complementary subspace to $[f,\mf g]$ in $\mf g$ we chose $\mf g^e$,
its orthocomplement being $[e,\mf g]$.
We have the corresponding vector space decompositions
$$
\mf g=\mf g^f\oplus[e,\mf g]
\,\,\text{ and }\,\,
[e,\mf g]=\mf g_1\oplus[e,\mf g_{-1}]
=\bar q_{12}(\mf{gl}_r)\oplus\bar h(\mf{gl}_r)
\,.
$$
We denote by $\sharp$ the corresponding quotient map $\mf g\twoheadrightarrow\mf g^f$.
In particular,
\begin{equation}\label{eq:sharp}
\sharp\circ\bar q_{21}=\bar q_{21}
\,\,,\,\,\,\,
\sharp\circ\bar q_{11}=\sharp\circ\bar q_{22}=\frac12\bar\id
\,\,,\,\,\,\,
\sharp\circ\bar q_{12}=0
\,.
\end{equation}
Again by Theorem \ref{thm:structure-W},
the corresponding homomorphism of differential algebras $\mc V(\mf g_{\leq\frac12})\to\mc V(\mf g^f)$
restricts to an isomorphism $\mc W(\mf g,f)\stackrel{\sim}{\rightarrow}\mc V(\mf g^f)$,
and we denote by $\psi:\,\mc V(\mf g^f)\stackrel{\sim}{\rightarrow}\mc W(\mf g,f)$
the inverse map.

In the present paper we consider the generators 
$w_{1}(a)=w\bar\id(a)$ and $w_0(a)=w\bar q_{21}(a)$, $a\in\mf{gl}_r$,
while in \cite{DSKV14a} we considered the alternative set of generators
$\psi_{1}(a)=\psi(\bar\id(a))$ 
and $\psi_{0}=\psi(\bar q_{21}(a))$, $a\in\mf{gl}_r$.
Formulas \eqref{eq:short-gener-explicit} for the generators $w_{ji;k}$ can be rewritten as
\begin{equation}\label{eq:short-gener-linear}
w_1(a)
=
\bar\id(a)
\,\,,\,\,\,\,
w_0(a)
=
\bar q_{21}(a)-\bar q_{22}(a)^\prime
-\sum_{h,k=1}^r\bar q_{11}(E_{kh}a)\bar q_{22}(E_{hk})
\,.
\end{equation}
Applying the maps $\sharp$ and $\psi$ to both equations in \eqref{eq:short-gener-linear},
and using \eqref{eq:sharp},
we find the relation between these generators and those in \cite{DSKV14a}:
\begin{equation}\label{eq:short-gen-compare}
\begin{split}
& w_1(a)
=
\psi_1(a)
=
\bar\id(a)
\,,\\
& w_0(a)
=
\psi_0(a)-\frac12\bar\id(a)^\prime
-\frac14\sum_{h,k=1}^r\bar\id(E_{kh}a)\bar\id(E_{hk})
\,.
\end{split}
\end{equation}
It is now straightforward to check that
equations \eqref{eq:short-gener-linear} and \eqref{eq:short-gen-compare}
agree with \cite[Eq.(4.4)]{DSKV14a}.

We can also compute the $\lambda$-brackets between the generators \eqref{eq:short-gener-linear}.
Equation \eqref{eq:rectangular-0} gives ($a,b\in\mf{gl}_r$)
\begin{equation}\label{eq:short-0}
\begin{split}
& \{w_1(a)_\lambda w_1(b)\}^{\mc W}_0
=
w_1([a,b])+2\tr(ab)\lambda
\,,\\
& \{w_1(a)_\lambda w_0(b)\}^{\mc W}_0
=
w_0([a,b])-w_1(ab)\lambda-\tr(ab)\lambda^2
\,,\\
& \{w_0(a)_\lambda w_0(b)\}^{\mc W}_0
=
\sum_{i,j=1}^r\Big(
w_1(aE_{ij})w_0(bE_{ji})
-w_0(aE_{ij})w_1(bE_{ji}) \\
&\qquad
+w_1(aE_{ij})(\lambda+\partial)w_1(bE_{ji})
\Big)
+(\lambda+\partial)w_0(ba)+w_0(ab)\lambda \\
&\qquad
+(\lambda+\partial)^2w_1(ba)-w_1(ab)\lambda^2
-\tr(ab)\lambda^3
\,.
\end{split}
\end{equation}
It is straightforward to check that equations \eqref{eq:short-0}
agree with the analogous equations in \cite[Thm.4.4]{DSKV14a},
if we take into account formula \eqref{eq:short-gen-compare}
for the change of generators.
Note that the coefficient $2$ in first equation of \eqref{eq:short-0}
comes from the fact that $\tr_{\mf{gl}_N}(\bar\id(a)\bar\id(b))=2\tr_{\mf{gl}_r}(ab)$.
Similarly, equation \eqref{eq:rectangular-1} gives
\begin{equation}\label{eq:short-1}
\begin{split}
& \{w_1(a)_\lambda w_1(b)\}^{\mc W}_1
=0
\,,\\
& \{w_1(a)_\lambda w_0(b)\}^{\mc W}_1
=\tr(\bar S[a,b])
\,,\\
& \{w_0(a)_\lambda w_0(b)\}^{\mc W}_1
=
w_1(a\bar Sb-b\bar Sa)+\tr(\bar S\{a,b\})\lambda
\,,
\end{split}
\end{equation}
where $\{a,b\}=ab+ba$ is the anticommutator in $\mf{gl}_r$.
Again, one can check that equations \eqref{eq:short-1}
agree with the analogous equations in \cite[Thm.4.4]{DSKV14a}.

\subsection{Example 4: minimal nilpotent $f$}\label{sec:7.4}

The minimal nilpotent element $f$ in $\mf g=\mf{gl}_N$
is associated to the partition $p=(2,1,\dots,1)$.
In this case, $L_1(\partial)$ is the following scalar pseudodifferential operator:
\begin{equation}\label{eq:L1-minimal}
L_1(\partial)
=
-\partial^2-w_{11;1}\partial+w_{11;0} \\
- w_{+1}(\id_{N-2}\partial+W_{++})^{-1}w_{1+}
\,,
\end{equation}
where
\begin{equation}\label{eq:minimal-notation}
\begin{split}
& w_{+1}
=
\left(\begin{array}{lll} w_{21;0} & \dots & w_{r1;0} \end{array}\right) 
\,,\\
& W_{++}
=
\left(\begin{array}{ccc} 
w_{22;0} & \dots & w_{r2;0} \\
\vdots & \ddots & \vdots \\
w_{2r;0} & \dots & w_{rr;0} 
\end{array}\right)
\,\,,\,\,\,\,
w_{1+}
=
\left(\begin{array}{l} w_{12;0} \\ \vdots \\ w_{1r;0} \end{array}\right)
\,.
\end{split}
\end{equation}

In order to find a formula for the generators of $\mc W(\mf g,f)$,
we need to compute $L_1(\partial)$ using the definition \eqref{eq:La}
and equate the result to \eqref{eq:L1-minimal}.
For our choice of $f$, \eqref{eq:La} becomes
\begin{equation}\label{eq:L1-minimal-def}
L_1(\partial)
=
\Big|
\left(\begin{array}{ccc}
\partial+q_{(11),(11)} & q_{(12),(11)} & q_{+(11)} \\
1 & \partial+q_{(12),(12)} & q_{+(12)} \\
q_{(11)+} & q_{(12)+} & \id_{N-2}\partial+Q_{++}
\end{array}\right)
\Big|_{12}
\,,
\end{equation}
where ($k=1,2$)
\begin{equation*}
\begin{split}
& q_{+(1k)}
=
\left(\begin{array}{lll} q_{(21),(1k)} & \dots & q_{(r1),(1k)} \end{array}\right) 
\,,\\
& Q_{++}
=
\left(\begin{array}{ccc} 
q_{(21),(21)} & \dots & q_{(r1),(21)} \\
\vdots & \ddots & \vdots \\
q_{(21),(r1)} & \dots & q_{(r1),(r1)} 
\end{array}\right)
\,\,,\,\,\,\,
q_{(1k)+}
=
\left(\begin{array}{l} q_{(1k),(21)} \\ \vdots \\ q_{(1k),(r1)} \end{array}\right)
\,.
\end{split}
\end{equation*}
We can compute the quasideterminant \eqref{eq:L1-minimal-def} by the usual formula, see \cite[Prop.4.2]{DSKVnew}.
As a result we get, after a straightforward computation,
\begin{equation}\label{eq:L1-minimal-gener}
\begin{split}
&\vphantom{\Big(}
L_1(\partial)
=
q_{(12),(11)}
-q_{+(11)}(\partial+Q_{++})^{-1}\circ q_{(12)+}
-\Big(\partial+q_{(11),(11)} \\
&\vphantom{\Big(}
-q_{+(11)}(\partial+Q_{++})^{-1}\circ q_{(11)+}\Big)\circ 
\Big(1-q_{+(12)}(\partial+Q_{++})^{-1}\circ q_{(11)+}\Big)^{-1}\circ \\
&\vphantom{\Big(}
\Big(\partial+q_{(12),(12)}-q_{+(12)}(\partial+Q_{++})^{-1}\circ q_{(12)+}\Big)
\,.
\end{split}
\end{equation}
In order to find the explicit formula for the generators, we need to equate \eqref{eq:L1-minimal}
and \eqref{eq:L1-minimal-gener}.
In fact, to find generators, it suffices to find the first few terms in the expansions
of the pseudodifferential operators  \eqref{eq:L1-minimal} and \eqref{eq:L1-minimal-gener}
and equate them.
From \eqref{eq:L1-minimal} we get
\begin{equation}\label{eq:L1-minimal-expansion}
L_1(\partial)
=
-\partial^2-w_{11;1}\partial+w_{11;0}
-w_{+1}w_{1+}\partial^{-1}
+\Big(w_{+1}w_{1+}^\prime+w_{+1}W_{++}w_{1+}\Big)\partial^{-2}+\dots
\,.
\end{equation}
Finding the expansion of \eqref{eq:L1-minimal-gener} up to order $\partial^{-2}$,
we find the expression of the generators.
This is a rather long computation of which we omit the details.
The answer is
\begin{equation}\label{eq:minimal-gener}
\begin{split}
&\vphantom{\Big(}
w_{11;1}=q_{(11),(11)}+q_{(12),(12)}+q_{+(12)}q_{(11)+} \,,\\
&\vphantom{\Big(}
w_{11;0}=q_{(12),(11)}+q_{+(11)}q_{(11)+}+q_{+(12)}q_{(12)+}
-q_{(12),(12)}^\prime-q_{(11),(11)}q_{(12),(12)} \\
&\vphantom{\Big(}\qquad 
-w_{11;1}q_{+(12)}q_{(11)+}
+q_{+(12)}Q_{++}q_{(11)+}-q_{+(12)}^\prime q_{(11)+} \,,\\
&\vphantom{\Big(}
w_{+1}=q_{+(11)}+q_{+(12)}Q_{++}-q_{+(12)}^\prime
-q_{(11),(11)}q_{+(12)}
-q_{+(12)}q_{(11)+}q_{+(12)} \,,\\
&\vphantom{\Big(}
w_{1+}=q_{(12)+}+Q_{++}q_{(11)+}+q_{(11)+}^\prime
-q_{(12),(12)}q_{(11)+}
-q_{(11)+}q_{+(12)}q_{(11)+} \,,\\
&\vphantom{\Big(}
W_{++}=Q_{++}-q_{(11)+}q_{+(12)} \,.
\end{split}
\end{equation}

As we did in the case of a short nilpotent in Section \ref{sec:7.3},
we want to compare these generators with the analogous formulas
in \cite{DSKV14a},
The $\ad x$-eigenspace decomposition \eqref{eq:grading} is, in this case,
\begin{equation*}
\begin{split}
& \mf g_{-1}=\mb F q_{(12),(11)}
\,,\,\,
\mf g_{-\frac12}= \Span\{ q_{(12),(j1)},\,q_{(j1),(11)} \}_{j=2}^{N-2} \,,\\
& \mf g_0=\mb F q_{(11),(11)}\oplus\mb F q_{(12),(12)}\oplus\Span\{ q_{(j1),(i1)} \}_{i,j=1}^{N-2} \,,\\
& \mf g_{+\frac12}= \Span\{ q_{(11),(j1)},\,q_{(j1),(12)} \}_{j=2}^{N-2}
\,,\,\,
\mf g_{+1}=\mb F q_{(11),(12)} 
\,.
\end{split}
\end{equation*}
Note that 
$$
\mf g^f=\mf g_{-1}\oplus\mf g_{-\frac12}\oplus\mb F\bar\id\oplus\Span\{ q_{(j1),(i1)} \}_{i,j=1}^{N-2}
\,,
$$
where $\bar\id=q_{(11),(11)}+q_{(12),(12)}$.
As complementary subspaces to $\mf g^f$ we have two different choices:
$$
U^\perp = \mb F q_{(12),(12)}\oplus\mf g_{+\frac12}\oplus\mf g_{+1}
\,,
$$
that we used in the present paper,
and 
$$
[e,\mf g] = \mb F q_{(12),(12)}\oplus\mf g_{+\frac12}\oplus\mf g_{+1}
\,,
$$
that we used in \cite{DSKV14a}.
Let us denote, as in Section \ref{sec:7.3},
by $\pi:\,\mf g\to\mf g^f$ the quotient map with kernel $U^\perp$,
which induces a differential algebra isomorphism $\mc W(\mf g,f)\simeq\mc V(\mf g^f)$,
whose inverse is the map $w$,
and by $\sharp:\,\mf g\to\mf g^f$ the quotient map with kernel $[e,\mf g]$,
which induces a different isomorphism $\mc W(\mf g,f)\simeq\mc V(\mf g^f)$,
whose inverse we denote by $\psi$.
Applying the maps $\sharp$ and $\psi$ to all formulas \eqref{eq:minimal-gener},
we express all the generators $w_{ji;k}=w(f_{ji;k})$ (cf. \eqref{eq:basisgf}-\eqref{eq:basisW})
in terms of the alternative generators $\psi_{ji;k}:=\psi(f_{ji;k})$.
We have
\begin{equation}\label{eq:minimal-gener-compare}
\begin{split}
& w_{11;1}
=
\psi(\bar\id)
\,,\,\,
w_{11;0}=\psi(f)-\frac12\psi(\bar\id)^\prime-\frac14\psi(\bar\id)^2
\,,\\
& w_{+1}=\psi(q_{+(11)})
\,,\,\,
w_{1+}=\psi(q_{(12)+})
\,,\,\,
W_{++}=\psi(Q_{++})
\,.
\end{split}
\end{equation}
It is now straightforward to check that
equations \eqref{eq:minimal-gener} and \eqref{eq:minimal-gener-compare}
agree with \cite[Thm.3.2]{DSKV14a}.

We can also compute the $0$-th and $1$-st $\lambda$-brackets 
$\{\cdot\,_\lambda\,\cdot\}^{\mc W}_0$ and $\{\cdot\,_\lambda\,\cdot\}^{\mc W}_1$
between generators,
either by using the explifit formulas \eqref{eq:minimal-gener} of the generators,
or by combining the expression \eqref{eq:L1-minimal} of $L_1(\partial)$
with the Adler identity \eqref{eq:adler-L1}
and the bi-Adler identity \eqref{eq:bi-adler-L1} respectively.
This is a long computation that we are going to make, in all detail,
in the next Section in the more general case of a ``vector constrained'' nilpotent element $f$.
We here report only the resulting formulas for the two PVA structures:
\begin{equation}\label{eq:minimal-0}
\begin{split}
& \{{w_{11;1}}_\lambda{w_{11;1}}\}^{\mc W}_0
=
2\lambda
\,,\\
& \{{w_{11;1}}_\lambda{w_{11;0}}\}^{\mc W}_0
=
-w_{11;1}\lambda-\lambda^2
\,,\,\,
\{{w_{11;0}}_\lambda{w_{11;1}}\}^{\mc W}_0
=
-(\lambda+\partial)w_{11;1}+\lambda^2
\,,\\
& \{{w_{11;0}}_\lambda{w_{11;0}}\}^{\mc W}_0
=
(\partial+2\lambda)w_{11;0}+w_{11;1}(\lambda+\partial)w_{11;1}+(\partial+2\lambda)\partial w_{11;1}-\lambda^3
\,,\\
& \{{w_{11;1}}_\lambda{w_{+1}}\}^{\mc W}_0
=
-\{{w_{+1}}_\lambda{w_{11;1}}\}^{\mc W}_0
=
-w_{+1}
\,,\\
& \{{w_{11;1}}_\lambda{w_{1+}}\}^{\mc W}_0
=
-\{{w_{1+}}_\lambda{w_{11;1}}\}^{\mc W}_0
=
w_{1+}
\,,\\
& \{{w_{11;0}}_\lambda{w_{+1}}\}^{\mc W}_0
=
(\partial+\lambda)w_{+1}-w_{+1}W_{++}+w_{11;1}w_{+1}
\,,\\
& \{{w_{+1}}_\lambda{w_{11;0}}\}^{\mc W}_0
=
w_{+1}\lambda+w_{+1}W_{++}-w_{11;1}w_{+1}
\,,\\
& \{{w_{11;0}}_\lambda{w_{1+}}\}^{\mc W}_0
=
(\partial+2\lambda)w_{1+}+W_{++}w_{1+}-w_{11;1}w_{1+}
\,,\\
& \{{w_{1+}}_\lambda{w_{11;0}}\}^{\mc W}_0
=
(\partial+2\lambda)w_{1+}-W_{++}w_{1+}+w_{11;1}w_{1+}
\,,\\
& \{{w_{+1}}_\lambda{w_{+1}}\}^{\mc W}_0
=
\{{w_{1+}}_\lambda{w_{1+}}\}^{\mc W}_0
=0 \,\\
& \{{w_{1+}}_\lambda{w_{+1}}\}^{\mc W}_0
=
-(\lambda+\partial+w_{11;1}-W_{++})(\lambda-W_{++})+w_{11;0}\id_{N-2}
\,,\\
& \{{w_{+1}^T}_\lambda{w_{1+}^T}\}^{\mc W}_0
=
(\lambda+\partial+W_{++}^T)(\lambda-w_{11;1}+W_{++}^T)-w_{11;0}\id_{N-2}
\,,\\
& \{{w_{11;k}}_\lambda{W_{++}}\}^{\mc W}_0
=
\{{W_{++}}_\lambda{w_{11;k}}\}^{\mc W}_0
=
0\,,\,\,k=0,1\,,\\
& \{{(w_{+1})_k}_\lambda{(W_{++})_{ij}}\}^{\mc W}_0
=-\delta_{ik}(w_{+1})_j
=-\{{(W_{++})_{ij}}_\lambda{(w_{+1})_k}\}^{\mc W}_0
\,,
\\
&
\{{(w_{1+})_k}_\lambda{(W_{++})_{ij}}\}^{\mc W}_0
=\delta_{kj}(w_{1+})_i
=-\{{(W_{++})_{ij}}_\lambda{(w_{1+})_k}\}^{\mc W}_0
\,,\\
& \{{(W_{++})_{ij}}_\lambda{(W_{++})_{hk}}\}^{\mc W}_0
=
\delta_{ki}(W_{++})_{hj}-\delta_{jh}(W_{++})_{ik}+\delta_{jh}\delta_{ik}\lambda
\,.
\end{split}
\end{equation}
and 
\begin{equation}\label{eq:minimal-1}
\begin{split}
& \{{w_{11;0}}_\lambda{w_{11;0}}\}^{\mc W}_1
=
2\lambda
\,,\\
& \{ {(w_{+1})_i}_{\lambda} (w_{1+})_j \}^{\mc W}_1
=-\{ {(w_{1+})_j}_{\lambda} (w_{+1})_i \}^{\mc W}_1
=-\delta_{ij}
\,\\
&\text{ all other } \lambda-\text{brackets of generators} = 0
\,.
\end{split}
\end{equation}
It is straightforward to check that equations \eqref{eq:minimal-0} and \eqref{eq:minimal-1}
agree with the analogous equations in \cite[Thm.3.4]{DSKV14a},
if we take into account formula \eqref{eq:minimal-gener-compare}
for the change of generators
and the opposite sign in the definition of Adler type operators.

\subsection{Example 5: vector and matrix constrained KP hierarchies}\label{sec:7.5}

Consider the partition $p=p_1+1+\dots+1$,
where the multiplicity of $1$ is $s=N-p_1$.
In this case \eqref{eq:L1-explicit2} becomes
the following scalar pseudodifferential operator
\begin{equation}\label{eq:L1-constraint}
L_1(\partial)
=
-(-\partial)^{p_1}
+\sum_{k=0}^{p_1-1}w_{11;k}(-\partial)^k
-w_{+1}(\partial+W_{++})^{-1}\circ w_{1+}
\,,
\end{equation}
where $w_{+1}$, $w_{1+}$ and $W_{++}$ are as in \eqref{eq:minimal-notation}.
In the last term of the RHS of \eqref{eq:L1-constraint}, $\partial$ stands for $\id_s\partial$.

It is possible, as we did in \eqref{eq:minimal-gener}
for the special case $p_1=2$,
to use the definition \eqref{eq:La} of $L_1(\partial)$ to find explicitly all the generators $w_{ji;k}$
of the $\mc W$-algebra as elements of $\mc V(\mf g_{\leq\frac12})$.
This, however, seems computationally too involved to be solved explicitly for every $p_1$.

Now we shall demonstrate how to use the Adler and bi-Adler identities \eqref{eq:adler-L1} and \eqref{eq:bi-adler-L1}
to find explicit formulas for the $0$-th and $1$-st Poisson structure of $\mc W_\epsilon(\mf g,f,S)$, $\epsilon\in\mb F$,
respectively.

First, we use the sesquilinearity and Leibniz rule axioms
to compute the LHS $\{L_1(z)_\lambda L_1(w)\}$
of both \eqref{eq:adler-L1} and \eqref{eq:bi-adler-L1},
where $\{\cdot\,_\lambda\,\cdot\}$
is either $\{\cdot\,_\lambda\,\cdot\}^{\mc W}_0$ or $\{\cdot\,_\lambda\,\cdot\}^{\mc W}_1$.
By \eqref{eq:L1-constraint} we have 
\begin{equation}\label{eq:constrained-eq1}
\begin{split}
& \{L_1(z)_\lambda L_1(w)\}
=
\sum_{h,k=0}^{p_1-1}(-z)^h(-w)^k\{{w_{11;h}}_\lambda{w_{11;k}}\} \\
& -\sum_{h=0}^{p_1-1}(-z)^h
\big\{{w_{11;h}}_\lambda {w_{+1}(w+\partial+W_{++})^{-1} w_{1+}} \big\} \\
& -\sum_{k=0}^{p_1-1}(-w)^k
\big\{ {w_{+1}(z+\partial+W_{++})^{-1} w_{1+}} _\lambda {w_{11;k}} \big\} \\
& +
\big\{ {w_{+1}(z+\partial+W_{++})^{-1} w_{1+}} _\lambda {w_{+1}(w+\partial+W_{++})^{-1} w_{1+}} \big\}
\,.
\end{split}
\end{equation}
In order to compute the $\lambda$-brackets in the RHS of \eqref{eq:constrained-eq1}
we use standard $\lambda$-bracket techniques.
In particular, applying \cite[Lem.2.3-2.5]{DSKVnew}, 
we compute the second term in the RHS of \eqref{eq:constrained-eq1}:
\begin{equation}\label{eq:constrained-eq2}
\begin{split}
& -\sum_{h=0}^{p_1-1}(-z)^h
\big\{{w_{11;h}}_\lambda {w_{+1}(w+\partial+W_{++})^{-1} w_{1+}} \big\} 
=
-\sum_{h=0}^{p_1-1}(-z)^h \\
& \Bigg(
\{{w_{11;h}}_\lambda w_{+1}\}(w+\partial+W_{++})^{-1} w_{1+} 
+w_{+1}(w+\lambda+\partial+W_{++})^{-1}  \{{w_{11;h}}_\lambda {w_{1+}} \} \\
& -w_{+1} (w+\lambda+\partial+W_{++})^{-1} \{{w_{11;h}}_\lambda {W_{++}} \} (w+\partial+W_{++})^{-1} w_{1+}
\Bigg)
\,,
\end{split}
\end{equation}
where all products are row by column multiplications,
and similarly the third term in the RHS of \eqref{eq:constrained-eq1} is
\begin{equation}\label{eq:constrained-eq3}
\begin{split}
& -\sum_{k=0}^{p_1-1}(-w)^k
\big\{ {w_{+1}(z+\partial+W_{++})^{-1} w_{1+}} _\lambda {w_{11;k}} \big\} 
=-\sum_{k=0}^{p_1-1}(-w)^k \\
& \Bigg(
\{ {w_{+1}}_{\lambda+\partial} {w_{11;k}} \}_\to (z\!+\!\partial\!+\!W_{++})^{-1} w_{1+}
+\{ { w_{1+}^T }_{\lambda+\partial} {w_{11;k}} \}_\to (z\!-\!\lambda\!-\!\partial\!+\!W_{++}^T)^{-1} w_{+1}^T \\
& -\sum_{i,j=1}^s \{ {(W_{++})_{ij}}_{\lambda+\partial} {w_{11;k}} \}_\to 
\Big((z\!-\!\lambda\!-\!\partial\!+\!W_{++}^T)^{-1} w_{+1}^T\Big)_i \Big((z\!+\!\partial\!+\!W_{++})^{-1} w_{1+}\Big)_j
\Bigg)
\,,
\end{split}
\end{equation}
where the superscript $T$ denotes taking the transpose of a vector or a matrix.
In writing these formulas we use parenthesis in order to indicate where $\partial$ acts. 
For example, if we write an expression such as $a(\partial bc)d$, we mean $ab'cd+abc'd$.
To derive \eqref{eq:constrained-eq3} we used the fact that,
if $S_{ij}(z)=z\delta_{ji}+(W_{++})_{ij}$, 
then 
\begin{equation}\label{eq:S*}
(S^*)^{-1}_{ij}(\lambda-z)=(z-\lambda-\partial+W_{++}^T)^{-1}_{ij}1
\,.
\end{equation}
Combining \eqref{eq:constrained-eq2} and \eqref{eq:constrained-eq3},
we get the fourth term in the RHS of \eqref{eq:constrained-eq1}:
\begin{equation}\label{eq:constrained-eq4}
\begin{split}
& \big\{ {w_{+1}(z+\partial+W_{++})^{-1} w_{1+}} _\lambda {w_{+1}(w+\partial+W_{++})^{-1} w_{1+}} \big\} \\
& = \Bigg(
\{ {(w_{+1})_i}_{\lambda+\partial} {(w_{+1})_h} \}_\to (z+\partial+W_{++})^{-1}_{ij} (w_{1+})_j
\Bigg)
(w+\partial+W_{++})^{-1}_{hk} (w_{1+})_k \\
& +
(w_{+1})_h(w+\lambda+\partial+W_{++})^{-1}_{hk}  \{ {(w_{+1})_i}_{\lambda+\partial} (w_{1+})_k \}_\to 
(z+\partial+W_{++})^{-1}_{ij} (w_{1+})_j \\
& -
(w_{+1})_h (w+\lambda+\partial+W_{++})^{-1}_{hp} 
\Bigg(
\{ {(w_{+1})_i}_{\lambda+\partial} (W_{++})_{pq} \}_\to  \\
&\qquad\qquad
(z+\partial+W_{++})^{-1}_{ij} (w_{1+})_j
\Bigg)
(w+\partial+W_{++})^{-1}_{qk} (w_{1+})_k \\
& +
\Bigg(
\{ {(w_{1+}^T)_j}_{\lambda+\partial} (w_{+1})_h\}_\to
(z-\lambda-\partial+W_{++}^T)^{-1}_{ji} (w_{+1}^T)_i 
\Bigg)
(w+\partial+W_{++})^{-1}_{hk} (w_{1+})_k \\
& +
(w_{+1})_h(w+\lambda+\partial+W_{++})^{-1}_{hk}  \{ {(w_{1+}^T)_j}_{\lambda+\partial} {(w_{1+})_k} \}_\to
(z-\lambda-\partial+W_{++}^T)^{-1}_{ji} (w_{+1}^T)_i \\
& -
(w_{+1})_h (w+\lambda+\partial+W_{++})^{-1}_{hp} 
\Bigg(
\{ {(w_{1+}^T)_j}_{\lambda+\partial} {(W_{++})_{pq}} \}_\to \\
&\qquad\qquad
(z-\lambda-\partial+W_{++}^T)^{-1}_{ji} (w_{+1}^T)_i 
\Bigg)
(w+\partial+W_{++})^{-1}_{qk} (w_{1+})_k 
\end{split}\end{equation}%
\begin{equation*}\begin{split}%
& -
\Bigg(
\{ {(W_{++})_{lr}}_{\lambda+\partial} (w_{+1})_h\}_\to
\Big((z-\lambda-\partial+W_{++}^T)^{-1}_{li} (w_{+1}^T)_i\Big) \\
&\qquad\qquad
\Big((z+\partial+W_{++})^{-1}_{rj} (w_{1+})_j\Big)
\Bigg)
(w+\partial+W_{++})^{-1}_{hk} (w_{1+})_k \\ 
& -
(w_{+1})_h(w+\lambda+\partial+W_{++})^{-1}_{hk}  
\{ {(W_{++})_{lr}}_{\lambda+\partial} {(w_{1+})_k} \}_\to \\
&\qquad\qquad
\Big((z-\lambda-\partial+W_{++}^T)^{-1}_{li} (w_{+1}^T)_i\Big) \Big((z+\partial+W_{++})^{-1}_{rj} (w_{1+})_j\Big) \\
& +
(w_{+1})_h (w+\lambda+\partial+W_{++})^{-1}_{hp} 
\Bigg(
\{ {(W_{++})_{lr}}_{\lambda+\partial} {(W_{++})_{pq}} \}_\to 
\Big((z-\lambda-\partial+W_{++}^T)^{-1}_{li} \\
&\qquad\qquad
(w_{+1}^T)_i\Big) 
\Big((z+\partial+W_{++})^{-1}_{rj} (w_{1+})_j\Big)
\Bigg)
(w+\partial+W_{++})^{-1}_{qk} (w_{1+})_k
\,,
\end{split}
\end{equation*}
where we use the Einstein convention of summing over repeated indices.

Next, we compute the RHS of \eqref{eq:adler-L1}.
Note that if $L_1(\partial)$ is given by \eqref{eq:L1-constraint},
then
\begin{equation}\label{eq:L1-constraint-b}
L_1(z)
=
\sum_{k=0}^{p_1}w_{11;k}(-z)^k
-w_{+1}(z+\partial+W_{++})^{-1} w_{1+}
\,,
\end{equation}
where $w_{11;p}=-1$,
and (cf. \eqref{eq:S*})
\begin{equation}\label{eq:L1-constraint-c}
L_1^*(\lambda-z)
=
\sum_{k=0}^{p_1}(-z+\lambda+\partial)^k w_{11;k}
-w_{1+}^T(z-\lambda-\partial+W_{++}^T)^{-1} w_{+1}^T
\,.
\end{equation}
Hence, the RHS of \eqref{eq:adler-L1} is
\begin{equation}\label{eq:constrained-eq5}
\begin{split}
& 
L_1(w+\lambda+\partial)\iota_z(z\!-\!w\!-\!\lambda\!-\!\partial)^{-1}L_1^*(\lambda-z)
- L_1(z)\iota_z(z\!-\!w\!-\!\lambda\!-\!\partial)^{-1}L_1(w) \\
& =
\sum_{h,k=0}^{p_1}
w_{11;h}
\Big(
(\!-\!w\!-\!\lambda\!-\!\partial)^h(\!-\!z\!+\!\lambda\!+\!\partial)^k 
- (\!-\!z)^h(\!-\!w)^k
\Big)
\iota_z(z\!-\!w\!-\!\lambda\!-\!\partial)^{-1}
w_{11;k} \\
& -
\sum_{h=0}^{p_1}
w_{11;h}
\iota_z(z\!-\!w\!-\!\lambda\!-\!\partial)^{-1}
\Big(
(-w-\lambda-\partial)^h
w_{1+}^T(z-\lambda-\partial+W_{++}^T)^{-1} w_{+1}^T \\
& \qquad\qquad\qquad\qquad\qquad\qquad
-(-z)^h
w_{+1}(w+\partial+W_{++})^{-1} w_{1+}
\Big) \\
& -
\sum_{k=0}^{p_1}
\Big(
w_{+1}(w+\lambda+\partial+W_{++})^{-1}\circ w_{1+}
(-z+\lambda+\partial)^k \\
& \qquad\qquad
-
\big( w_{+1}(z+\partial+W_{++})^{-1} w_{1+} \big)
(-w)^k
\Big)
\iota_z(z\!-\!w\!-\!\lambda\!-\!\partial)^{-1}
w_{11;k} \\
& +
w_{+1}(w+\lambda+\partial+W_{++})^{-1} w_{1+}
\iota_z(z\!-\!w\!-\!\lambda\!-\!\partial)^{-1}
w_{1+}^T(z-\lambda-\partial+W_{++}^T)^{-1} w_{+1}^T \\
& -
\big( w_{+1}(z+\partial+W_{++})^{-1} w_{1+} \big)
\iota_z(z\!-\!w\!-\!\lambda\!-\!\partial)^{-1}
w_{+1}(w+\partial+W_{++})^{-1} w_{1+}
\,.
\end{split}
\end{equation}
We expand the first term in the RHS of \eqref{eq:constrained-eq5}
in powers of $-z$ and $-w$.
As a result we get (cf. \eqref{eq:20150911-1})
\begin{equation}\label{eq:constrained-eq6}
\begin{split}
& 
\sum_{h,k=0}^{p_1}
w_{11;h}
\Big(
(\!-\!w\!-\!\lambda\!-\!\partial)^h(\!-\!z\!+\!\lambda\!+\!\partial)^k 
- (\!-\!z)^h(\!-\!w)^k
\Big)
\iota_z(z\!-\!w\!-\!\lambda\!-\!\partial)^{-1}
w_{11;k} \\
& =
\sum_{h,k=0}^{p_1}
(-z)^h(-w)^k
\sum_{n=0}^{p_1-h-1}
\sum_{a=\max\{0,n-k\}}^{p_1+n-k}
\Bigg(
\binom{n}{a}(-1)^{a}w_{11;n+h+1}(\lambda+\partial)^{a}w_{11;k+a-n} \\
& +\!\!\!
\sum_{b=0}^{p_1-n-h-1}
\!\!\!
\binom{k+a}{a}
\binom{h\!+\!n\!+\!b\!+\!1}{b}
(-1)^{a+1}
w_{11;k+a-n}(\lambda+\partial)^{a+b}w_{11;h+n+b+1}
\Bigg)
\,.
\end{split}
\end{equation}
The second and third terms in the RHS of \eqref{eq:constrained-eq5}
can be expanded in view of the following identity,
which can be easily proved for all $n\in\mb Z_+$,
\begin{equation}\label{eq:identity-rational}
\begin{split}
& (x^n(y+T)^{-1}-y^n(x+T)^{-1})(x-y)^{-1}
=
\sum_{i=0}^{n-2}\sum_{j=0}^{n-2-i} x^iy^j(-T)^{n-2-i-j} \\
& +\! \sum_{i=0}^{n-1}\!\! x^i(-T)^{n-1-i} (y\!+\!T)^{-1}
\!+\!
\sum_{i=0}^{n-1}\!\! y^i(-\!T)^{n-1-i} (x\!+\!T)^{-1}
\!+\!
(-\!T)^n(x\!+\!T)^{-1}(y\!+\!T)^{-1}
\,,
\end{split}
\end{equation}
where $x,y,T$ are commuting variables.
For example, the second term in the RHS of \eqref{eq:constrained-eq5}
can be rewritten as
\begin{equation}\label{eq:constrained-eq7a}
\begin{split}
& -
\sum_{h=0}^{p_1}
w_{11;h}
\iota_z(z\!-\!w\!-\!\lambda\!-\!\partial)^{-1}
\Big(
(-w-\lambda-\partial)^h
w_{1+}^T(z-\lambda-\partial+W_{++}^T)^{-1} w_{+1}^T \\
& \qquad\qquad\qquad\qquad\qquad\qquad
-(-z)^h
w_{+1}(w+\partial+W_{++})^{-1} w_{1+}
\Big) \\
& =
-\sum_{h=0}^{p_1}(-1)^h
w_{11;h}\Big|_{\mu=\partial}
w_{+1}
\big((w+\lambda+\mu)^h(z-\lambda-\mu+\partial+W_{++})^{-1} \\
&\qquad
-z^h(w+\partial+W_{++})^{-1}\big)(z-w-\lambda-\mu)^{-1}w_{1+}
\,,
\end{split}
\end{equation}
and, therefore, applying identity \eqref{eq:identity-rational} 
with $x=z$, $y=w+\lambda+\mu$
and $T=-\lambda-\mu+\partial+W_{++}$,
it is equal to
\begin{equation}\label{eq:constrained-eq7b}
\begin{split}
&
\sum_{h=0}^{p_1}
\sum_{i=0}^{h-2}\sum_{j=0}^{h-2-i} 
w_{11;h}
(-z)^i(-w-\lambda-\partial)^j
w_{1+}^T
(-\lambda-\partial+W_{++}^T)^{h-2-i-j} w_{+1}^T \\
& -
\sum_{h=0}^{p_1}
\sum_{i=0}^{h-1}
w_{11;h}
(-z)^i
\Big(
(w+\partial+W_{++})^{-1}
w_{1+}
\Big)^T
(-\lambda-\partial+W_{++}^T)^{h-1-i} w_{+1}^T \\
& -
\sum_{h=0}^{p_1}
\sum_{i=0}^{h-1} 
w_{11;h}
(-\!w\!-\!\lambda\!-\!\partial)^i 
w_{1+}^T
(-\!\lambda\!-\!\partial+W_{++}^T)^{h-1-i} (z\!-\!\lambda\!-\!\partial\!+\!W_{++}^T)^{-1} w_{+1}^T \\
& +
\sum_{h=0}^{p_1}
w_{11;h}
\Big((w\!+\!\partial\!+\!W_{++})^{-1}w_{1+}\Big)^T
(-\!\lambda\!-\!\partial\!+\!W_{++}^T)^h (z\!-\!\lambda\!-\!\partial\!+\!W_{++}^T)^{-1} w_{+1}^T
\,.
\end{split}
\end{equation}
Similarly, the third term in the RHS of \eqref{eq:constrained-eq5}
can be rewritten as
\begin{equation}\label{eq:constrained-eq8a}
\begin{split}
& -
\sum_{k=0}^{p_1}
\Big(
w_{+1}(w+\lambda+\partial+W_{++})^{-1} w_{1+}
(-z+\lambda+\partial)^k
\iota_z(z\!-\!w\!-\!\lambda\!-\!\partial)^{-1}
w_{11;k} \\
& \qquad\qquad
-
\big( w_{+1}(z+\partial+W_{++})^{-1} w_{1+} \big)
(-w)^k
\iota_z(z\!-\!w\!-\!\lambda\!-\!\partial)^{-1}
w_{11;k} 
\Big) \\
& = -
\sum_{k=0}^{p_1} (-1)^k \big(\big|_{\mu=\partial}w_{11;k}\big) w_{+1}
\big(
(z-\lambda-\mu)^k
(w+\lambda+\mu+\partial+W_{++})^{-1} \\
&\qquad
-w^k
(z+\partial+W_{++})^{-1}
\big)
(z-w-\lambda-\mu)^{-1}
w_{1+}
\,,
\end{split}
\end{equation}
and, therefore, applying identity \eqref{eq:identity-rational} 
with $x=z-\lambda-\mu$, $y=w$
and $T=\lambda+\mu+\partial+W_{++}$,
it is equal to
\begin{equation}\label{eq:constrained-eq8b}
\begin{split}
& 
-\sum_{k=0}^{p_1}\sum_{i=0}^{k-2}\sum_{j=0}^{k-2-i}
w_{+1}
(\lambda+\partial+W_{++})^{k-2-i-j}
w_{1+}
(-z+\lambda+\partial)^i
(-w)^j
w_{11;k} \\
& +
\sum_{k=0}^{p_1}\sum_{i=0}^{k-1}
w_{+1}
(\lambda+\partial+W_{++})^{k-1-i}
(w+\lambda+\partial+W_{++})^{-1} w_{1+}
(-z+\lambda+\partial)^i
w_{11;k} \\
& +
\sum_{k=0}^{p_1}\sum_{i=0}^{k-1}
w_{+1}
(\lambda+\partial+W_{++})^{k-1-i}
\big( (z+\partial+W_{++})^{-1} w_{1+} \big)
(-w)^i
w_{11;k} \\
& -\sum_{k=0}^{p_1}
w_{+1}
(\lambda+\partial+W_{++})^k
(w+\lambda+\partial+W_{++})^{-1}
\big( (z+\partial+W_{++})^{-1} w_{1+} \big)
w_{11;k}
\,.
\end{split}
\end{equation}
Finally, the last two terms in the RHS of \eqref{eq:constrained-eq5}
can be rewritten as
\begin{equation}\label{eq:constrained-eq9a}
\begin{split}
&
w_{+1}(w+\lambda+\partial+W_{++})^{-1} w_{1+}
\iota_z(z\!-\!w\!-\!\lambda\!-\!\partial)^{-1}
w_{1+}^T(z-\lambda-\partial+W_{++}^T)^{-1} w_{+1}^T \\
& -
\big( w_{+1}(z+\partial+W_{++})^{-1} w_{1+} \big)
\iota_z(z\!-\!w\!-\!\lambda\!-\!\partial)^{-1}
w_{+1}(w+\partial+W_{++})^{-1} w_{1+} \\
& =
\!
\Big(w_{+1}(w\!+\!\mu\!+\!\partial\!+\!W_{++})^{-1} w_{1+}\Big)
(z\!-\!w\!-\!\mu)^{-1}
\Big|_{\mu=\lambda\!+\!\partial}
\Big(w_{+1}(z\!-\!\mu\!+\!\partial\!+\!W_{++})^{-1} w_{1+}\Big)
\\
& -
\Big( w_{+1}(z\!+\!\partial\!+\!W_{++})^{-1} w_{1+} \Big)
(z\!-\!w\!-\!\mu)^{-1}
\Big|_{\mu=\lambda\!+\!\partial}
\Big( w_{+1}(w\!+\!\partial\!+\!W_{++})^{-1} w_{1+} \Big) 
\,.
\end{split}
\end{equation}
We next use the following identity
of rational functions, which can be easily checked:
\begin{equation}\label{eq:identity-rational2}
\begin{split}
& ((x+S)^{-1}(y+T)^{-1}-(x+T)^{-1}(y+S)^{-1})(x-y)^{-1} \\
& =
(x+S)^{-1}(x+T)^{-1}(y+S)^{-1}(y+T)^{-1}(S-T)
\,.
\end{split}
\end{equation}
where $x,y,S,T$ are commuting variables.
Using \eqref{eq:identity-rational2} 
with $x=z$, $y=w+\mu$, $S=(\partial+W_{++})\otimes1$ and $T=-\mu+1\otimes(\partial+W_{++})$,
we can rewrite the RHS of \eqref{eq:constrained-eq9a} as
\begin{equation}\label{eq:constrained-eq9b}
\begin{split}
& 
w_{+1}(w+\lambda+\partial+W_{++})^{-1}
\big( (z+\partial+W_{++})^{-1} w_{1+} \big) \\
&\qquad\times 
\big( (\partial+W_{++})(w+\partial+W_{++})^{-1} w_{1+} \big)^T
(z-\lambda-\partial+W_{++}^T)^{-1} w_{+1}^T \\
& -w_{+1}(\lambda+\partial+W_{++})(w+\lambda+\partial+W_{++})^{-1}
\big( (z+\partial+W_{++})^{-1} w_{1+} \big) \\
&\qquad\times
\big( (w+\partial+W_{++})^{-1} w_{1+} \big)^T
(z-\lambda-\partial+W_{++}^T)^{-1} w_{+1}^T 
\,.
\end{split}
\end{equation}
We now combine, on one hand, 
equations \eqref{eq:constrained-eq1},
\eqref{eq:constrained-eq2},
\eqref{eq:constrained-eq3}
and \eqref{eq:constrained-eq4},
and, on the other hand, 
equations \eqref{eq:constrained-eq5},
\eqref{eq:constrained-eq6},
\eqref{eq:constrained-eq7b},
\eqref{eq:constrained-eq8b},
and \eqref{eq:constrained-eq9b}.
Comparing the results, we get the desired formulas for all the $\lambda$-brackets.
For example, 
comparing the coefficient of $(-z)^h(-w)^k$, for $h,k\geq0$, in \eqref{eq:constrained-eq1}
and in \eqref{eq:constrained-eq6}, \eqref{eq:constrained-eq7b}, and \eqref{eq:constrained-eq8b},
we get
\begin{equation}\label{eq:constrained-lambda-1}
\begin{split}
& \{{w_{11;h}}_\lambda{w_{11;k}}\}^{\mc W}_0 \\
& =
\sum_{n=0}^{p_1-h-1}
\sum_{a=\max\{0,n-k\}}^{p_1+n-k}
\Bigg(
\binom{n}{a}(-1)^{a}w_{11;n+h+1}(\lambda+\partial)^{a}w_{11;k+a-n} \\
&\quad +\!\!\!
\sum_{b=0}^{p_1\!-\!n\!-\!h\!-\!1}
\binom{k+a}{a}
\binom{h\!+\!n\!+\!b\!+\!1}{b}
(-1)^{a+1}
w_{11;k+a-n}(\lambda+\partial)^{a+b}w_{11;h+n+b+1}
\Bigg) \\
& +
\!\!\!
\sum_{a=0}^{p_1\!-\!h\!-\!k\!-\!2}\sum_{b=0}^{p_1\!-\!h\!-\!k\!-\!2\!-\!a}
\!\!\!
\binom{k+a}{a}
w_{11;a+b+h+k+2}
(-\lambda-\partial)^{a}
\Big((-\!\lambda\!-\!\partial\!+\!W_{++}^T)^{b} w_{+1}^T\Big)^T
w_{1+} \\
& -
\sum_{a=0}^{p_1\!-\!h\!-\!k\!-\!2}\sum_{b=0}^{p_1\!-\!h\!-\!k\!-\!2\!-\!a}
\binom{h+a}{a}
w_{+1}
(\lambda+\partial+W_{++})^{b}
w_{1+}
(\lambda+\partial)^{a}
w_{11;a+b+h+k+2} 
\,.
\end{split}
\end{equation}
Note that for $p_1\leq2$ the last two terms in the RHS of \eqref{eq:constrained-lambda-1}
vanish, in accordance with the first four lines in \eqref{eq:minimal-0}.
Next, comparing the terms with non-negative powers of $z$
in \eqref{eq:constrained-eq2}
and in \eqref{eq:constrained-eq7b}-\eqref{eq:constrained-eq8b},
we get
\begin{equation*}
\begin{split}
& - \sum_{h=0}^{p_1-1}(-z)^h
\{{w_{11;h}}_\lambda w_{+1}\}^{\mc W}_0 (w+\partial+W_{++})^{-1} w_{1+} 
\\
& - \sum_{h=0}^{p_1-1}(-z)^h
w_{+1}(w+\lambda+\partial+W_{++})^{-1}  \{{w_{11;h}}_\lambda {w_{1+}} \}^{\mc W}_0 \\
& + \sum_{h=0}^{p_1-1}(-z)^h
w_{+1} (w+\lambda+\partial+W_{++})^{-1} \{{w_{11;h}}_\lambda {W_{++}} \}^{\mc W}_0 
(w+\partial+W_{++})^{-1} w_{1+}
\\
& =
-
\sum_{h=0}^{p_1}
\sum_{i=0}^{h-1} 
(-z)^i
w_{11;h}
\big( (w+\partial+W_{++})^{-1} w_{1+} \big)^T
(-\lambda-\partial+W_{++}^T)^{h-1-i} w_{+1}^T \\
& +
\sum_{k=0}^{p_1}
\sum_{i=0}^{k-1}
w_{+1}
(\lambda+\partial+W_{++})^{k-1-i}
(w+\lambda+\partial+W_{++})^{-1} w_{1+}
(-z+\lambda+\partial)^i
w_{11;k} 
\,,
\end{split}
\end{equation*}
which implies
\begin{equation}\label{eq:constrained-lambda-2}
\begin{split}
& 
\{{w_{11;h}}_\lambda w_{+1}\}^{\mc W}_0
=
\sum_{a=0}^{p_1-h-1}
w_{11;a+h+1}
\Big((-\lambda-\partial+W_{++}^T)^{a} w_{+1}^T\Big)^T
\,,\\
& 
\{{w_{11;h}}_\lambda {w_{1+}} \} ^{\mc W}_0
= -
\!\!\!
\sum_{a=0}^{p_1-h-1}
\!
\sum_{b=0}^{p_1-h-1-a}
\!\!\!
\binom{a+h}{a}
(\lambda\!+\!\partial\!+\!W_{++})^{b}
w_{1+}
(\lambda\!+\!\partial)^a
w_{11;a+b+h+1} 
\,,\\
&
\{{w_{11;h}}_\lambda {W_{++}} \}^{\mc W}_0
=
0
\,.
\end{split}
\end{equation}
Similarly, looking at the terms with non-negative powers of $w$
in \eqref{eq:constrained-eq3}
and in \eqref{eq:constrained-eq7b}-\eqref{eq:constrained-eq8b},
we get
\begin{equation*}
\begin{split}
& - \sum_{k=0}^{p_1-1}(-w)^k
{\{ {w_{+1}}_{\lambda+\partial} {w_{11;k}} \}^{\mc W}_0}_\to (z\!+\!\partial\!+\!W_{++})^{-1} w_{1+} 
\\
& - \sum_{k=0}^{p_1-1}(-w)^k
{\{ { w_{1+}^T }_{\lambda+\partial} {w_{11;k}} \}^{\mc W}_0}_\to (z\!-\!\lambda\!-\!\partial\!+\!W_{++}^T)^{-1} w_{+1}^T 
\\
& + \sum_{k=0}^{p_1-1}(-w)^k
\sum_{i,j=1}^s {\{ {(W_{++})_{ij}}_{\lambda+\partial} {w_{11;k}} \}^{\mc W}_0}_\to 
\Big((z\!-\!\lambda\!-\!\partial\!+\!W_{++}^T)^{-1} w_{+1}^T\Big)_i 
\\
&\qquad\qquad\times
\Big((z\!+\!\partial\!+\!W_{++})^{-1} w_{1+}\Big)_j
\\
& =
-
\sum_{h=0}^{p_1}
w_{11;h}
\sum_{i=0}^{h-1} (-\!w\!-\!\lambda\!-\!\partial)^i 
w_{1+}^T
(z\!-\!\lambda\!-\!\partial\!+\!W_{++}^T)^{-1} (-\!\lambda\!-\!\partial\!+\!W_{++}^T)^{h-1-i} w_{+1}^T
 \\
& +
\sum_{k=0}^{p_1}
\sum_{i=0}^{k-1}
w_{+1}
(\lambda+\partial+W_{++})^{k-1-i}
\big( (z+\partial+W_{++})^{-1} w_{1+} \big)
(-w)^i
w_{11;k} 
\,,
\end{split}
\end{equation*}
which implies
\begin{equation}\label{eq:constrained-lambda-3}
\begin{split}
& 
\{ {w_{+1}}_{\lambda} {w_{11;k}} \}^{\mc W}_0
=
-
\sum_{a=0}^{p_1-k-1}
w_{+1}
(\lambda+\partial+W_{++})^{a}
w_{11;a+k+1} 
\,,\\
& 
\{ { w_{1+}^T }_{\lambda} {w_{11;k}} \}^{\mc W}_0
=
\!\!\!
\sum_{a=0}^{p_1\!-\!k\!-\!1}
\sum_{b=0}^{p_1\!-\!k\!-\!1\!-\!a}
\!\!\!
\binom{a\!+\!k}{a}
w_{11;a+b+k+1}
(-\!\lambda\!-\!\partial)^a 
w_{1+}^T
(-\!\lambda\!-\!\partial\!+\!W_{++}^T)^{b}1
\,,\\
&
\{ {(W_{++})_{ij}}_{\lambda} {w_{11;k}} \}^{\mc W}_0 
=
0
\,.
\end{split}
\end{equation}
Finally, taking all the remaining terms, with negative powers in both $z$ and $w$,
in \eqref{eq:constrained-eq4},
and in \eqref{eq:constrained-eq7b}, \eqref{eq:constrained-eq8b} and \eqref{eq:constrained-eq9b},
we get the remaining $\lambda$-brackets:
\begin{equation}\label{LHS}
\begin{split}
& 
\{ {(w_{+1})_i}_{\lambda} (w_{1+})_j \}^{\mc W}_0
=
- \sum_{k=0}^{p_1}
(\lambda+\partial+W_{++})^k_{ji}
w_{11;k}
\,,\\
& 
\{ {(w_{1+})_j}_{\lambda} (w_{+1})_i\}^{\mc W}_0
=
\sum_{h=0}^{p_1}
w_{11;h}
(-\lambda-\partial+W_{++}^T)^h_{ij}1
\,,\\
&
\{ {(W_{++})_{ij}}_{\lambda} {(W_{++})_{hk}} \}^{\mc W}_0
=
\delta_{ik}(W_{++})_{hj}-\delta_{hj}(W_{++})_{ik}
+\delta_{ik}\delta_{hj} \lambda
\,,\\
& \{{(w_{+1})_k}_\lambda{(W_{++})_{ij}}\}^{\mc W}_0
=-\delta_{ik}(w_{+1})_j
=-\{{(W_{++})_{ij}}_\lambda{(w_{+1})_k}\}^{\mc W}_0
\,,
\\
&
\{{(w_{1+})_k}_\lambda{(W_{++})_{ij}}\}^{\mc W}_0
=\delta_{kj}(w_{1+})_i
=-\{{(W_{++})_{ij}}_\lambda{(w_{1+})_k}\}^{\mc W}_0
\,.
\end{split}
\end{equation}
It is immediate to check that, in the special case $p_1=2$,
the $\lambda$-brackets \eqref{eq:constrained-lambda-1}, \eqref{eq:constrained-lambda-2}, 
\eqref{eq:constrained-lambda-3}, and \eqref{LHS},
reduce to \eqref{eq:minimal-0}. 

Next, we derive in a similar way the $1$-st Poisson structure $\{\cdot\,_\lambda\,\cdot\}^{\mc W}_1$,
by computing the RHS of \eqref{eq:bi-adler-L1}:
\begin{equation}\label{eq:constrained-1st-1}
\begin{split}
& \iota_z(z\!-\!w\!-\!\lambda)^{-1}
\big(
L_1(w+\lambda)
- L_1(z)
\big)
+
\iota_z(z\!-\!w\!-\!\lambda\!-\!\partial)^{-1}
\big(
(L_1)^*(\lambda-z)-L_1(w)
\big) \\
& =
\sum_{k=0}^{p_1}w_{11;k}((-w-\lambda)^k-(-z)^k)(z-w-\lambda)^{-1} \\
&\qquad
+ \sum_{k=0}^{p_1}((-z+\lambda+\partial)^k-(-w)^k)(z-w-\lambda-\partial)^{-1} w_{11;k} \\
&\qquad
-w_{+1}(z-w-\lambda)^{-1}
\big((w+\lambda+\partial+W_{++})^{-1}-(z+\partial+W_{++})^{-1}\big)w_{1+} \\
&\qquad
-\Big|_{\mu=\lambda+\partial}
w_{+1}(z-w-\mu)^{-1}
\big((z-\mu+\partial+W_{++})^{-1}-(w+\partial+W_{++})^{-1}\big)w_{1+}
\\
& =
\sum_{h=0}^{p_1\!-\!1}\sum_{k=0}^{p_1\!-\!1\!-\!h}\sum_{\ell=0}^{p_1\!-\!h\!-\!k\!-\!1}
\!\!\! (-z)^h(-w)^k
\Big(
\binom{\ell+k}{k}(-\lambda)^\ell
-\binom{\ell+h}{h}(\lambda+\partial)^\ell
\Big) w_{11;\ell+h+k+1} \\
&\qquad
-w_{+1}(w+\lambda+\partial+W_{++})^{-1}(z+\partial+W_{++})^{-1}w_{1+} \\
&\qquad
+\big((z-\lambda-\partial+W_{++}^T)^{-1}w_{+1}^T\big)^T(w+\partial+W_{++})^{-1}w_{1+}
\,.
\end{split}
\end{equation}
For the first equality we used \eqref{eq:L1-constraint-b}-\eqref{eq:L1-constraint-c},
while the second equality is obtained by a straightforward computation.
Comparing 
equation \eqref{eq:constrained-1st-1}
with equations \eqref{eq:constrained-eq1},
\eqref{eq:constrained-eq2},
\eqref{eq:constrained-eq3}
and \eqref{eq:constrained-eq4},
we get, in view of \eqref{eq:bi-adler-L1},
the $1$-st Poisson structure:
\begin{equation}\label{eq:constraint-1st}
\begin{split}
&  \{ {w_{11;h}}_{\lambda} w_{11;k} \}^{\mc W}_1
=
\sum_{\ell=0}^{p_1\!-\!h\!-\!k\!-\!1}
\!\!\!
\Big(
\binom{\ell+k}{k}(-\lambda)^\ell
-\binom{\ell+h}{h}(\lambda+\partial)^\ell
\Big) w_{11;\ell+h+k+1}
\,,\\
& \{ {(w_{+1})_i}_{\lambda} (w_{1+})_j \}^{\mc W}_1
=-\{ {(w_{1+})_j}_{\lambda} (w_{+1})_i \}^{\mc W}_1
=-\delta_{ij}
\,,\\
&\text{ all other } \lambda-\text{brackets of generators} = 0
\,.
\end{split}
\end{equation}
Again, in the special case $p_1=2$,
the $\lambda$-brackets \eqref{eq:constraint-1st}
reduce to \eqref{eq:minimal-1}.

Note that, in agreement with Corollary \ref{cor:casimirs},
all entries of the matrix $W_{++}$ 
are central with respect to the $\lambda$-bracket
$\{\cdot\,_\lambda\,\cdot\}_1^{\mc W}$.
Therefore we can consider the differential ideal
$\mc J$ of $\mc W$ generated by these elements,
which will be a PVA ideal with respect to the $\lambda$-bracket $\{\cdot\,_\lambda\,\cdot\}_1^{\mc W}$.
We may then apply a Dirac reduction with respect 
to the $\lambda$-bracket $\{\cdot\,_\lambda\,\cdot\}_0^{\mc W}$,
to get a bi-PVA structure on the quotient $\mc W/\mc J$,
with compatible $\lambda$-brackets
induced by the $1$-st $\lambda$-bracket $\{\cdot\,_\lambda\,\cdot\}_1^{\mc W}$, which remains local,
and by the Dirac modification $\{\cdot\,_\lambda\,\cdot\}_0^{\mc W,D}$ of the $0$-th $\lambda$-bracket,
which becomes non-local.

The image of the operator $L_1(\partial)$ in the quotient space $\mc W/\mc J$ is
$$
\overline{L}_1(\partial)
=
-(-\partial)^{p_1}
+\sum_{k=0}^{p_1-1}w_{11;k}(-\partial)^k
-w_{+1}\partial^{-1}\circ w_{1+}
\,\in(\mc W/\mc J)((\partial^{-1}))\,,
$$
and therefore the corresponding integrable hierarchy
of Lax equations
$$
\frac{d}{dt_n}\overline{L}_1(\partial)=[(\overline{L}_1(\partial)^{\frac n{p_1}})_+,\overline{L}_1(\partial)]
\,,
$$
is the $(N-p_1)$-vector $p_1$-constrained KP hierarchy, \cite{YO76,Ma81,KSS91,Che92,KS92,SS93,ZC94}.
This isomorphism was originally stated in \cite{DSKV15-cor}.

\begin{remark}\label{rem:p11}
The case $p_1=1$ corresponds to $f=0$ and $S=\diag(1,0,\dots,0)$,
namely the family of affine PVAs $\mc V_\epsilon(\mf{gl}_N,S)$, $\epsilon\in\mb F$ (cf. Example \ref{ex:A}).
The corresponding bi-Adler type operator for this bi-PVA structure
is $|A(\partial)|_{11}$,
which, after the Dirac reduction by the elements $q_{ij}$ for $i=j=1$ and $i,j\geq2$,
becomes $\partial+\sum_{j=2}^N q_{j1}\partial^{-1}\circ q_{1j}$.
The corresponding hierarchy of bi-Hamiltonian equations for $N=2$
is the NLS hierarchy.
\end{remark}

We can consider, 
more generally, the partition $p=p_1+\dots+p_1+1+\dots+1$, 
with $r_1$ parts of size $p_1$ and $q_1$ parts of size $1$,
so that $r=r_1+q_1$, and $N=r_1p_1+q_1$.
In this case \eqref{eq:L1-explicit2} becomes
the following $r_1\times r_1$ matrix pseudodifferential operator
\begin{equation}\label{eq:L1-constraint-matrix}
L_1(\partial)
=
-\id_{r_1}(-\partial)^{p_1}
+\sum_{k=0}^{p_1-1}W_{1;k}(-\partial)^k
-W_2(\id_{q_1}\partial+W_4)^{-1}\circ W_3
\,,
\end{equation}
where 
\begin{equation*}
\begin{split}
& W_{1;k}
=\big(w_{ji;k}\big)_{1\leq i,j\leq r_1}
\,,\,\,
W_2=\big(w_{ji;0}\big)_{1\leq i\leq r_1<j\leq r}
\,,\\
& W_3=\big(w_{ji;0}\big)_{r_1<i\leq r,\,1\leq j\leq r_1}
\,,\,\,
W_4=\big(w_{ji;0}\big)_{r_1<i,j\leq r}
\,.
\end{split}
\end{equation*}
It is possible to compute (but we will not do it) the corresponding 
compatible PVA structures $\{\cdot\,_\lambda\,\cdot\}^{\mc W}_0$
and $\{\cdot\,_\lambda\,\cdot\}^{\mc W}_1$
for the family of PVAs $\mc W_\epsilon(\mf{gl}_N,f,S_1)$,
generalizing the formulas obtained for $r_1=1$.

As before, 
all entries $w_{ij;0}$, $r_1< i,j\leq r$ of the matrix $W_4$ are central with respect to the $1$-st 
PVA $\lambda$-bracket $\{\cdot\,_\lambda\,\cdot\}_1^{\mc W}$.
Therefore we can consider the differential ideal
$\mc J=\langle w_{ij;0}\rangle_{r_1< i,j\leq r}$ of $\mc W$ generated by these elements,
which will be a PVA ideal with respect to the $\lambda$-bracket $\{\cdot\,_\lambda\,\cdot\}_1^{\mc W}$,
and we may then apply a Dirac reduction with respect to the $\lambda$-bracket $\{\cdot\,_\lambda\,\cdot\}_0^{\mc W}$.
As a result, we get a bi-PVA structure on the quotient $\mc W/\mc J$,
with compatible $\lambda$-brackets
induced by the $1$-st $\lambda$-bracket $\{\cdot\,_\lambda\,\cdot\}_1^{\mc W}$, which remains local,
and by the Dirac modification $\{\cdot\,_\lambda\,\cdot\}_0^{\mc W,D}$ of the $0$-th $\lambda$-bracket,
which becomes non-local.
The image of the operator $L_1(\partial)$ in the quotient space $\mc W/\mc J$ is
$$
\overline{L}_1(\partial)
=
-(-\partial)^{p_1}
+\sum_{k=0}^{p_1-1}W_{1;k}(-\partial)^k
-W_2\partial^{-1}\circ W_3
\,\in(\mc W/\mc J)((\partial^{-1}))\,.
$$
The corresponding integrable hierarchy of Lax equations
$$
\frac{d}{dt_n}\overline{L}_1(\partial)=[(\overline{L}_1(\partial)^{\frac n{p_1}})_+,\overline{L}_1(\partial)]
\,,
$$
is a matrix analogue of the $q_1$-vector $p_1$-constrained KP hierarchy.

\appendix

\section{Simpler proof of Theorem \ref{prop:L2}}\label{sec:appendix}

\begin{lemma}\label{lem:app1}
Let $\mc V$ be a differential algebra with a $\lambda$-bracket $\{\cdot\,_\lambda\,\cdot\}$.
Let  $A(\partial)\in\Mat_{N\times N}\mc V((\partial^{-1}))$
be a matrix pseudodifferential operator of Adler type
with respect to the $\lambda$-bracket $\{\cdot\,_\lambda\,\cdot\}$,
and assume that $A(\partial)$ is invertible in $\Mat_{N\times N}\mc V((\partial^{-1}))$.
Then
\begin{equation}\label{eq:app1}
\begin{split}
\{A_{ij}(z)_\lambda (A^{-1})_{hk}(w)\}
& = - \delta_{hj} \sum_{t=1}^N 
\iota_z(z\!-\!w\!-\!\lambda\!-\!\partial)^{-1}(A_{it})^*(\lambda-z)(A^{-1})_{tk}(w)
\\
& 
+ \delta_{ik} \sum_{t=1}^N
(A^{-1})_{ht}(w+\lambda+\partial)A_{tj}(z)\iota_z(z\!-\!w\!-\!\lambda)^{-1}
\,.
\end{split}
\end{equation}
\end{lemma}
\begin{proof}
By the identity $\sum_\ell A_{r\ell}(w+\partial)(A^{-1})_{\ell k}(w)=\delta_{r,k}$, we have
\begin{equation}\label{eq:app5}
\begin{array}{l}
\displaystyle{
0=\{A_{ij}(z)_\lambda \delta_{r,k}\}
=\sum_{\ell=1}^N
\{A_{ij}(z)_\lambda A_{r\ell}(w+x)\}\big|_{x=\partial}(A^{-1})_{\ell k}(w)
} \\
\displaystyle{
+\sum_{\ell=1}^N
A_{r\ell}(w+\lambda+\partial) \{A_{ij}(z)_\lambda (A^{-1})_{\ell k}(w)\}
\,.}
\end{array}
\end{equation}
Applying $(A^{-1})_{hr}(w+\lambda+\partial)$ to both sides of \eqref{eq:app5}
and summing over $r$, we get
\begin{equation}\label{eq:app6}
\begin{array}{l}
\displaystyle{
\{A_{ij}(z)_\lambda (A^{-1})_{h k}(w)\}
} \\
\displaystyle{
=
- \sum_{\ell,r=1}^N
(A^{-1})_{hr}(w+\lambda+\partial)
\{A_{ij}(z)_\lambda A_{r\ell}(w+x)\}\big|_{x=\partial}(A^{-1})_{\ell k}(w)
\,.}
\end{array}
\end{equation}
We finally use the Adler identity \eqref{eq:adler} to rewrite the RHS of \eqref{eq:app6} as
\begin{equation}\label{eq:app6}
\begin{array}{l}
\displaystyle{
- \sum_{\ell,r=1}^N
(A^{-1})_{hr}(w\!+\!\lambda\!+\!\partial)
\Big(
A_{rj}(w\!+\!\lambda\!+\!\partial\!+\!x)\iota_z(z\!-\!w\!-\!\lambda\!-\!\partial\!-\!x)^{-1}
(A_{i\ell})^*(\lambda-z)
} \\
\displaystyle{
- A_{rj}(z)\iota_z(z\!-\!w\!-\!\lambda\!-\!\partial\!-\!x)^{-1}A_{i\ell}(w+x)
\Big)
\Big|_{x=\partial}(A^{-1})_{\ell k}(w)
} \\
\displaystyle{
=
- \sum_{\ell,r=1}^N
\delta_{h,j}
\iota_z(z\!-\!w\!-\!\lambda\!-\!\partial)^{-1}
(A_{i\ell})^*(\lambda-z)
(A^{-1})_{\ell k}(w)
} \\
\displaystyle{
+ \sum_{\ell,r=1}^N
(A^{-1})_{hr}(w\!+\!\lambda\!+\!\partial)
A_{rj}(z)\iota_z(z\!-\!w\!-\!\lambda\!-\!\partial)^{-1}
\delta_{i,k}
\,,}
\end{array}
\end{equation}
proving equation \eqref{eq:app1}.
\end{proof}
\begin{lemma}[{\cite[{Lem.3.1(b)}]{DSKV13}}]\label{lem:app2}
Consider the pencil of affine Poisson vertex algebras $\mc V=\mc V_\epsilon(\mf g,S)$
from Example \ref{ex:A}, with $S\in\mf g_d$.
For $a\in\mf g_{\geq\frac12}$ and $g\in\mc V(\mf g)$, we have
$\rho\{a_\lambda \rho(g)\}_\epsilon=\rho\{a_\lambda g\}_\epsilon$.
\end{lemma}

\begin{proof}[Proof of Theorem \ref{prop:L2}]
By the definition \eqref{eq:La} of the matrix $L_1(\partial)$
and by Lemma \ref{lem:app2}, we have
\begin{equation}\label{eq:app3}
\rho\{a_\lambda L_1^{-1}(w)_{ij}\}_\epsilon
=
\rho\{a_\lambda \rho(J_1A^{-1}(w)I_1)_{ij}\}_\epsilon \\
=
\rho\{a_\lambda A^{-1}(w)_{(ip_1),(j1)}\}_\epsilon
\,.
\end{equation}
Let $a=q_{(\tilde j,\tilde k),(\tilde i,\tilde h)}\in\mf g_{\geq\frac12}$.
Note that, by the definition \eqref{eq:A} of the matrix $A(\partial)$,
we have
$a=A(z)_{(\tilde i,\tilde h),(\tilde j,\tilde k)}$.
Hence, we can apply Lemma \ref{lem:app1} to get, from \eqref{eq:app3},
\begin{equation}\label{eq:app4}
\begin{array}{l}
\displaystyle{
\rho\{a_\lambda L_1^{-1}(w)_{ij}\}_\epsilon
=
\rho
\{{A(z)_{(\tilde i,\tilde h),(\tilde j,\tilde k)}}\,_\lambda\, A^{-1}(w)_{(ip_1),(j1)}\}_\epsilon 
} \\
\displaystyle{
= 
- \delta_{(ip_1)(\tilde j,\tilde k)} \sum_{\tau\in\mc J} 
\iota_z(z\!-\!w\!-\!\lambda\!-\!\partial)^{-1}(A_{(\tilde i,\tilde h)\tau})^*(\lambda-z)(A^{-1})_{\tau(j1)}(w)
} \\
\displaystyle{
+ \delta_{(\tilde i,\tilde h)(j1)} \sum_{\tau\in\mc J}
(A^{-1})_{(ip_1)\tau}(w+\lambda+\partial)A_{\tau(\tilde j,\tilde k)}(z)\iota_z(z\!-\!w\!-\!\lambda)^{-1}
\,.}
\end{array}
\end{equation}
To conclude we observe that the RHS of \eqref{eq:app4} is zero.
Indeed,
since by assumption $a=q_{(\tilde j,\tilde k),(\tilde i,\tilde h)}\in\mf g_{\geq\frac12}$,
we have $ \delta_{(ip_1)(\tilde j,\tilde k)} = \delta_{(\tilde i,\tilde h)(j1)} =0$
for every $(\tilde i\tilde h),(\tilde j\tilde k)\in\mc J$.
\end{proof}


\begin{thebibliography}{00} 

\bibitem[Adl79]{Adl79}
Adler M.,
\emph{On a trace functional for formal pseudodifferential operators and the symplectic structure of the Korteweg-de Vries equation},
Invent. Math. {\bf50} (1979), 219-248.

\bibitem[BFOFW90]{BFOFW90}
Balog J., Feh\'er L., O'Raifeartaigh L., Forg\'acs P., Wipf A.,
\emph{Toda theory and W-algebra from a gauged WZNW point of view},
Ann. Physics {\bf 203} (1990), n.1, 76-136.

\bibitem[BDSK09]{BDSK09}
Barakat A., De Sole A., Kac V. G.,
\emph{Poisson vertex algebras in the theory of Hamiltonian equations}, 
Jpn. J. Math. {\bf 4} (2009), n.2, 141-252.

\bibitem[BG07]{BG07}
Brundan J., Goodwin S.M.,
\emph{Good grading polytopes},
Proc. Lond. Math. Soc. (3) {\bf 94} (2007), n.1, 155-180.

\bibitem[BdGHM93]{BdGHM93}
Burruoughs N., de Groot M., Hollowood T., Miramontes L.,
\emph{Generalized Drinfeld-Sokolov hierarchies II: the Hamiltonian structures},
Comm. Math. Phys. {\bf 153} (1993), 187-215.

\bibitem[Che92]{Che92}
Cheng Y.,
\emph{Constraints of the Kadomtsev-Petviashvili hierarchy},
J. Math. Phys. {\bf 33} (1992), n.11, 3774-3782.

\bibitem[dGHM92]{dGHM92}
de Groot M., Hollowood T., Miramontes L.,
\emph{Generalized Drinfeld-Sokolov hierarchies},
Comm. Math. Phys. {\bf 145} (1992), 57-84.

\bibitem[DF95]{DF95}
Delduc F., Feh\'er L.,
\emph{Regular conjugacy classes in the Weyl group and integrable hierarchies},
J. Phys. A {\bf 28} (1995), n.20, 5843-5882 .

\bibitem[DSK06]{DSK06}
De Sole A., Kac V.G.,
\emph{Finite vs affine W-algebras},
Jpn. J. Math. {\bf 1} (2006), n.1, 137-261. 

\bibitem[DSK13]{DSK13}
De Sole A., Kac V.G.,
\emph{Non-local Poisson structures and applications to the theory of integrable systems},
Jpn. J. Math. {\bf8} (2013), n.2, 233-347.

\bibitem[DSKV13]{DSKV13}
De Sole A., Kac V. G., Valeri D.,
\emph{Classical $\mc W$-algebras and generalized Drinfeld-Sokolov bi-Hamiltonian systems
within the theory of Poisson vertex algebras},
Comm. Math. Phys. {\bf 323} (2013), n.2, 663-711.

\bibitem[DSKV14a]{DSKV14a}
De Sole A., Kac V. G., Valeri D.,
\emph{Classical $\mc W$-algebras and generalized Drinfeld-Sokolov hierarchies
for minimal and short nilpotents},
Comm. Math. Phys. {\bf 331} (2014), n.2, 623-676.

\bibitem[DSKV14b]{DSKV14b}
De Sole A., Kac V. G., Valeri D.,
\emph{Dirac reduction for Poisson vertex algebras},
Comm. Math. Phys. {\bf 331} (2014), n.3, 1155-1190.

\bibitem[DSKV15]{DSKV15}
De Sole A., Kac V. G., Valeri D.,
\emph{Adler-Gelfand-Dickey approach to classical $\mc W$-algebras within the theory of Poisson vertex algebras},
Int.Math.Res.Not. 2015, n.21, 11186-11235.

\bibitem[DSKV15-cor]{DSKV15-cor}
De Sole A., Kac V. G., Valeri D.,
\emph{Erratum to: Classical $\mc W$-algebras and generalized Drinfeld-Sokolov hierarchies
for minimal and short nilpotents},
Comm. Math. Phys. {\bf 333} (2015), n.3, 1617-1619.

\bibitem[DSKV16]{DSKV16}
De Sole A., Kac V. G., Valeri D.,
\emph{Structure of classical (finite and affine) $\mc W$-algebras},
to appear in JEMS (2016), arXiv:1404.0715.

\bibitem[DSKVnew]{DSKVnew}
De Sole A., Kac V. G., Valeri D.,
\emph{A new scheme of integrability for (bi)-Hamiltonian PDE},
to appear in Comm. Math. Phys., arXiv:1508.02549.

\bibitem[DSKV17]{DSKV17}
De Sole A., Kac V. G., Valeri D.,
\emph{Finite $W$-algebras for $\mf{gl}_N$},
preprint arXiv:1605.02898.

\bibitem[Dic03]{Dic03}
Dickey L. A.,
\emph{Soliton equations and Hamiltonian systems},
Advanced series in mathematical physics, World scientific, Vol. 26, 2nd Ed., 2003.

\bibitem[DS85]{DS85}
Drinfeld V.G., Sokolov V.V.,
\emph{Lie algebras and equations of KdV type},
Soviet J. Math. {\bf 30} (1985), 1975-2036.

\bibitem[EK05]{EK05}
Elashvili A.G., Kac V.G., 
\emph{Classification of good gradings of simple Lie algebras},
Lie groups and invariant theory, 85-104, 
Amer. Math. Soc. Transl. Ser. 2, {\bf 213}, 
Amer. Math. Soc., Providence, RI, 2005.

\bibitem[FHM93]{FHM93}
Feh\'er L., Harnad J., Marshall I.,
\emph{Generalized Drinfeld-Sokolov reductions and KdV type hierarchies},
Comm. Math. Phys. {\bf 154} (1993), n.1, 181-214.

\bibitem[FGMS95]{FGMS95}
Fern\'andez-Pousa C., Gallas M., Miramontes L., S\'anchez Guill\'en J., 
\emph{$\mc W$-algebras from soliton equations and Heisenberg subalgebras}. 
Ann. Physics {\bf 243} (1995), n.2, 372-419.

\bibitem[FGMS96]{FGMS96}
Fern\'andez-Pousa C., Gallas M., Miramontes L., S\'anchez Guill\'en J., 
\emph{Integrable systems and $\mc W$-algebras}, 
VIII J. A. Swieca Summer School on Particles and Fields (Rio de Janeiro, 1995), 475-479.

\bibitem[GGRW05]{GGRW05}
Gelfand I.M., Gelfand S.I., Retakh V. and Wilson R.L.,
\emph{Quasideterminants}
Adv. Math. {\bf 193} (2005), n.1, 56-141.

\bibitem[GD76]{GD76}
Gelfand I. M., Dickey L. A.,
\emph{Fractional powers of operators and Hamiltonian systems},
Funct. Anal. Appl. {\bf 10} (1976), 259-273.

\bibitem[GD78]{GD78}
Gelfand I. M., Dickey L. A.,
\emph{A family of Hamiltonian structures connected with integrable nonlinear differential equations},
Akad. Nauk SSSR Inst. Prikl. Mat. Preprint N.136 (1978).

\bibitem[KSS91]{KSS91} 
Konopelchenko B., Sidorenko J., Strampp W.,
\emph{(1+1)-dimensional integrable 
systems as symmetry constraints of (2+1)-dimensional system},
Phys. Lett. A {\bf 157} (1991), 17-21.

\bibitem[KS92]{KS92} 
Konopelchenko B., Strampp W.,
\emph{New reductions of the Kadomtsev-Petviashvili and two-dimensional Toda lattice 
hierarchies via symmetry constraints},
J. Math. Phys. {\bf 33} (1992), n.11, 3676-3686. 

\bibitem[Ma81]{Ma81} Ma Y.-C.,
\emph{The resonant interaction among long and short waves},
Wave motion {\bf3} (1981), 257-267.

\bibitem[Mag78]{Mag78}
Magri F.,
\emph{A simple model of the integrable Hamiltonian equation},
J. Math. Phys. {\bf 19} (1978), n.5, 1156-1162.

\bibitem[MR15]{MR15}
Molev A.I., Ragoucy E., \emph{Classical W-algebras in types A, B, C, D and G},
Comm. Math. Phys. {\bf 336} (2015), n.2, 1053-1084.

\bibitem[SS93]{SS93}
Sidorenko J., Strampp W.,
\emph{Multicomponent integrable reductions in the Kadomtsev-Petviashvili hierarchy},
J. Math. Phys. {\bf34} (1993), n.4, 1429-1446. 

\bibitem[YO76]{YO76}
Yajima N., Oikawa M.,
\emph{Formation and interaction of sonic-Langmuir solitons--inverse scattering method},
Progr. Theoret. Phys. {\bf56} (1976), n.6, 1719-1739. 

\bibitem[ZC94]{ZC94}
Zhang Y.J., Cheng Y.,
\emph{Solutions for the vector k-constrained KP hierarchy},
J. Math. Phys. {\bf35} (1994), n.11, 5869-5884. 

\end{thebibliography}
\end{document}